\documentclass[11pt,letterpaper]{article}
\usepackage{anysize}
\usepackage{algorithm}
\usepackage[noend]{algpseudocode}
\usepackage{etoolbox}
\usepackage{enumitem}
\usepackage{algorithmicx}
\usepackage[margin=1in]{geometry}
\usepackage{tabularx}
\usepackage{graphicx}
\usepackage{authblk}
\usepackage{subfigure}
\usepackage{amsmath,amsthm,amssymb,amsfonts,thmtools}
\usepackage[usenames,dvipsnames]{xcolor}
\usepackage{latexsym,amssymb}
\usepackage[draft]{fixme}
\usepackage{mathtools}
\usepackage[colorlinks]{hyperref}
\usepackage[capitalize]{cleveref}
\hypersetup{citecolor=PineGreen}
\usepackage{tikz}
\usetikzlibrary{patterns,patterns.meta}
\usepackage{appendix}
\usepackage{dsfont}
\newcommand{\rsX}[2]{R^{X,#1}_{#2}}
\newcommand{\qsX}[2]{Q^{X,#1}_{#2}}
\newcommand{\rsY}[2]{R^{Y,#1}_{#2}}
\newcommand{\qsY}[2]{Q^{Y,#1}_{#2}}
\newcommand{\ri}[2]{R_{#2}^{#1}}
\newcommand{\qi}[2]{Q_{#2}^{#1}}
\newcommand{\lft}{\mathsf{left}}
\newcommand{\rgt}{\mathsf{right}}
\newcommand{\shf}{\mathsf{shift}}
\newtheorem*{theorem*}{Theorem}
\newtheorem*{lemma*}{Lemma}
\newtheorem*{corollary*}{Corollary}
\newtheorem*{proposition*}{Proposition}
{}

{}

\fxsetup{mode=multiuser,theme=color, layout=inline}
\FXRegisterAuthor{tk}{atk}{\color{blue} {\bf Tomasz}}
\FXRegisterAuthor{bs}{abs}{\color{red} {\bf Barna}}
\FXRegisterAuthor{am}{aam}{\color{purple} {\bf Anish}}

\newcommand*{\bdiv}{%
  \nonscript\mskip-\medmuskip\mkern5mu%
  \mathbin{\operator@font div}\penalty900\mkern5mu%
  \nonscript\mskip-\medmuskip
}

\usepackage{bm}
\def\range{\@ifstar\range@@\range@}
\def\range@#1#2{[\,#1\,.\,.\,#2\,]}
\def\range@@#1#2{\!\left[\,#1\,.\,.\,#2\,\right]\!}

\newcommand\ED{\operatorname{ED}}
\newcommand\TD{\operatorname{TD}}

\newcommand\Int{\mathbb Z}

\newcommand{\Oh}{\mathcal{O}}
\newcommand{\Ot}{\ensuremath{\widetilde{\Oh}}}

\newcommand{\dd}{.\,.}
\newcommand{\Exp}{\mathrm{Exp}}

\def\poly{\operatorname{poly}}
\def\polylog{\operatorname{polylog}}

\newcommand{\Zz}{\mathbb{Z}_{\ge 0}}
\newcommand{\lbl}{\mathsf{label}}
\newcommand{\Lbl}{\mathrm{Labels}}
\newcommand{\val}{\mathsf{val}}
\newcommand{\ins}{\mathsf{insert}}
\newcommand{\sub}{\mathsf{delete}}
\newcommand{\unlbl}{\mathsf{pos}}
\newcommand{\del}{\mathsf{delete}}

\renewcommand{\sub}{\subseteq}

\newcommand{\Tt}{\tilde{T}}

\newcommand{\Xd}{X_{\$}}

\newtheorem{claim}{Claim}[section]
\newtheorem{theorem}[claim]{Theorem}
\newtheorem{lemma}[claim]{Lemma}

\newtheorem{corollary}[claim]{Corollary}

\theoremstyle{definition}
\newtheorem{definition}[claim]{Definition}
\theoremstyle{remark}
\newtheorem{remark}[theorem]{Remark}

\setlist[enumerate]{nosep, topsep=1ex}
\setlist[itemize]{nosep, topsep=1ex}
\setlist[description]{nosep}
\allowdisplaybreaks
\title{Approximating Edit Distance in the Fully Dynamic Model}
\author[1]{Tomasz Kociumaka}
\author[2]{Anish Mukherjee\thanks{Research supported in part by the Centre for Discrete Mathematics and its Applications (DIMAP) and by EPSRC award EP/V01305X/1.}}
\author[3]{Barna Saha\thanks{Partially supported by NSF grants 1652303, 1909046, 2112533, and HDR TRIPODS Phase II grant 2217058.}}

\affil[1]{Max Planck Institute for Informatics, Saarland Informatics Campus, Saarbr\"ucken, Germany}
\affil[ ]{\url{tomasz.kociumaka@mpi-inf.mpg.de}}
\affil[2]{University of Warwick, Coventry, UK}
\affil[ ]{\url{anish.mukherjee@warwick.ac.uk}}
\affil[3]{University of California, San Diego, US}
\affil[ ]{\url{barnas@ucsd.edu}}

\date{}
\begin{document}

\maketitle
\setcounter{page}{0}
\thispagestyle{empty}
\begin{abstract}
    The edit distance is a fundamental measure of sequence similarity, defined as the minimum number of character insertions, deletions, and substitutions needed to transform one string into the other. Given two strings of length at most $n$, a simple dynamic programming computes their edit distance exactly in $\Oh(n^2)$ time, which is also the best possible (up to subpolynomial factors) assuming the Strong Exponential Time Hypothesis (SETH). The last few decades have seen tremendous progress in edit distance approximation, where the runtime has been brought down to subquadratic, to near-linear, and even to sublinear at the cost of approximation.

    In this paper, we study the \emph{dynamic} edit distance problem where the strings change dynamically as the characters are substituted, inserted, or deleted over time. Each change may happen at any location of either of the two strings.
    The goal is to maintain the (exact or approximate) edit distance of such dynamic strings while minimizing the update time. 
    The exact edit distance can be maintained in $\tilde{O}(n)$ time per update (Charalampopoulos, Kociumaka, Mozes; 2020), which is again tight assuming SETH. 
    Unfortunately, even with the unprecedented progress in edit distance approximation in the static setting, strikingly little is known regarding dynamic edit distance approximation.
    Utilizing the best near-linear-time (Andoni, Nosatzki; 2020) and sublinear-time (Goldenberg, Kociumaka, Krauthgamer, Saha; 2022) approximation algorithm, an old exact algorithm (Landau and Vishkin; 1988), and a generic dynamic strings implementation (Mehlhorn, Sundar, Uhrig; 1996), it is possible to achieve an $\Oh(n^{c})$-approximation in \smash{$n^{0.5-c+o(1)}$} update time for any constant \smash{$c\in [0,\frac16]$}. Improving upon this trade-off, characterized by the approximation-ratio and update-time product $n^{0.5+o(1)}$, remains wide open.

    The contribution of this work is a dynamic $n^{o(1)}$-approximation algorithm with amortized expected update time of $n^{o(1)}$. In other words, we bring the approximation-ratio and update-time product down to $n^{o(1)}$, which is also the best possible with the current state of the art in static algorithms. 
    Our solution utilizes an elegant framework of precision sampling tree for edit distance approximation (Andoni, Krauthgamer, Onak; 2010). We show how it is possible to dynamically maintain precision sampling trees which comes with significant nontriviality and can be an independent tool of interest for further development in dynamic string algorithms. 
\end{abstract}
\newpage
\section{Introduction}
Dynamic algorithms model the real-world scenario where data is evolving rapidly and solutions must be efficiently maintained upon every update. 
Owing to significant advances in recent decades, many fundamental problems in graphs and sequences are now well understood in the dynamic setting.
Examples include maximum matchings in graphs \cite{Sankowski07, GP13, BKS23}, connectivity \cite{NS17,ChuzhoyGLNPS20}, maximum flow \cite{ChenGHPS20, BLS23}, minimum spanning trees \cite{NSW17,ChuzhoyGLNPS20},  clustering \cite{CharikarCFM04, HK20, BateniEFHJMW23}, diameter estimation \cite{BN19}, independent set \cite{AOSS18}, pattern matching \cite{ABR00, GKKLS18, CKW20, CGKMU22}, lossless compression \cite{NIIBT20}, string similarity \cite{ABR00, CGP20}, longest increasing subsequence \cite{MS20, KS21, GJ21}, suffix arrays \cite{KK22}, and many others. 
Several related models with varying difficulties have been studied including incremental (insertion-only, where elements can only be added) \cite{BhattacharyaK20, GutenbergWW20, KMS22, BLS23}, decremental (deletion-only, where elements can only be deleted) \cite{Bernstein16, HKN16, BernsteinGW20, BernsteinGS20}, and fully-dynamic that supports both insertions and deletions \cite{DI04, GP13, RZ16, HHS22}.

In this paper, we study dynamic algorithms for computing edit distance, a fundamental measure of sequence similarity that, given two strings $X$ and $Y$, counts the minimum number of character insertions, deletions, and substitutions needed to convert $X$ into~$Y$. 
The edit distance computation is one of the cornerstone problems in computer science with wide applications. 
In the static setting, the complexity of edit distance problem is well understood: Under the Strong Exponential Time Hypothesis (SETH), there does not exist a truly subquadratic  algorithm for computing edit distance~\cite{ABW15,BK15,AHWW16,BI18}, whereas a textbook dynamic programming~\cite{Vin68,NW70,WF74,Sel74} solves the problem trivially in $\Oh(n^2)$ time for strings of lengths at most~$n$.
This has naturally fuels the quest of designing fast approximation algorithms that run in subquadratic, near-linear, or even sublinear time -- an area witnessing tremendous growth \cite{LV88,BEKMRRS03,BES06,AO09,AKO10,BEGHS18,CDGKS18,GKS19,GRS20,KS20b,BR20,AN20,KS20a,GKKS22,BCFN22a,BCFN22b}.
Starting with a $\sqrt{n}$-approximation in linear time \cite{LV88}, a sequence of works \cite{BEKMRRS03,BES06} have led to sub-polynomial~\cite{AO09} and polylogarithmic-factor~\cite{AKO10} approximations in near-linear time, constant-factor approximations in subquadratic time \cite{BEGHS18,CDGKS18,GRS20}, and recently culminated with a constant-factor approximation in near-linear time~\cite{AN20}. 

The development of dynamic edit distance algorithms started with easier variants of the problem.
Initially, most work has focused on updates restricted to the endpoints of the maintained strings~\cite{LMS98, KP04, IIST05,Tis08}. 
The most general of these contributions is Tiskin's~\cite{Tis08} linear-time algorithm for maintaining the edit distance subject to updates 
that can insert or delete a character at either endpoint of either of the two strings.
Much more recently, Charalampopoulos, Kociumaka, and Mozes~\cite{CKM20} gave an $\tilde{O}(n)$-update-time algorithm for the general dynamic edit distance problem, where updates are allowed anywhere within the maintained strings. Note SETH prohibits sublinear update time even in the most restrictive among the settings considered.

Achieving sublinear update time using approximation is therefore one of the tantalizing questions of dynamic edit distance. 
The unprecedented progress in static approximation algorithms implicitly leads to better dynamic algorithms.
In particular, a constant-factor approximation with $\Oh(n^{0.5+\epsilon})$ update time (for any constant $\epsilon > 0$) can be achieved as follows:
The algorithm maintains $\Oh(\log_C n)$ levels, each of which is responsible for distinguishing instances with edit distance at most $k$ from instances with distance at least $3 C k$ (for a constant $C \gg 1$ chosen based on $\epsilon$).
At each level, we test every $k$ updates whether the edit distance is below $2k$ or above $2kC$. 
If it is not above $2k C$, then it cannot get above $k+2kC\le 3k C$ for the next $k$ updates.
Similarly, if it is not below $2k$, then it cannot get below $2k-k=k$ for the next $k$ updates. 
If $k$ is small, we run the exact algorithm of \cite{LV88} using any dynamic strings data structure (e.g.,~\cite{m97,ABR00,GKKLS18})
that, at $\Ot(1)$ updated time, supports substring equality queries in $\Ot(1)$ time.
The cost of this implementation of the \cite{LV88} algorithm is $\Ot(k^2)$, which amortizes to $\Ot(k)$ per update. 
If $k$ is large, on the other hand, we use the approximation algorithm of~\cite{AN20}, which provides a $C$-approximation of edit distance in $\Oh(n^{1+\epsilon})$ time, amortizing to $\Oh(n^{1+\epsilon}/k)$ per update.
Choosing the best of the two approaches, we achieve $\Oh(n^{0.5+\epsilon})$ time for the worst-case level.
A similar strategy yields a more general trade-off: an $\Oh(n^c)$-approximation in $\Ot(n^{0.5-c})$ update time for any positive constant $c\le \frac16$.
The only difference is that we set $C=n^c$ (rather than a constant) and use the recent sublinear-time algorithm of~\cite{GKKS22} for large $k$. Its cost is $\smash{\Ot\big(\frac{nk^2}{(n^{c}k)^2}+\frac{n\sqrt{k}}{n^ck}\big)}$, which amortizes to $\smash{\Ot\big(\frac{n^{1-2c}}{k}+\frac{n^{1-c}}{k^{1.5}}\big)}$ per update.
Choosing the best of the two approaches, we achieve $\Ot\big(n^{0.5-c}+n^{0.4(1-c)}\big)$ time for the worst-case level, which simplifies to $\Ot(n^{0.5-c})$ provided that $c\le \tfrac16$.

Therefore, even with the recent progress in sublinear-time algorithms for edit distance estimation \cite{GKS19,KS20a,GKKS22,BCFN22a,BCFN22b}, improving the trade-off, characterized by the approximation-ratio and update time product $n^{0.5-o(1)}$ remains wide open.
In the recent STOC'22 workshop\footnote{\url{https://sites.google.com/view/stoc22-dynamic-workshop/}} on dynamic algorithms, designing better dynamic algorithms for approximating edit distance has been explicitly posted as an open question.

In this paper, we develop a dynamic $n^{o(1)}$-approximation algorithm with amortized expected update time of $n^{o(1)}$. In other words, we bring the approximation-ratio and update-time product down from $n^{0.5+o(1)}$ to $n^{o(1)}$, which is also the best possible with the current state of the art in static algorithms. We consider the most general fully-dynamic case which allows for all edits (insertions, deletions, and substitutions) at arbitrary locations within the input strings.

\begin{restatable*}{theorem}{thmmain}\label{thm:main}
    There exists a dynamic algorithm that, initialized with integer parameters $2 \le b \le n$,  maintains strings $X$ and $Y$ of lengths at most $n$
    subject to character edits and, upon each update, handled in $b^2 \cdot (\log n)^{\Oh(\log_b n)}$ amortized expected time, outputs an $\Oh(b \log_b n)$-factor approximation of the edit distance $\ED(X,Y)$. The answers are correct with high probability (at least $1-\frac{1}{n}$).
 \end{restatable*}

In particular, with \smash{$b=2^{\Theta(\sqrt{\log n \log \log n})}$}, our algorithm achieves a \smash{$2^{\Oh(\sqrt{\log n \log \log n})}$}-factor approximation with \smash{$2^{\Oh(\sqrt{\log n \log \log n})}$} update time.
Appropriate $b=\log^{\Theta(1/\epsilon)} n$, on the other hand, yields a $(\log{n})^{\Oh(1/\epsilon)}$-factor approximation with $\Oh(n^{\epsilon})$ update time for any constant $\epsilon > 0$.


At the heart of our solution is the static algorithm of Andoni, Krauthgamer, and Onak~\cite{AKO10}, who developed an elegant framework of precision sampling tree for edit distance approximation. One of our major contributions is to show how such a precision sampling tree can be maintained under arbitrary updates. This is especially challenging as a single insertion or deletion can affect the computation at all nodes of the precision sampling tree maintained in \cite{AKO10}. We believe dynamic precision tree maintenance can be an independent tool of interest for further development in dynamic sequence algorithms.
\subsection{Related Work}
Multiple works have considered dynamic data structures for various string problems. The goal here is to maintain one or more dynamic strings subject to various queries such as testing equality between strings~\cite{m97}, longest common prefix and lexicographic comparison queries \cite{m97,ABR00,GKKLS18}, pattern matching queries \cite{CKW20, CGKMU22}, text indexing queries \cite{GKKLS15, NIIBT20,KK22}, longest common substring \cite{CGP20}, longest increasing subsequence~\cite{CCP13,MS20, KS21, GJ21}, compressed representation of a text \cite{NIIBT20}, compact representation of repetitive fragments \cite{ABCK19}, and so on. 
Beyond character edits, some works on dynamic strings (including \cite{m97,ABR00,GKKLS18}) consider more general split and concatenate operations. These primitives can be used to implement not only edits but also cut-paste and copy-paste operations.

One of the major advancements in algorithms recently has been in the area of fully dynamic graph algorithms for many central graph problems such as spanning forest \cite{NS17}, minimum spanning forest \cite{ChuzhoyGLNPS20, NSW17}, transitive closure \cite{Sankowski04, BNS19}, strongly connected component \cite{KMS22}, shortest path \cite{DI04, BN19}, maximum matching \cite{Sankowski07, GP13}, maximal independent set \cite{AOSS18}, maximal matching \cite{BaswanaGS15, Solomon16}, set cover \cite{BHN19, AAGPS19}, vertex cover \cite{BK19} and various others. See the survey~\cite{HHS22} on recent advances in fully dynamic graph algorithms.     
\section{Preliminaries}\label{sec:prelims}

\paragraph{Notations}
The \emph{edit distance} $\ED(X, Y)$ of two strings $X$ and $Y$ is defined as the minimum number of character edits, i.e., \emph{insertions}, \emph{deletions}, and \emph{substitutions}, necessary to transform $X$ into $Y$. Throughout, we assume that the two strings are of length at most $n$, and we want to solve the $(k,K)$-gap edit distance problem, which asks to distinguish $\ED(X,Y)\le k$ from $\ED(X,Y)>K$. In the intermediate regime of $k < \ED(X,Y) \le K$, the algorithm is allowed to output any answer.

Denote by $|X|$ the length of $X$. The range $[i \dd j)$ denotes the set $\{i, i+1, \dots, j-1\}$. For an integer $i\in [0\dd |X|)$, let $X[i]$ denote the $i$-th character in $X$. For integers $0\le i \le j \le |X|$, denote by $X[i \dd j)$ the substring of $X[i]\cdot X[i+1]\cdots X[j-1]$.

Let $T$ be a rooted tree. A node of $T$ is called a \emph{leaf} if it has no children; otherwise, we call it an \emph{internal} node. The depth $d_v$ of a node $v$ of $T$ is the length (the number of edges) of the root-to-$v$ path, and the depth $d$ of the entire tree $T$ is the maximum depth $d_v$ of a node $v$ in~$T$.
We also denote the degree (the number of children) of a node $v$ by $b_v$ and the maximum degree by $b$.

\paragraph*{Precision Sampling Lemma.} 
The Precision Sampling Lemma is a generic technique to estimate a sum $\sum_i A_i$ from weak estimates of each $A_i$. 
Given a set of unknown numbers $A_1, \dots, A_n \in \mathbb{R}_{\ge 0}$,  suppose we have access to the values $\widetilde A_1, \dots, \widetilde A_n$ such that each $\widetilde A_i$ is an $(\alpha,\beta)$-approximation to $A_i$, that is,  $\frac1 \alpha A_i - \beta \leq \widetilde A_i \leq \alpha A_i + \beta$.
Then, $\sum_i \widetilde A_i$ is trivially an $(\alpha, n\beta)$-approximation to $\sum_i A_i$.
Andoni, Krauthgamer, and Onak~\cite{AKO10} gave the following elegant solution which shows that one can avoid the blow-up in the additive error with a good success probability if the estimates $\widetilde A_i$ instead have additive error $\beta \cdot u_i$ for a \emph{non-uniformly sampled} list of \emph{precisions} $u_1, \dots, u_n$.

\newcommand{\Dst}{\mathcal{D}}
\newcommand{\hDst}{\hat{\mathcal{D}}}
\begin{lemma}
[{\cite[Lemma 9]{BCFN22a}}]
\label{lem:precision-sampling}
Let $\epsilon, \delta > 0$. There is a distribution $\Dst(\epsilon, \delta)$ supported over the real interval $(0,1]$ with the following properties: 
\begin{description}
    \item[Accuracy:] Let $A_1, \dots, A_n \in \mathbb{R}_{\ge 0}$ and let $u_1, \dots, u_n \sim \Dst(\epsilon, \delta)$ be sampled independently. There is a recovery algorithm with the following guarantee: Given $(\alpha, \beta \cdot u_i)$-approximations~$\widetilde A_i$ of~$A_i$, the algorithm $((1+\epsilon) \cdot \alpha, \beta)$-approximates $\sum_i A_i$ in time $O(n \cdot \epsilon^{-2} \log(\delta^{-1}))$, with probability at least $1 - \delta$ and for any parameters $\alpha \geq 1$ and $\beta \geq 0$.
    \item[Efficiency:] Sample $u \sim \Dst(\epsilon, \delta)$. Then, for any $N \geq 1$, there is an event $E = E(u,N)$ happening with probability at least $1 - 1/N$, such that $\Exp[\,u^{-1} \mid E\,] \leq \Ot(\epsilon^{-2} \log(\delta^{-1}) \log N)$.
\end{description}
\end{lemma}
The lemma was first shown in~\cite{AKO10} and later refined and simplified in~\cite{AndoniKO11, Andoni17}. 
To streamline the analysis of our dynamic algorithm, we use the following simple corollary.
\begin{restatable}{corollary}{precisionsampling}
\label{cor:precision-sampling}
    Let $\epsilon, \delta > 0$. There is a distribution $\hDst(\epsilon, \delta)$ supported over the real interval $(0,1]$ with the following properties:
    \begin{description}
        \item[Accuracy:] Let $A_1, \dots, A_n \in \mathbb{R}_{\ge 0}$ and let $u_1, \dots, u_n \sim \hDst(\epsilon, \delta)$ be sampled independently. There is a recovery algorithm $\textsc{Recover}(\widetilde A_1,\ldots,\widetilde A_n,u_1,\ldots,u_n)$ with the following guarantee: Given $(\alpha, \beta \cdot u_i)$-approximations $\widetilde A_i$ of~$A_i$, the algorithm $((1+\epsilon) \cdot \alpha, \beta)$-approximates $\sum_i A_i$ in time $O(n \cdot \epsilon^{-2} \log(\delta^{-1}))$, with probability at least $1 - n\delta$ and for any parameters $\alpha \geq 1$ and $\beta \geq 0$.
        \item[Efficiency:] If $u \sim \hDst(\epsilon, \delta)$, then $\Exp[\,u^{-1}\,] \leq \Ot(\epsilon^{-2} \log^2(\delta^{-1}))$.
    \end{description}
\end{restatable}


\section{Technical Overview}\label{sec:tech-overview}
\subsection{Step 0: Static AKO Algorithm \& Modifications (\cref{sec:ako-static})}

The starting point of our dynamic algorithm is the Andoni--Krauthgamer--Onak algorithm \cite{AKO10} that,
given strings $X,Y$ of length at most $n$ and a parameter $b\in [2\dd n]$,
solves the $(k,K)$-gap edit distance problem for $K/k = \Oh(b \log_b n)$ in $n \cdot (\log n)^{\Oh(\log_b n)}$ time.
They introduce an auxiliary string similarity measure \emph{tree distance} $\TD(X,Y)$ which approximates the edit distance $\ED(X,Y)$ and is defined over an underlying rooted tree $T$ called the \emph{precision sampling tree}.

\begin{restatable*}[Precision Sampling Tree]{definition}{defpstree}\label{def:ps-tree}
    A rooted tree $T$ with $n$ leaves is called a \emph{precision sampling tree} if each node $v$ in $T$ is labeled with a non-empty interval $I_v$ such that 
\begin{itemize}
\item for the root node $v_{root}$, we have $I_{v_{root}} = [0\dd n)$;
\item for every internal node $v$ with children $v_0, \dots, v_{b_v-1}$, $I_v$ is the concatenation of $I_{v_0}, \dots, I_{v_{b_v-1}}$;
\item for every leaf node $v$, we have $|I_v|=1$.
\end{itemize}
\end{restatable*}

\begin{restatable*}[Tree Distance]{definition}{deftd}\label{def:tree-distance}
Let $X, Y$ be strings, and let $T$ be a precision sampling tree with $|X|$ leaves. For every node $v$ in $T$ and every shift $s \in \mathbb{Z}$:
\begin{itemize}
\item If $v$ is a leaf with $I_v = \{i\}$, then 
\[\TD_{v, s}(X, Y)=\begin{cases} \ED(X[i],Y[i+s]) & \text{if }i+s\in [0\dd |Y|),\\
    1 & \text{otherwise.}\end{cases}\]
\item If $v$ is an internal node with children $v_0, \dots, v_{b_v-1}$, then
\begin{equation} \label{eq:tree-distance}
    \TD_{v, s}(X, Y) = \sum_{h=0}^{b_v-1} \min_{s_h \in \Int} \,(\TD_{v_h, s_h}(X, Y) + 2 \cdot |s - s_h| ).
\end{equation}
\end{itemize}
We write $\TD_T(X, Y) = \TD_{v_{root}, 0}(X, Y)$.
\end{restatable*}

The \emph{precision sampling tree} splits the computation into independent subproblems. The main idea is to approximate the tree distances $\TD_{v,s}(X,Y)$ of a block in one string (represented by a node $v$ in the tree) with several shifts $s$ of its corresponding block in the other string, and then carefully combine the results. Figure~\ref{fig:tree-dist} gives an illustration of the definition. 

\begin{figure}[tb!]
    \centering
    \begin{tikzpicture}[scale=0.95]
    \draw[very thick] (0,0) -- (17, 0) (0,.7) -- (17,.7) (0,3) -- (17, 3) (0,3.7) -- (17,3.7);
    \filldraw[draw=blue, fill=green!50!black, fill opacity = 0.5, text opacity=1] (2, 0) rectangle node{$v_0$} (5, 0.7);
    \draw[green!50!black] (2, 0.7) -- (2.7, 3) (5, 0.7) -- (5.7, 3);
    \filldraw[draw=blue, fill=green!50!black, fill opacity = 0.5, text opacity=1] (2.7, 3) rectangle (5.7, 3.7);
    \filldraw[draw=blue, fill=green!50!black, fill opacity = 0.5, text opacity=1] (3, 4.4) rectangle (6, 5.1);

    \filldraw[draw=blue, fill=orange, fill opacity = 0.5, text opacity=1] (5, 0) rectangle node{$v_1$} (8, 0.7);
    \draw[orange] (5, 0.7) -- (5.5, 3) (8, 0.7) -- (8.5, 3);
    \filldraw[draw=blue, fill=orange, fill opacity = 0.5, text opacity=1] (5.5, 3) rectangle (8.5, 3.7);
    \filldraw[draw=blue, fill=orange, fill opacity = 0.5, text opacity=1] (5.25, 4.4) rectangle (8.25, 5.1);

    \filldraw[draw=blue, fill=green!50!black, fill opacity = 0.5, text opacity=1] (8, 0) rectangle node{$v_2$} (11, 0.7);
    \draw[green!50!black] (8, 0.7) -- (10.8, 3) (11, 0.7) -- (13.8, 3);
    \filldraw[draw=blue, fill=green!50!black, fill opacity = 0.5, text opacity=1] (10.8, 3) rectangle (13.8, 3.7);
    \filldraw[draw=blue, fill=green!50!black, fill opacity = 0.5, text opacity=1] (10.5, 4.4) rectangle (13.5, 5.1);

    \filldraw[draw=blue, fill=orange, fill opacity = 0.5, text opacity=1] (11, 0) rectangle node{$v_3$} (14, 0.7);
    \draw[orange, opacity=0.5] (11, 0.7) -- (13.6, 3) (14, 0.7) -- (16.6, 3);
    \filldraw[draw=blue, fill=orange, fill opacity = 0.5, text opacity=1] (13.6, 3) rectangle (16.6, 3.7);
    \filldraw[draw=blue, fill=orange, fill opacity = 0.5, text opacity=1] (13.5, 4.4) rectangle (16.5, 5.1);

    \foreach \x in {2,5,8,11,14}{
        \draw[densely dashed, thick] (\x, .7) -- (\x+2, 3);
        \draw[densely dashed, thick, blue] (\x+2, 3) -- (\x+2, 3.7);
    }
    \draw[blue] (2, 0.7) -- (2, 3);
    \draw[thick, draw=blue,text=black, latex-latex] (2.7, 3.2) -- node[above=-3]{$s{-}s_0$} (4, 3.2);
    \draw[thick, draw=blue,text=black, latex-latex] (5.5, 3.2) -- node[above=-3]{$s{-}s_1$} (7, 3.2);
    \draw[thick, draw=blue,text=black, latex-latex] (10.8, 3.2) -- node[above=-3]{$s_2{-}s$} (10, 3.2);
    \draw[thick, draw=blue,text=black, latex-latex] (13.6, 3.2) -- node[above=-3]{$s_3{-}s$} (13, 3.2);

    \draw[latex-latex, thick] (2, -.2) node[below]{$\ell\vphantom{r \ell}$} -- (14, -.2) node[below]{$r\vphantom{r \ell}$};
    \draw[latex-latex, thick] (4, 2.8) node[below]{$\ell{+}s\vphantom{r \ell}$} -- (16, 2.8) node[below]{$r{+}s\vphantom{r \ell}$};
    \draw[latex-latex, thick] (2, 2.8)  -- node[below=-2]{$s$} (4, 2.8);

    \foreach \x in {0,0.75,1.5,...,17} {
        \draw[blue] (\x, 3.7) -- (\x, 4.3);
    }
    \draw (-0.2, 0.35) node[right]{$X[\ell\dd r)$};
    \draw (-0.2, 3.35) node[right]{$Y[\ell+s\dd r+s)$};
    \draw (-0.2, 4.75) node[right]{\footnotesize \parbox{\widthof{only specific shifts}}{only specific shifts\\ are allowed}};

\end{tikzpicture}
    \caption{The definition of tree distance $\TD_{v,s}(X,Y)$ at a node $v$ with $I_v = [\ell\dd r)$, shift $s$, and children $v_0$, $v_1$, $v_2$, and $v_3$. The dashed lines show the shift given by $s$. The blue lines show the relative shift $|s-s_h|$ at each child. Shift sparsification is shown on the top, which only allows a subset of the shifts.}
    \label{fig:tree-dist}
\end{figure}
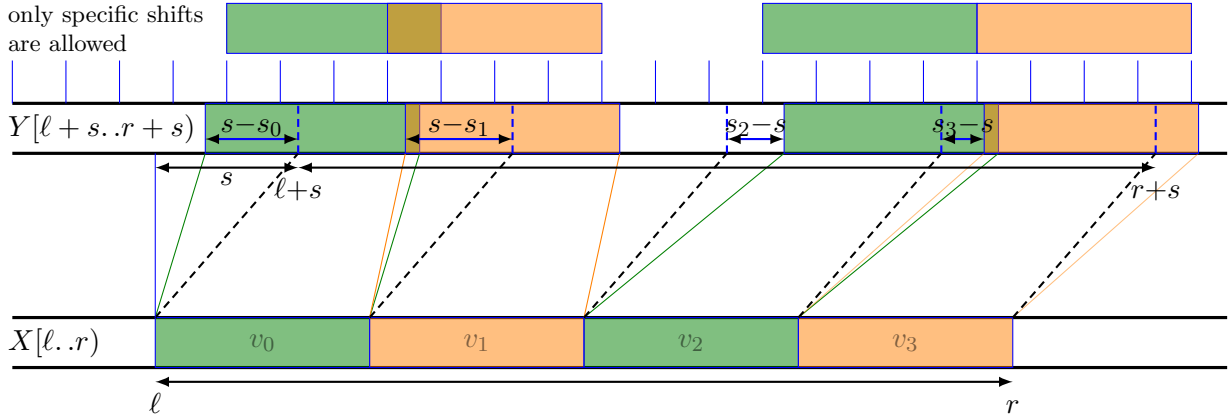

The following lemma establishes the relation (and the usefulness) of the tree distance measure with the edit distance. 

\begin{restatable*}[{\cite[Lemma 6]{BCFN22a}}]{lemma}{lemequivedtd}\label{lem:equiv-ed-td}
Let~$X, Y$ be strings and let~$T$ be a precision sampling tree with $|X|$ leaves, degree at most $b$, and depth at most $d$. 
Then, $\ED(X, Y) \leq \TD_T(X, Y) \leq 2 bd \cdot \ED(X, Y)$.
\end{restatable*}

In particular, if the precision sampling tree is balanced, its depth is bounded by $\Oh(\log_b n)$, so the above lemma shows that $\TD_T(X, Y)$ 
is an $\Oh(b\log_b n)$-approximation of $\ED(X,Y)$.

Since we are only interested in the $(k,K)$-gap edit distance problem for $K/k=\Theta(b \log_b n)$, we only need to consider the shifts $s$ in a range of $[-K \dd K]$ and also focus on computing the capped tree distances $\TD_{v,s}^{\le K}(X,Y) = \min(\TD_{v,s}(X,Y), K)$. Approximating the {tree distance} problem boils down to the following.

For each node $v$ of the precision sampling tree $T$, given appropriate parameters $\alpha_v\ge 1$ and $\beta_v\ge 0$,
for every shift $s\in [-K\dd K]$, compute an $(\alpha_v,\beta_v)$-approximation  $\Delta^{\le K}_{v,s}(X,Y)$ of $\TD^{\le K}_{v,s}(X,Y)$.
In other words, the value $\Delta^{\le K}_{v,s}(X,Y)$ needs to satisfy the following property:
\[
    \tfrac1{\alpha_v} \TD_{v,s}^{\le K}(X,Y) - \beta_v \leq \Delta^{\le K}_{v,s}(X,Y) \leq \alpha_v \TD^{\le K}_{v, s}(X, Y) + \beta_v.
\]

The AKO algorithm evaluates the following expression approximately (see \cref{def:tree-distance}):
\begin{equation}\label{eq:noshiftsparse}
    \TD^{\le K}_{v, s}(X, Y) = \min\Bigg(\sum_{h=0}^{b_v-1} A_{h, s},\,K\Bigg), \qquad A_{h, s} = \min_{s_h \in [-K \dd K]} \left(\TD^{\le K}_{v_h, s_h}(X, Y) + 2 \cdot |s - s_h|\right).
\end{equation}
The algorithm processes the precision sampling tree $T$ in a bottom-up fashion so that $(\alpha_{v_h}, \beta_{v_h})$-approximations $\Delta^{\le K}_{v_h,s}(X,Y)$ of the values $\TD^{\le K}_{v, s}(X, Y)$ are available when the algorithm processes~$v$.
These values are used to obtain \[\tilde{A}_{h,s}= \min_{s_h \in [-K \dd K]} \left(\Delta^{\le K}_{v_h, s_h}(X, Y) + 2 \cdot |s - s_h|\right).\]  
Hence, $\tilde{A}_{h,s}$ is also an $(\alpha_{v_h},\beta_{v,h})$-approximation of $A_{h,s}$.
By setting $\beta_{v_h}=\beta_{v}\cdot u_{v_h}$ and $\alpha_{v_h}=\alpha_v/(1+\epsilon)$, and using the Precision Sampling Lemma (\cref{lem:precision-sampling}) with appropriate parameters, they obtain $\Delta^{\le K}_{v,s}(X,Y)$ as an $(\alpha_v, \beta_v)$-approximation of $\TD^{\le K}_{v, s}(X, Y)$. In particular, this leads to an efficient computation of $(10, 0.001K)$-approximation $\Delta^{\le K}_{v_{root},0}(X,Y)$  of $\TD^{\le K}_{T}(X,Y)$.
To distinguish $\ED(X,Y)\le k$ from $\ED(X,Y)>K$, it suffices to check whether $\Delta^{\le K}_{v_{root},0}(X,Y)\le 0.099K$.

\paragraph{Modification to the Static Algorithm}
Unfortunately, the off-the-shelf AKO algorithm is not sufficient for us to build on it in the dynamic setting. First, we need to strengthen the definition of capped tree distance as follows.

\begin{restatable*}[Capped Tree Distance]{definition}{defcaptd}\label{def:captd}
    For strings $X,Y$, precision sampling tree $T$ with $|X|$ leaves, and an integer threshold $K \geq 0$, define \emph{$K$-capped tree distance} for each node $v \in T$, and shift $s \in [-K \dd K]$ as follows:
    \begin{itemize}
        \item If $v$ is a leaf with $I_v=\{i\}$, then $\TD^{\le K}_{v,s}(X, Y)=\min{(\ED(X[i], Y[i+s]),K-|s|)}$, assuming that $\ED(X[i],Y[i+s])=1$ if $i+s\notin [0\dd |Y|)$.
        \item If $v$ is an internal nodes with children $v_0,\ldots, v_{b_v-1}$, then 
        \begin{equation} \label{eq:capped-td}
        \TD^{\le K}_{v,s}(X, Y) = \min\left(\sum_{h=0}^{b_v-1} \,\min_{s_h \in [-K \dd K]}\, \left( \TD_{v_h, s_h}^{\le K}(X, Y) + 2 \cdot |s - s_h| \right),\, K-|s| \right).
    \end{equation}
    \end{itemize}
\end{restatable*}
By capping the distance to $K-|s|$ rather than $K$, we can prove a stronger relation between the capped tree distance and original tree distance at all internal nodes which previously was true only at the root and only for $s=0$.

\begin{restatable*}[Relationship between Capped and Original Tree Distance]{lemma}{lemcaptd}
    \label{lem:captd}
    For every node $v$ and shift $s\in [-K\dd K]$, we have 
    $\TD_{v,s}^{\le K}(X,Y) = \min(\TD_{v,s}(X,Y), K-|s|)$.
    In particular, $\TD^{\le K}_T(X,Y)= \TD_{v_{root},0}^{\le K}(X,Y) = \min(\TD_T(X,Y),K)$.
\end{restatable*}

Next, to reduce the number of computations per node while evaluating \cref{eq:noshiftsparse}, we consider only \emph{active} nodes, where a node $v$ is called active if $|I_v| > \beta_v$; otherwise, returning $\Delta^{\le K}_{v,s}(X,Y)=0$ for all shifts $s \in [-K \dd K]$ is a valid $(\alpha_v,\beta_v)$-approximation. Therefore, all \emph{inactive} nodes can be pruned from further computation. Moreover, on each active node, we only consider a subset $S_v \subseteq [-K \dd K]$ of shifts selected uniformly at random with rate $\min\left(1,\, 48\beta_{v}^{-1}(b_v+1)\ln{n}\right)$. 

By \cref{def:tree-distance}, any two shifts $s,s' \in [-K\dd K]$ satisfy
\begin{equation}
    \left|\TD_{v,s}^{\le K}(X,Y)-\TD_{v,s'}^{\le K}(X,Y)\right|\leq 2(b_{v_h}+1)|s-s'|
\end{equation}
Consequently, by definition of $S_{v}$, for any shift $s \in [-K \dd K]$, the closest shift 
$\tilde{s} \in S_{v}$ satisfies 
\[\Pr\left[|s-\widetilde s| > \frac{\beta_{v}}{8(b_{v}+1)}\right] \le \left(1-\frac{48(b_{v}+1)\ln n}{\beta_{v}}\right)^{\frac{\beta_{v}}{8(b_{v}+1)}} \le \exp\left(-\frac{48(b_{v}+1)\ln n \cdot \beta_{v}}{\beta_{v}\cdot 8(b_{v}+1)}\right)\le   \frac{1}{n^6}.\]

Utilizing the above, we can show the shift sparsification does not add too much additional error in the computation of $\Delta^{\le K}_{v,s}(X,Y)$. Shift sparsification has been considered before (at a much lower rate and deterministically) in \cite{BCFN22b}, whereas pruning was used in \cite{AKO10, BCFN22a}. 
However, in the static setting, simultaneous pruning and sparsification is not useful (see \cref{remark:staticvsdyn}). 
On the other hand, it is crucial that we utilize both in the dynamic setting. 
Therefore, our analysis in the static setting (building up to the dynamic setting) needs to be different. 
The running time bound follows from \cref{lem:exp-precision}, which proves that  \[\Exp[\beta_v^{-1}] = \tfrac1K (\log n)^{\Oh(d_v)} = \tfrac1K \cdot (\log n)^{\Oh(\log_b n )}.\]
The time taken by \cref{alg:ako} for execution at node $v$ is bounded by $\Ot(\sum_{h=0}^{b_v-1}(|S_v|+|S_{v_h}|))$ (\cref{lem:ako-time}).
 Since $\Exp[|S_w|]\le 48K (b_w+1) \ln n\cdot \Exp[\beta_w^{-1}]=b \cdot (\log n)^{\Oh(\log_b n)}$ holds by \cref{lem:exp-precision} for every node $w$, this is $b^2\cdot (\log n)^{\Oh(\log_b n)}$ per node. Across all nodes, the expected running time of \cref{alg:ako} is $nb^2 \cdot (\log n)^{\Oh(\log_b n)}$.
Moreover, the approximation guarantees hold with high probability ($1-1/\poly(n)$).

\subsection{Step 1: Substitutions in both X and Y (\cref{sec:subXsubY})}
With the modifications in the static algorithm in place, we now consider a simple dynamic scenario for which already no better result exists in the literature. 
Namely, we allow for updates to be substitutions anywhere within the maintained strings $X$ and $Y$. 
We can assume that a precision sampling tree $T(X,Y)$ has already been built, and the goal is to maintain it so that, for all $v$ and $s \in S_v$, the value $\Delta^{\le K}_{v,s}(X,Y)$ remains a valid $(\alpha_v, \beta_v)$-approximation of $\TD_{v,s}^{\le K}(X,Y)$ at all times.

An advantage of allowing only substitutions is that structurally the precision sampling tree does not change, that is, at every node $v$, the interval $I_v$ remains the same.
Upon a substitution, the algorithm recomputes all values $\Delta^{\le K}_{v,s}(X,Y)$ for which $\TD^{\le K}_{v,s}(X,Y)$ might have potentially changed.
Note that considering all shifts $s \in [-K\dd K]$ even on a single node is not possible as $K$ can be large. 
Thus, shift sparsification is necessary.
On the other hand, considering even a single shift per node for which the $\TD^{\le K}_{v,s}(X,Y)$ value can be affected could be costly for some levels of the precision sampling tree (e.g., the leaf level). We show the number of active nodes on those levels is small and, overall, few values $\Delta^{\le K}_{v,s}(X,Y)$ need updating.

The following notion of \emph{relevance} characterizes the shifts for which $\TD^{\le K}_{v,s}(X,Y)$ might have changed: 
We will show that a substitution can affect $\TD^{\le K}_{v,s}(X,Y)$ only if $s$ is relevant.
If $v$ is not a leaf, then, in order to recompute $\Delta^{\le K}_{v,s}(X,Y)$, we will need to access some values $\Delta_{v_h,s_h}^{\le K}(X,Y)$
at the children $v_h$ of $v$. To capture the scope of the necessary shifts $s_h$, we introduce the notion of \emph{quasi-relevance}.

\begin{restatable*}{definition}{defrelevance}\label{def:relevant}
    Consider a node $v\in T$. We define the sets $\rsX{i}{v},\rsY{i}{v}\sub [-K\dd K]$ of shifts \emph{relevant} for substitutions at $X[i]$ and $Y[i]$, respectively, as follows, denoting $K_v := \min(K, |I_v|)$:
    \begin{align*}
        \rsX{i}{v} &= \{s\in [-K\dd K] : i \in I_v\},\\
        \rsY{i}{v} &= \{s\in [-K\dd K] : [i-s-K_v\dd i-s+K_v]\cap I_v \ne \emptyset\}.
    \end{align*}
    Moreover, we say that a shift  is \emph{quasi-relevant} with respect to a substitution if it is at distance at most $K_v$ from a relevant shift.
    Formally, \begin{align*}
        \qsX{i}{v} &= \{s\in [-K\dd K] : \exists_{s'\in \rsX{i}{v}}\; |s-s'|\le K_v\},\\
        \qsY{i}{v} &= \{s\in [-K\dd K] : \exists_{s'\in \rsY{i}{v}}\; |s-s'|\le K_v\}.
    \end{align*}
    Furthermore, we say that a node $v$ is \emph{in scope} of a substitution if the corresponding set of relevant shifts is non-empty.
\end{restatable*}

\begin{figure}[t]
    \centering
    \begin{tikzpicture}[scale=0.8]
    \draw[fill=black!50] (1, 3) rectangle (8, 3.5);
    \draw(4.5, 3) -- (1.5, 2.5) (4.5, 3) -- (4.5, 2.5) (4.5,3) -- (7.5, 2.5);

    \foreach \x in {0,1,2}{
        \draw[fill=orange!50] (1/3+3*\x, 2) rectangle (8/3+3*\x, 2.5);
        \draw(1.5+3*\x, 2) -- (0.5+3*\x, 1.5) (1.5+3*\x, 2) -- (1.5+3*\x, 1.5) (1.5+3*\x,2) -- (2.5+3*\x, 1.5);
    }
    \foreach \x in {1}{
        \draw[fill=black!50] (1/3+3*\x, 2) rectangle (8/3+3*\x, 2.5);
    }
    \foreach \x in {0,1,...,8}{
        \draw (1/9+\x, 1) rectangle (8/9+\x, 1.5);
        \draw(0.5+\x, 1) -- (1/6+\x, 0.5) (0.5+\x, 1) -- (0.5+\x, 0.5) (0.5+\x,1) -- (5/6+\x, 0.5);
    }
    \foreach \x in {3,5}{
        \draw[fill=orange!50] (1/9+\x, 1) rectangle (8/9+\x, 1.5);
    }
    \foreach \x in {4}{
        \draw[fill=black!50] (1/9+\x, 1) rectangle (8/9+\x, 1.5);
    }
    \foreach \x in {0,1,...,26}{
        \draw (1/27+\x/3, 0) rectangle (8/27+\x/3, 0.5);
        \draw(1/6+\x/3, 0) -- (1/18+\x/3, -0.5) (1/6+\x/3, 0) -- (1/6+\x/3, -0.5) (1/6+\x/3,0) -- (5/18+\x/3, -0.5);
    }
    \foreach \x in {13,14}{
        \draw[fill=orange!50] (1/27+\x/3, 0) rectangle (8/27+\x/3, 0.5);
    }
    \foreach \x in {12}{
        \draw[fill=black!50] (1/27+\x/3, 0) rectangle (8/27+\x/3, 0.5);
    }

    \foreach \x in {0,1,...,80}{
        \draw (1/81+\x/9, -1) rectangle (8/81+\x/9, -0.5);
    }
    \foreach \x in {36,38}{
        \draw[fill=orange!50] (1/81+\x/9, -1) rectangle (8/81+\x/9, -0.5);
    }
    \foreach \x in {37}{
        \draw[fill=black!50] (1/81+\x/9, -1) rectangle (8/81+\x/9, -0.5);
    }

    \draw[right] (0, 3.25) node{$X$};
    \draw (4.5/81+37/9, 3.5) node[above]{$X[i]$};

    \filldraw[yellow, fill opacity=0.3] (1/81+37/9,3.5) rectangle (8/81+37/9, -1);

    \begin{scope}[xshift=11cm]

        \draw (1, 4) rectangle (8, 4.5);

        \draw[fill=black!50] (1, 3) rectangle (8, 3.5);
        \draw(4.5, 3) -- (1.5, 2.5) (4.5, 3) -- (4.5, 2.5) (4.5,3) -- (7.5, 2.5);
    
        \foreach \x in {0,1,2}{
            \draw[fill=orange!50] (1/3+3*\x, 2) rectangle (8/3+3*\x, 2.5);
            \draw(1.5+3*\x, 2) -- (0.5+3*\x, 1.5) (1.5+3*\x, 2) -- (1.5+3*\x, 1.5) (1.5+3*\x,2) -- (2.5+3*\x, 1.5);
        }
        \foreach \x in {1}{
            \draw[fill=black!50] (1/3+3*\x, 2) rectangle (8/3+3*\x, 2.5);
        }
        \foreach \x in {0,1,...,8}{
            \draw (1/9+\x, 1) rectangle (8/9+\x, 1.5);
            \draw(0.5+\x, 1) -- (1/6+\x, 0.5) (0.5+\x, 1) -- (0.5+\x, 0.5) (0.5+\x,1) -- (5/6+\x, 0.5);
        }
        \foreach \x in {3,4,5}{
            \draw[fill=black!50] (1/9+\x, 1) rectangle (8/9+\x, 1.5);
        }
        \foreach \x in {0,1,...,26}{
            \draw (1/27+\x/3, 0) rectangle (8/27+\x/3, 0.5);
            \draw(1/6+\x/3, 0) -- (1/18+\x/3, -0.5) (1/6+\x/3, 0) -- (1/6+\x/3, -0.5) (1/6+\x/3,0) -- (5/18+\x/3, -0.5);
        }
        \foreach \x in {16,17}{
            \draw[fill=orange!50] (1/27+\x/3, 0) rectangle (8/27+\x/3, 0.5);
        }
        \foreach \x in {9,10,...,15}{
            \draw[fill=black!50] (1/27+\x/3, 0) rectangle (8/27+\x/3, 0.5);
        }
    
        \foreach \x in {0,1,...,80}{
            \draw (1/81+\x/9, -1) rectangle (8/81+\x/9, -0.5);
        }
        \foreach \x in {27,47}{
            \draw[fill=orange!50] (1/81+\x/9, -1) rectangle (8/81+\x/9, -0.5);
        }
        \foreach \x in {28,29,...,46}{
            \draw[fill=black!50] (1/81+\x/9, -1) rectangle (8/81+\x/9, -0.5);
        }
    
        \draw[left] (9, 3.25) node{$X$};
        \draw[left] (9, 4.25) node{$Y$};

        \draw (4.5/81+37/9, 4.5) node[above]{$Y[i]$};
    
        \filldraw[yellow, fill opacity=0.3] (1/81+37/9,4) -- (1/81+37/9,4.5) -- (8/81+37/9, 4.5) -- (8/81+37/9, 4) -- (8/81+37/9+1, 3.5) -- (8/81+37/9+1, -1) --(1/81+37/9-1, -1) -- (1/81+37/9-1, 3.5) -- cycle;

        \draw[thick, latex-latex] (8/81+37/9+1, -1.2) --node[below]{$\Oh(K)$} (1/81+37/9-1, -1.2);
    
    \end{scope}

\end{tikzpicture}
    \caption{(Left) Substitution in $X$ at $X[i]$. The relevant nodes are shown in dark forming a root-to-leaf path. (Right) Substitution in $Y$ at $Y[i]$. A node $v$ is relevant if there exists $s\in [-K\dd K]$ such that $I_v \cap [i-s-K_v\dd i-s+K_v]\ne \emptyset$. The relevant nodes are shown in dark within a contiguous yellow band of width $\Oh(k)$. Computation at relevant nodes may need to access values stored at orange nodes. The shifts at orange nodes that are needed for computation at relevant nodes are called quasi-relevant shifts.}
    \label{fig:2}
\end{figure}
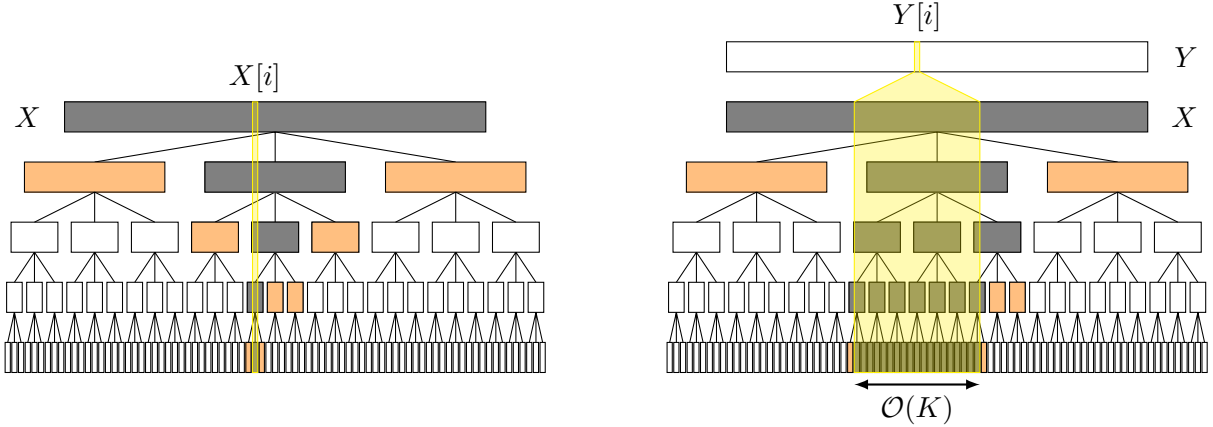

Figure~\ref{fig:2} illustrates the nodes in scope. For a substitution in $X$, they form a single root-to-leaf path, whereas for a substitution in $Y$ using \cref{def:relevant}, either the number of nodes in scope is $\Oh(1)$ per level, or the total interval spanned by such nodes per level is $\Oh(K)$. The orange nodes shown in Figure~\ref{fig:2} are children of nodes in scope. They are not in scope themselves, but we may need to access need to access $\Delta_{v',s'}^{\le K}(X,Y)$ stored at such a node $v'$ (colored orange) for shifts $s'$ if $s'$ is a quasi-relevant shift at its parent. 

The following locality lemma shows that $\TD_{v,s}^{\le K}(X,Y)$ changes only if $s$ is a relevant shift.

\begin{restatable*}{lemma}{lemlocality}\label{lem:locality2}
    Consider a node $v$ with $I_v = [i_v\dd j_v]$ of a precision sampling tree $T$ with $n$ leaves.
    Moreover, consider strings $X,X'\in \Sigma^n$ and $Y,Y'\in \Sigma^*$,
    as well as a shift $s\in [-K\dd K]$ such that $\TD^{\le K}_{v,s}(X,Y) = t$.
    If $X[i_v\dd j_v]=X'[i_v\dd j_v]$ and $Y[i_v+s-t\dd j_v+s+t]=Y'[i_v+s-t\dd j_v+s+t]$,
    then $\TD^{\le K}_{v,s}(X',Y')=\TD^{\le K}_{v,s}(X,Y)$.
\end{restatable*}

\begin{restatable*}{corollary}{corlocality}\label{cor:locality}
Consider a node $v\in T$ and a shift $s\in [-K\dd K]$.
If $s\notin \rsY{i}{v}$ is not relevant for a substitution at $Y[i]$, then the value $\TD^{\le K}_{v,s}(X,Y)$ is not affected by this substitution. Similarly, if $s\notin \rsX{i}{v}$ is not relevant for a substitution at $X[i]$, then the value $\TD^{\le K}_{v,s}(X,Y)$ is not affected by this substitution. 
\end{restatable*}

The update time bound is established by (i) showing that the total number of quasi-relevant shifts is $\Oh(K\cdot \log_b n)$ (\cref{lem:subtree}), and (ii) bounding the number of active nodes in scope (\cref{lem:active}). The probability of $v$ being active does not exceed $\frac{|I_v|}{K}\cdot (\log{n})^{\Oh(\log_b{n})} \leq 
\frac{K_v}{K}\cdot (\log{n})^{\Oh(\log_b{n})}$ (\cref{eq:active}).
Note that each node $v$ in scope of a substitution at $Y[i]$ contributes at least $|\qsY{i}{v}|\ge K_v$ quasi-relevant shifts (from the definition of quasi-relevance).
Consequently, for such node, the probability of being active does not exceed $\frac{|\qsY{i}{v}|}{K}\cdot (\log{n})^{\Oh(\log_b{n})}$. Therefore, using \cref{lem:subtree}, the number of active nodes in scope is bounded by $(\log{n})^{\Oh(\log_b{n})}$.
The running time of our procedure handling a substitution at $Y[i]$ can be shown to be $\Ot(\sum_{h=0}^{b_v-1}(|\rsY{i}{v}\cap S_v|+|\qsY{i}{v}\cap S_{v_h}|))$,
which is $b^2\cdot \frac{|\qsY{i}{v}|}{K}\cdot (\log{n})^{\Oh(\log_b{n})}$ per node $v$ in scope (\cref{lem:subY}). Now utilizing the bound on quasi-relevant shifts (\cref{lem:subtree}), and the number of active nodes in scope (\cref{lem:active}), we get the expected running time to be $b^2(\log{n})^{\Oh(\log_b{n})}$ in total.

\subsection{Step 2: Substitutions in X and Edits in Y (\cref{sec:subXedY})}

We now move to the considerably more challenging case of handling insertions and deletions. Note that, as long as the updates in $X$ are only substitutions, the structure of the precision sampling tree does not change following the discussion of the preceding section. Therefore, while we allow all possible edits (insertions, deletions, and substitutions) in $Y$, we still restrict to only substitutions in $X$. Even then, as Figure~\ref{fig:3} illustrates, a single deletion (or insertion) in $Y$ can affect all the nodes in $T(X,Y)$. Therefore, we would need to update the $\Delta^{\le K}_{v,s}(X,Y)$ values for all $v \in T(X,Y)$ and $s \in S_v$ (recall that $S_v$ are the allowed shifts at $v$ after sparsification), which is equivalent to building the precision sampling tree from scratch. In light of that, it seems almost hopeless that a precision sampling tree $T(X,Y)$ can be updated efficiently when edits are allowed even in one string.

\begin{figure}[H]
    \centering
    \begin{tikzpicture}[scale=0.8]
    \foreach \x in {0,0.4,...,17}{
        \fill[black!30] (\x,3) rectangle (\x+0.2,3.7);
    }

    \fill[red!30] (0,3) rectangle (0.2,3.7);
    \draw[red](0.1, 3.35) node{$\times$};

    \foreach \x in {0,0.4,...,16.6}{
        \fill[black!30] (\x+.2,3.9) rectangle (\x+0.4,4.6);
    }
    \draw[very thick] (0,0) -- (17, 0) (0,.7) -- (17,.7) (0,3) -- (17, 3) (0,3.7) -- (17,3.7) (0, 3.9) -- (16.8, 3.9) (0, 4.6) -- (16.8, 4.6);
    \filldraw[draw=blue, fill=green!50!black, fill opacity = 0.5, text opacity=1] (2, 0) rectangle node{$v_0$} (5, 0.7);
    \filldraw[draw=blue, fill=orange, fill opacity = 0.5, text opacity=1] (5, 0) rectangle node{$v_1$} (8, 0.7);
    \filldraw[draw=blue, fill=green!50!black, fill opacity = 0.5, text opacity=1] (8, 0) rectangle node{$v_2$} (11, 0.7);
    \filldraw[draw=blue, fill=orange, fill opacity = 0.5, text opacity=1] (11, 0) rectangle node{$v_3$} (14, 0.7);

    \foreach \x in {2,5,8,11,14}{
        \draw[densely dashed, very thick] (\x, .7) -- (\x+2, 3);
    }
    \draw[densely dotted, very thick] (0.2, 3.7) -- (0, 3.9);
    \draw[densely dotted, very thick] (17, 3.7) -- (16.8, 3.9);

    \filldraw[draw=blue, fill=green!50!black, fill opacity = 0.25, text opacity=1] (4, 3) rectangle (7, 3.7) (4, 3.9) rectangle (7, 4.6);
    \filldraw[draw=blue, fill=orange, fill opacity = 0.25, text opacity=1] (7, 3) rectangle (10, 3.7) (7, 3.9) rectangle (10, 4.6);
    \filldraw[draw=blue, fill=green!50!black, fill opacity = 0.25, text opacity=1] (10, 3) rectangle (13, 3.7) (10, 3.9) rectangle (13, 4.6);
    \filldraw[draw=blue, fill=orange, fill opacity = 0.25, text opacity=1] (13, 3) rectangle (16, 3.7) (13, 3.9) rectangle (16, 4.6);

    \draw[right] (0, 2.5) node{\footnotesize \parbox{\widthof{$Y[0]$ gets}}{$Y[0]$ gets\\deleted}};

    \draw[right] (17, 0.35) node{$X$};
    \draw[right] (17, 3.35) node{$Y$};
    \draw[right] (17, 4.25) node{$Y$ \footnotesize{after deletion}};

\end{tikzpicture}
    \caption{A single insertion or deletion can affect all the nodes in the precision sampling tree. For example, when $Y[0]$ is deleted, then $\Delta^{\le K}_{v,s}(X,Y)$ may need to be updated for all nodes. In the figure, each of $v_0$, $v_1$, $v_2$, and $v_3$ is being compared with different substrings of $Y$ after the deletion.}
    \label{fig:3}
\end{figure}
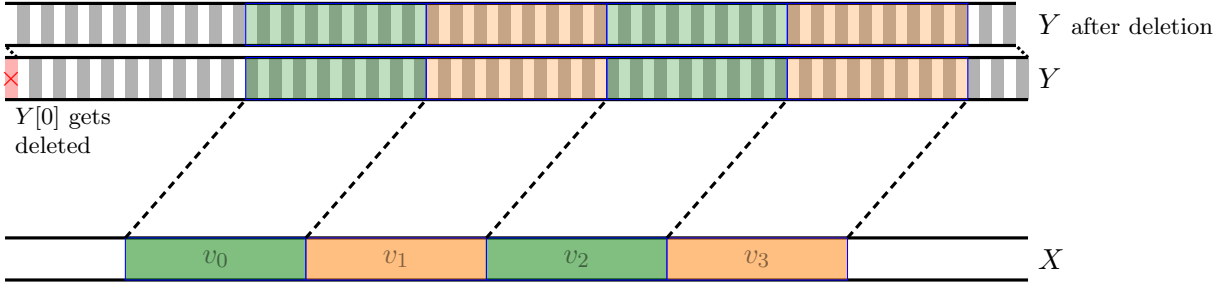
We bring in several major ideas here. Suppose that $\ED(X,Y) \leq k$.
Removing at most $k$ characters corresponding to the edits according to the optimal alignment, we can decompose $X$ and $Y$ into at most $k+1$ disjoint parts that are matched perfectly. We use this observation to solve the following {\sc Dynamic Pattern Matching} problem building upon earlier work on dynamic internal pattern matching~\cite{GKKLS18,CKW20} and dynamic suffix arrays~\cite{KK22}.

\begin{restatable*}[{\sc Dynamic Pattern Matching}]{lemma}{lemipm}\label{lem:ipm}
    There is a dynamic algorithm that maintains two strings $X$ and $Y$ subject to edits and, given a fragment $Y[i\dd j)$ and an integer $k\ge 0$,
    either:
    \begin{enumerate}
      \item returns a shift $s\in [-2k\dd 2k]$ such that $Y[i\dd j)=X[i+s\dd j+s)$, or
      \item returns NO, indicating that there is no shift $s\in [-k\dd k]$ such that $Y[i\dd j)=X[i+s\dd j+s)$.
    \end{enumerate}
    The algorithm supports updates in $\Ot(1)$ amortized time and queries in $\Ot(1)$ time.
\end{restatable*}

We use \cref{lem:ipm} iteratively to decompose $Y$ into $2k+1$ parts satisfying the following property.

\begin{restatable*}[{\sc Dynamic Decomposition}]{proposition}{prpipm}\label{prp:ipm}
    There is a dynamic algorithm that maintains two strings $X$ and $Y$ subject to edits and, upon a query, given an integer $k\in [0\dd |Y|]$,
    either reports that $\ED(X,Y)>k$
    or returns a partition of $Y=\bigodot_{i=0}^{2k} Y[y_i\dd y_{i+1})$ and a sequence $(s_0,s_2,\ldots,s_{2k})\in [-2k\dd 2k]^{k+1}$ such that $Y[y_i\dd y_{i+1})=X[y_i+s_i\dd y_{i+1}+s_i)$ if $i$ is even and $|Y[y_i\dd y_{i+1})|=1$ if $i$ is odd.
    The algorithm supports updates in $\Ot(1)$ amortized time and queries in $\Ot(k)$ time.
  \end{restatable*}

\begin{figure}[h]
    \centering
    \begin{minipage}{.48\textwidth}
        \centering
        \begin{tikzpicture}[scale=.9]
    \draw[white] (0,-3) -- (8, 3);
    \draw (0,0) rectangle (8, 0.5);
    \draw (5, 1) node {$Y$};

    \foreach [count=\i from 0] \x in {0,2,2.3,3.7,4,6.5,6.8,8}{
        \draw (\x,0.4) node[above]{\footnotesize $y_\i$};
    }
    \foreach \x in {2,3.7,6.5}{
        \draw[fill=orange] (\x, 0) rectangle (\x+.3, 0.5);
    }
    \fill[black, opacity=0.25] (0,0) rectangle (2,0.5) (2.3,0) rectangle (3.7,0.5) (4,0) rectangle (6.5,0.5) (6.8,0) rectangle (8,0.5);
    \draw[pattern={north west lines}] (0,0) rectangle (2,0.5);
    \draw[pattern={horizontal lines}] (2.3,0) rectangle (3.7,0.5);
    \draw[pattern={north east lines}] (4,0) rectangle (6.5,0.5);
    \draw[pattern={vertical lines}] (6.8,0) rectangle (8,0.5);

    \draw (0, -1) rectangle (8, -1.5);

    \fill[black, opacity=0.25] (0.5,-1) rectangle (2.5,-1.5);
    \fill[black, opacity=0.25] (1.8,-1) rectangle (3.2,-1.5);
    \fill[black, opacity=0.25] (4.1,-1) rectangle (6.6,-1.5);
    \fill[black, opacity=0.25] (6.5,-1) rectangle (7.7,-1.5);  
    \draw[pattern={north west lines}] (0.5,-1) rectangle (2.5,-1.5);
    \draw[pattern={horizontal lines},pattern color=black] (1.8,-1) rectangle (3.2,-1.5);
    \draw[pattern={north east lines}] (4.1,-1) rectangle (6.6,-1.5);
    \draw[pattern={vertical lines}] (6.5,-1) rectangle (7.7,-1.5);

    \foreach \x/\xx in {0/0.5, 2/2.5, 2.3/1.8, 3.7/3.2, 4/4.1,6.5/6.6, 6.8/6.5, 8/7.7}{
        \draw[thick, densely dotted] (\x, 0) -- (\xx, -1);
    }
\end{tikzpicture}
        \caption{Decomposition of $Y$ into $2k+1$ fragments such that if $i$ is even, then $Y[y_i\dd y_{i+1})=X[y_i+s_i\dd y_{i+1}+s_i)$ for some $s_i\in [-2K\dd 2K]$. Such a decomposition can be maintained in a ``lazy'' fashion.}\label{fig:4a}
    \end{minipage}\hfill
    \begin{minipage}{.48\textwidth}
        \begin{tikzpicture}[scale=0.9]
    \draw (0,0) rectangle (8, 0.5);
    \draw (5, 1) node {$Y$};
    \foreach [count=\i from 0] \x in {0,2,2.3,3.7,4,6.5,6.8,8}{
        \draw (\x,0.4) node[above]{\footnotesize $y_\i$};
    }
    \foreach \x in {2,3.7,6.5}{
        \draw[fill=orange] (\x, 0) rectangle (\x+.3, 0.5);
    }
    \fill[black, opacity=0.25] (0,0) rectangle (2,0.5) (2.3,0) rectangle (3.7,0.5) (4,0) rectangle (6.5,0.5) (6.8,0) rectangle (8,0.5);
    \draw[pattern={north west lines}] (0,0) rectangle (2,0.5);
    \draw[pattern={horizontal lines}] (2.3,0) rectangle (3.7,0.5);
    \draw[pattern={north east lines}] (4,0) rectangle (6.5,0.5);
    \draw[pattern={vertical lines}] (6.8,0) rectangle (8,0.5);

    \begin{scope}[yshift=-3.25cm,xshift=.5cm,scale=0.8]

        \draw (1, 3) rectangle (8, 3.5);
        \draw(4.5, 3) -- (1.5, 2.5) (4.5, 3) -- (4.5, 2.5) (4.5,3) -- (7.5, 2.5);
    
        \foreach \x in {0,1,2}{
            \draw (1/3+3*\x, 2) rectangle (8/3+3*\x, 2.5);
            \draw(1.5+3*\x, 2) -- (0.5+3*\x, 1.5) (1.5+3*\x, 2) -- (1.5+3*\x, 1.5) (1.5+3*\x,2) -- (2.5+3*\x, 1.5);
        }
        \foreach \x in {0,1,...,8}{
            \draw (1/9+\x, 1) rectangle (8/9+\x, 1.5);
            \draw(0.5+\x, 1) -- (1/6+\x, 0.5) (0.5+\x, 1) -- (0.5+\x, 0.5) (0.5+\x,1) -- (5/6+\x, 0.5);
        }
        \foreach \x in {0,1,...,26}{
            \draw (1/27+\x/3, 0) rectangle (8/27+\x/3, 0.5);
            \draw(1/6+\x/3, 0) -- (1/18+\x/3, -0.5) (1/6+\x/3, 0) -- (1/6+\x/3, -0.5) (1/6+\x/3,0) -- (5/18+\x/3, -0.5);
        }
    
        \foreach \x in {0,1,...,80}{
            \draw (1/81+\x/9, -1) rectangle (8/81+\x/9, -0.5);
        }
    
        \draw[left] (9, 3.25) node{$X$};

        \draw(6.5, 1.25) node{$v$};

        \filldraw[yellow, fill opacity=0.3] (6+1/9,1) -- (6+1/9,1.5) -- (5.5, 4.05) -- (7.5, 4.05) -- (6+8/9,1.5) -- (6+8/9, 1) -- cycle;

        \draw (4.5, -1) node[below]{$T(X,X)$};
    
    \end{scope}

\end{tikzpicture}
        \caption{Precision sampling tree built using $X$ vs $X$. $\Delta^{\le K}_{v,s}(X,Y)$ depends upon $X_v$ and $Y'=Y[i_v+s-K_v\dd j_v+s+K_v)$. If $Y'$ is con\-tained in $Y[y_2\dd y_3)=X[y_2+s_2\dd y_3+s_2)$, then $\TD^{\le K}_{v,s}(X,Y)=\min(\TD^{\le 3K}_{v,s+s_2}(X,X),K-|s|)$.}    \label{fig:4b}
    \end{minipage}
\end{figure}
Figure~\ref{fig:4a} illustrates such a decomposition of $Y$. 
Consider a node $v$ and a shift $s$ such that $[i_v+s-K_v, j_v+s+K_v) \subseteq [y_i \dd y_{i+1})$ holds for some even $i\in [0\dd 2k]$.
Then, according to \cref{prp:ipm}, $Y[i_v+s-K_v, j_v+s+K_v)=X[i_v+s+s_i-K_v, j_v+s+s_iK_v)$, where $s_i \in [-2K \dd 2K]$. Let us maintain a precision sampling tree $\tilde{T}(X,X)$ (that is, two copies of $X$) with a slightly higher threshold of $3K$ instead of $K$. 
Since we only allow substitutions on $X$, the algorithm from the preceding section (supporting only substitutions) can maintain $\tilde{T}^{\leq 3K}(X,X)$ in expected update time of $b^2(\log{n})^{\Oh(\log_{b}{n})}$ simply by replacing $Y$ with $X$ in that section (and changing the threshold). 
We observe that
$Y[i_v+s-K_v\dd j_v+s+K_v)=X[i_v+s+s_i-K_v\dd j_v+s+s_i+K_v)$ implies, by 
the locality lemma (\cref{lem:locality}) (which is a slight generalization of \cref{lem:locality2}), $\min(\TD_{v,s}(X,Y), K-|s|)=\min(\TD_{v,s+s_i}(X,X), K-|s|)$ as they are bounded by $\min(K-|s|, |I_v|) \leq K_v$. We get
\begin{align*}
    \TD_{v,s}^{\le K}(X,Y) 
    &= \min(\TD_{v,s}(X,Y), K-|s|)  & \text{From \cref{lem:captd}}\\
    &= \min(\TD_{v,s+s_i}(X,X), K-|s|) & \text{From \cref{lem:locality}}\\
     &= \min(\TD_{v,s+s_i}(X,X), K-|s|, 3K - |s+s_i|) & \text{Since $|s_i+s|\le |s_i|+|s| \le 2K+|s|$}\\
    &= \min(\TD_{v,s+s_i}^{\le 3K}(X,X), K-|s|) & \text{From \cref{lem:captd}}
\end{align*}
Therefore, in this scenario (formalized in \cref{lem:reuse}), we can obtain $\Delta^{\le K}_{v,s}(X,Y)$ essentially by just copying from $\Delta_{v,s+s_i}^{\leq 3K}(X,X)$; this can add a small extra additive error due to shift sparsification.

Overall, we only need to compute $\Delta^{\le K}_{v,s}(X,Y)$ from scratch if $v$ is in scope of $\{y_1,y_3,..,y_{2k-1}\}$. This can be done in time $b^2k(\log{n})^{\Oh(\log_b{n})}$ again following the analysis of the preceding section (\cref{sec:subXsubY}) as long as we have access to $T^{\leq 3K}(X,X)$. 
Whenever we need the value of $\Delta_{v',s'}^{\le K}(X,Y)$ for a quasi-relevant shift $s'$ at a node $v'$, if the shift is not relevant, we can copy the value from appropriate $\Delta_{v',s''}^{\leq 3K}(X,X)$ given by \cref{lem:reuse}. Finally, we only need to return $\Delta_{v_{root},0}^{\le K}(X,Y)$, which is always relevant for $\{y_1,y_3,..,y_{2k-1}\}$. 

We run this entire procedure (\cref{alg:dyn-edY}) lazily after every $k/2$ updates so that the amortized expected cost over $k/2$ updates is $b^2(\log{n})^{\Oh(\log_b{n})}$ with high probability. Since the edit distance can only change by $k/2$ in $k/2$ updates, we still solve the $(k/2,k\cdot \Oh(b\log_b n))$-gap edit distance problem correctly with high probability in amortized expected update time of $b^2 \cdot (\log n)^{\Oh(\log_b n)}$.

\subsection{Step 3: Edits in both X and Y (\cref{sec:edXedY})}
We now arrive at the most general dynamic scenario where we allow substitutions, insertions, and deletions in both $X$ and $Y$.
If we can maintain approximately the values $\Delta^{\le 3K}_{v,s}(X,X)$, then rest of the algorithm can follow the approach presented in \cref{sec:subXedY} to approximate $\ED(X,Y)$ based on $T^{\le 3K}(X,X)$. However, there are multiple significant challenges on the way.

The first difficulty is that the tree $T$ cannot be static anymore: insertions and deletions require adding and removing leaves. Multiple insertions can violate our upper bound $\Oh(b)$ on the node degrees, 
and thus occasional rebalancing is needed to simultaneously bound the degrees and the tree depth $d \le 2\log_b n$. What makes rebalancing challenging in our setting is that the data we maintain has bi-directional dependencies. On the one hand, the values $\TD^{\le K}_{v,s}(X,X)$ we approximate depend on the shape of the subtree rooted at $v$. On the other hand, our additive approximation rates $\beta_v$ depend on the values $u_w$ sampled at the ancestors of $w$.
Thus, we cannot use rotations or other local rules.
Instead, we use weight-balanced B-trees~\cite{BDF05,AV03}, which support $\Oh(\log_b n)$-amortized-time updates and satisfy several useful properties. Specifically, we use the following property crucially: {\it if a non-root node $v$ at depth $d_v$ is rebalanced, then $\Omega(b^{d-d_v})$ updates in its subtree are necessary before it is rebalanced again.} 

Whenever a node $v$ is rebalanced, we will recompute all the data associated to all its descendants $w$ (including the approximation rates $\beta_w$, the precisions $u_w$, the sets $S_w$, and the values $\Delta^{\le K}_{w,s}(X,X)$). Moreover,
we will recompute the values $\Delta^{\le K}_{w,s}(X,X)$ for all ancestors $w$ of $v$.

\begin{restatable*}{lemma}{lemrebalance}\label{lem:rebalance}
    Upon rebalancing a node $v$, the precision sampling tree $T$ can be updated in $|I_v|\cdot b^2 \cdot (\log n)^{\Oh(\log_b n)}$ expected time.
\end{restatable*}

The above lemma along with the property of weight-balanced B-trees allows us to amortize the cost over $|I_v|$ updates. We can thus separate the description of the rebalancing subroutine from the description of the actual update algorithm. From now on, we assume that an insertion or a deletion of a single character $X[x]$ does not alter the precision sampling tree except that the underlying leaf is enabled or disabled (assume there is an invisible leaf which becomes visible when an insertion happens; symmetrically, a visible leaf becomes invisible upon deletion).

Another major difficulty is that, even though we are considering $\Tt(X,X)^{\leq 3K}$, the number of relevant shifts can become $\Omega(K^2)$. Consider a an insertion at $X[x]$ (deletion is symmetric). For every leaf $v$ with $I_v=\{i\}$, the value $\TD^{\le K}_{v,s}(X,X)$
changes whenever $x\in (i_v \dd i_v+s)$. 
This is simply because, if $X$ has been obtained by inserting $X[x]$ to  $X'=X[0\dd x)\cdot X(x\dd |X|)$, then the character $X[i_v+s]$ is derived from $X'[i_v+s-1]$ rather than $X'[i_v+s]$. 
However, note that in this scenario, we have  $\TD^{\le K}_{v,s}(X,X)=\TD^{\le K}_{v,s-1}(X',X')$. Thus, instead of recomputing $\Delta^{\le K}_{v,s}(X,X)$ we should rather observe that we can copy $\Delta^{\le K}_{v,s-1}(X',X')$ instead.

To avoid the blow up in relevant shifts, we associate unique labels with characters of $X$ that stay intact even while insertions and deletions change the character's indices. In particular, the shift $s$ at node $v$ should be identified with the label of $X[i_v+s]$.
For this, we maintain a uniquely labelled string $\Xd$ over $\Sigma \times \Zz$. 
For $c:=(a,\ell)\in \Sigma\times \Zz$, we say that $\val(c):=a$ is the \emph{value} of $c$ and $\lbl(c):=\ell$ is the \emph{label}
of $c$. For a node $v \in T$ with $I_v=[i_v\dd j_v]$, we shall store $\lft_v = \lbl(\Xd[i_v+K])$ and $\rgt_v = \lbl(\Xd[j_v+K])$ (the addition of $K$ is to avoid out-of-bound indices; this is handled by actually assigning labels to a longer string $\Xd = \$^K\cdot X \cdot \$^K$). Observe that, unlike indices $i_v,j_v$, which are shifted by edits at position $x<i_v$, the labels $\lft_v$ and $\rgt_v$ generally stay intact (unless the leaves in subtree of $v$ are inserted or deleted).

We can now define shifts as $\shf_v(\Xd,\ell) = \unlbl(\Xd,\ell)-\unlbl(\Xd,\lft_v)$, where $\unlbl(S,\ell)$ returns the index in $S$ of the unique label $\ell$. 
Furthermore, if $s:= \shf_v(\Xd,\ell)\in [-K\dd K]$, we denote $\TD_{v,\ell}^{\le K}(X,X):=\TD_{v,s}^{\le K}(X,X)$.
We also say that $(v,\ell)$ is \emph{relevant} for an insertion at $X[x]$ whenever $(v,\shf_v(\Xd,\ell))$ is relevant for that insertion.

\begin{figure}[h]
    \centering
    \begin{tikzpicture}[scale=0.8]

    \begin{scope}
    \draw[fill=black!70] (1,1.5) rectangle (8,2);
    \draw (1,1.75) node[left]{$X$};

    \draw[fill=black!70] (1,0) rectangle (8,0.5);
    \draw (1,0.25) node[left]{$X$};

    \draw[fill=orange] (2.2, 0) rectangle (3.7, 0.5);
    \draw (2.2, 0) node[below]{$\vphantom{i j}i$};
    \draw (3.7, 0) node[below]{$\vphantom{i j}r$};
    \draw[fill=green!50!black] (4.1, 1.5) rectangle (5.6, 2);
    \draw (4.1, 1.5) node[below]{$\vphantom{i j}i{+}s$};
    \draw (5.6, 1.5) node[below]{$\vphantom{i j}j{+}s$};
    \draw[densely dotted] (2.2, 0.5) -- (4.1, 1.5) (3.7, 0.5) -- (5.6, 1.5);
    \end{scope}

    \begin{scope}[yshift=-3.5cm]
        \draw[fill=black!70] (1,1.5) rectangle (8.25,2);
        \draw (1,1.75) node[left]{$X'$};
    
        \draw[fill=black!70] (1,0) rectangle (8.25,0.5);
        \draw (1,0.25) node[left]{$X'$};    
        \draw (4.125, 2) node[above]{$x$};
        \draw[fill=white!50!red] (4,0) rectangle (4.25,0.5) (4, 1.5) rectangle (4.25, 2);
    
        \draw[fill=orange] (2.2, 0) rectangle (3.7, 0.5);
        \draw (2.2, 0) node[below]{$\vphantom{i j}i$};
        \draw (3.7, 0) node[below]{$\vphantom{i j}j$};
        \draw[fill=green!50!black] (4.35, 1.5) rectangle (5.85, 2);
        \draw (4.35, 1.5) node[below]{$\vphantom{i j}i{+}s{+}1$};
        \draw (5.85, 1.5) node[below]{$\vphantom{i j}j{+}s{+}1$};
        \draw[densely dotted] (2.2, 0.5) -- (4.35, 1.5) (3.7, 0.5) -- (5.85, 1.5);
    \end{scope}

    \begin{scope}[xshift=11cm]
        \draw[fill=black!70] (1,1.5) rectangle (8,2);
        \draw (1,1.75) node[left]{$X$};
    
        \draw[fill=black!70] (1,0) rectangle (8,0.5);
        \draw (1,0.25) node[left]{$X$};
    
        \draw[fill=orange] (2.2, 0) rectangle (3.7, 0.5);
        \draw[fill=orange] (2.2, -0.1) rectangle node{$\texttt{ABCD}$} (3.7, -0.6);
        \draw (3.7, -0.35) node[right]{labels};

        \draw[fill=green!50!black] (4.1, 1.5) rectangle (5.6, 2);
        \draw[fill=green!50!black] (4.1,2.1) rectangle node{$\texttt{@\#\$\%}$} (5.6, 2.6);
        \draw (5.6, 2.35) node[right]{labels};
        \draw[densely dotted] (2.2, 0.5) -- (4.1, 1.5) (3.7, 0.5) -- (5.6, 1.5);

    \end{scope}

    \begin{scope}[xshift=11cm, yshift=-3.5cm]
        \draw[fill=black!70] (1,1.5) rectangle (8.25,2);
        \draw (1,1.75) node[left]{$X'$};
    
        \draw[fill=black!70] (1,0) rectangle (8.25,0.5);
        \draw (1,0.25) node[left]{$X'$};    
        \draw (4.125, 2) node[above]{$x$};
        \draw[fill=white!50!red] (4,0) rectangle (4.25,0.5) (4, 1.5) rectangle (4.25, 2);
        \draw[fill=orange] (2.2, -0.1) rectangle node{$\texttt{ABCD}$} (3.7, -0.6);
        \draw (3.7, -0.35) node[right]{labels};
    
        \draw[fill=orange] (2.2, 0) rectangle (3.7, 0.5);
        \draw[fill=green!50!black] (4.35, 1.5) rectangle (5.85, 2);
        \draw[fill=green!50!black] (4.35,2.1) rectangle node{$\texttt{@\#\$\%}$} (5.85, 2.6);
        \draw (5.85, 2.35) node[right]{labels};
        \draw[densely dotted] (2.2, 0.5) -- (4.35, 1.5) (3.7, 0.5) -- (5.85, 1.5);
    \end{scope}

\end{tikzpicture}
    \caption{Original string $X$ and the new string $X'$ upon insertion of $x$ into $X$. (Left) Working with ordinary string representation, $\TD_{v,s}(X,X)\ne \TD_{v,s}(X',X')$. (Right) Working with uniquely labelled strings, $\TD_{v,@}(X,X)=\TD_{v,@}(X',X')$.}
    \label{fig:5}
\end{figure}
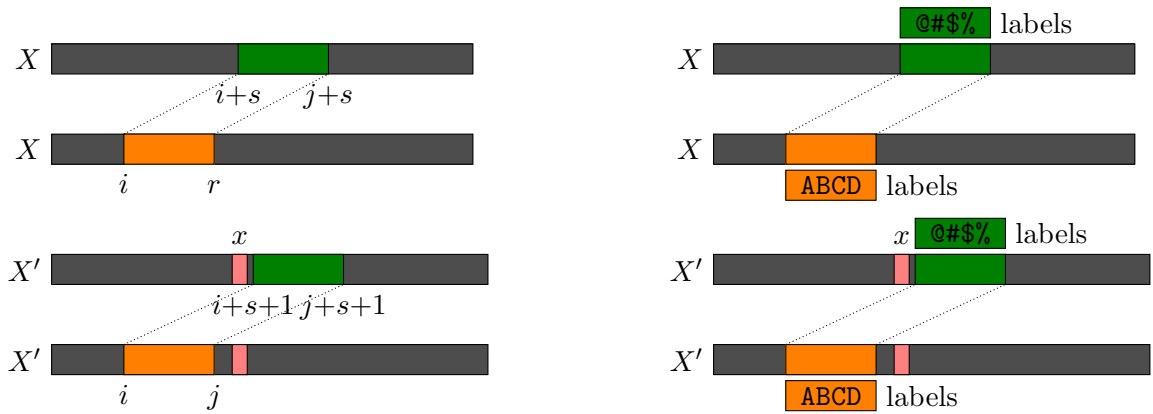

Working with labeled strings comes with significant advantages. Previously, whenever we had $\TD^{\le K}_{v,s}(X,X)=\TD^{\le K}_{v,s-1}(X',X')$, we would now have $\TD_{v,\ell}^{\le K}(X,X)=\TD_{v,\ell}^{\le K}(X',X')$ instead. Figure~\ref{fig:5} illustrates this. Adapting the proof of \cref{cor:locality}, we get the corollary below.

\begin{restatable*}{corollary}{corlocalityedit}\label{cor:locality-edit}
  Consider a string $X\in \Sigma^*$, a position $x\in [0\dd |X|)$, and a string $X'=X[0\dd x)\cdot X(x\dd |X|)$.
  Let $T$ be a precision sampling tree with $n$ leaves and let $T'$ be obtained from $T$ by removing the $x$th leaf.
  Consider a node $v\in T$ and a label $\ell\in \Lbl(\Xd)$ such that $s:= \shf_v(\Xd,\ell)$.
  If $s\in [-K\dd K]\setminus \ri{x}{v}$, then $\TD_{v,\ell}^{\le K}(X,X)=\TD_{v,\ell}^{\le K}(X',X')$.
\end{restatable*}

With this, we can almost follow our previous strategy to compute $\Delta_{v,\ell}^{\le K}(X,X)$ as in \cref{sec:subXsubY}: for each node $v$, identify the set of \emph{relevant} shifts $\ri{x}{v}$,
for each shift $s\in S_v\cap \ri{x}{v}$, update the underlying values $\Delta^{\le K}_{v,\ell}(X,X)$ based on the values $\Delta^{\le K}_{v_h,\ell_h}(X,X)$
obtained from the children $v_h$, with the shift $s_h$ restricted to $s_h \in S_{v_h}\cap \qi{x}{v}$, where $\qi{x}{v}$ consists of quasi-relevant shifts that are close to relevant ones. Note that $s := \shf_v(\Xd,\ell) = \unlbl(\Xd,\ell)-\unlbl(\Xd,\lft_v)$.

The remaining challenge is that, if an insertion at $X[x]$ transforms $X'$ to $X$ and $\shf_v(\Xd,\ell)\ne \shf_v(\Xd',\ell)$, then we should also reflect the change in the contents of the set $S_v$ of allowed shifts. In particular, our algorithm now maintains a set $L_v = \{\ell\in \Lbl(\Xd) : \shf_v(\Xd,\ell)\in S_v\}$, which is more `stable' than $S_v$ itself. The only issue is that, unlike $[-K\dd K]$, which is static, the set $\{\ell\in \Lbl(\Xd) : \shf_v(\Xd,\ell)\in [-K\dd K]\}$ may change. Fortunately, subject to an insertion at $X[x]$, it may change only for nodes in scope of the insertion,
with at most one element entering $L_v$ (the label of the newly inserted character) and one leaving $L_v$ (which used to correspond to $\pm K$ but now corresponds to $\pm (K+1)$). 
Thus, our algorithm can keep track of these changes, and, whenever a new label is inserted, we can insert it to $L_v$ with probability $\max(1, 48\beta_v^{-1}(b_v+1)\ln n)$ so that the set $S_v$ is still distributed as in \cref{sec:ako-static}. This also explains why we use randomization in the shift sparsification process. 

Finally, since $\TD^{\le K}_{v,s}(X,X)$ is capped by $K-|s|$, if an insertion of $x$ happens in between $I_v$ and $I_v$
 shifted by $s$, that is $[\min(x-s,x)\dd \max(x-s,x)]\cap I_v \ne \emptyset$ and $|s|\ge K-|I_v|$, then $\TD^{\le K}_{v,s}(X,X)=K-|s|$ might hold.
 In that case, $\TD^{\le K}_{v,s}(X,X)$ may change by $\pm 1$ as the numerical value of $s$ changes. 
 To account for that, we have to expand the set of relevant shifts.

 \begin{restatable*}{definition}
 {newrelevance}\label{def:newrelevance}
     Consider a string $v\in T$ and an index $x\in [0\dd |X|)$. We define the set $\ri{x}{v}\sub [-K\dd K]$ of shifts \emph{relevant} for an insertions at $X[x]$
  so that $s\in \ri{x}{v}$ if it satisfies at least one of the following conditions:
  \begin{enumerate}
    \item $x\in I_v$,
    \item $[x-s-K_v\dd x-s+K_v]\cap I_v \ne \emptyset$, or
    \item $[\min(x-s,x)\dd \max(x-s,x)]\cap I_v \ne \emptyset$ and $|s|\ge K-|I_v|$.
  \end{enumerate}
  Moreover, the set $\qi{x}{v}$ of shifts quasi-relevant for an insertion at $X[x]$ is defined as 
  \[ \qi{x}{v} = \{s \in [-K\dd K] : \exists_{s'\in \ri{x}{v}}\; |s-s'|\le K_v\}.\]
  Furthermore, we say that a node $v$ is in scope of an insertion at $X[x]$ if $\ri{x}{v}\ne \emptyset$.
 \end{restatable*}

Nevertheless, even with this modified definition, we can bound the number of quasi-relevant shifts by $\Oh(K \log_b n)$ (\cref{lem:subtree-edit}). The expected number of active nodes remain bounded at $(\log{n})^{\Oh(\log_{b} n)}$ (\cref{lem:active-edit}), which together with the bound on the quasi-relevant shifts leads to an expected update time of $b^2\cdot (\log{n})^{\Oh(\log_b{n})}$ to maintain $\Tt^{\leq 3K}(X,X)$. Following the analysis of \cref{sec:subXedY} for computing $\Delta_{v_{root},0}^{\le K}(X,Y)$, we get our final theorem. 

\thmmain*

\section{The Andoni-Krauthgamer-Onak Framework }\label{sec:ako-static}
Our dynamic algorithm builds on the algorithm of Andoni, Krauthgamer, and Onak~\cite{AKO10}. 
We begin with a self-contained overview of the algorithm of  \cite{AKO10} in the static setting, closely following the description presented in \cite{BCFN22a, BCFN22b}, with some crucial modifications that we describe along the way.

Given two length-$n$ strings $X$ and $Y$, the goal is to solve the $(k,K)$-gap edit distance problem for $K/k = (\log n)^{\Oh(1/\epsilon)}$. An important ingredient of their algorithm is a way to split the edit distance computation into smaller and independent subproblems. As the edit distance might depend on a global alignment of the two strings, one cannot simply divide the two strings into equally sized smaller blocks, recursively compute their edit distances, and combine the results. The main idea is to  compute the edit distances of a block in one string with several shifts of its corresponding block in the other string, and then carefully combine the results. In \cite{AKO10}, this is formalized through an auxiliary string similarity measure called  \emph{tree distance} (originally referred to as \emph{$\mathcal{E}$-distance}) that gives a good approximation to the edit distance. The tree distance is defined over an underlying tree $T$, which we call the \emph{precision sampling tree}, that splits the computation into independent subproblems.
\paragraph*{Precision Sampling Tree: Approximate Balanced Condition}

\defpstree

Here, we consider $T$ as an approximately balanced $b$-ary tree with $n:=|X|$ leaves so the depth of $T$ is bounded by $\Oh(\log_b n)$. In the static setting, $T$ is a perfectly balanced $b$-ary tree. As we will see later in the dynamic setting (while allowing insertions and deletions), it is not possible to maintain $T$ to be perfectly balanced. However, we will maintain $T$ to be approximately balanced,  with the degree of each internal node between $0.5b$ to $2b$ and depth at most $2\log_b n$.

The label $I_v$ of the $i$-th leaf of $T$ (from the left) is set to $\{i\}$, which also uniquely determines the intervals for all internal nodes. The tree determines a decomposition of two $n$ length strings $X$ and $Y$ such that, for a node $v$ labeled with interval $I_v = [i \dd j)$, we consider the substrings $X[i \dd j)$ and $Y[i \dd j)$.

We now define the tree distance as follows:
\deftd

The following establishes the relation (and the usefulness) of the tree distance measure with the edit distance.

\lemequivedtd

In particular, since an approximately balanced $b$-ary tree with $n$ leaves has depth bounded by $\Oh(\log_b n)$ and degree bounded by $\Oh(b)$, the above lemma shows that $\TD_T(X, Y)$ is at most $\Oh(b\log_b n)$ times the edit distance of $X, Y$. In the algorithm of \cite{AKO10}, the idea is to approximately compute the tree distance at each node $v$ in the precision sampling tree. For the leaf nodes, one can directly evaluate the edit distance, and for the internal nodes use \eqref{eq:tree-distance} to combine the recursively computed values.

\paragraph{Capped Distance: New Definition}

Tree distance \eqref{eq:tree-distance} minimizes over an infinite number of shifts $s' \in \Int$. Since we are only interested in the $(k, K)$-gap problem, if the tree distance exceeds the threshold $K$, we can report that they are far. Hence, one can restrict the attention to approximating the \emph{capped} tree distance by considering only shifts $s' \in [-K\dd K]$.

We introduce a few extra notations here. 
The substring $X[i\dd j)$ of $X$ corresponding to node $v$ is denoted by $X_v$. In particular, $X_v$ is a single character for each leaf $v$, and $X_v = X$ when $v$ is the root of $T$. 
 $T_v$ denote the subtree of $T$ rooted at node $v$.

\begin{definition}[Capped Edit Distance] \label{def:capped}
For strings $X,Y$ and $K \geq 0$, define the \emph{$K$-capped edit distance as $\ED^{\leq K}(X, Y)$ $= \min(\ED(X, Y), K)$.}
\end{definition}

\defcaptd
{\it Remark.}
 Note that the above definition of capped tree distance (\ref{eq:capped-td}) is different from \cite{AKO10,BCFN22a, BCFN22b} where the recurrence is defined as

\begin{align}
\label{eq:cap-old}
    \TD_{v, s}^{\leq K}(X, Y) = \min\left(\sum_{h=0}^{b_v-1} \,\min_{s' \in [-K \dd K]}\, \left( \TD_{v_h, s'}^{\leq K}(X, Y) + 2 \cdot |s - s'| \right),\, K \right).
\end{align}

While (\ref{eq:cap-old}) is sufficient to show 
$\TD^{\leq K}(X, Y)=\min(\TD(X, Y), K)$, 
it is not generalizable to internal nodes and for all shifts $s \in [-K \dd K]$. In particular, $\TD_{v,s}^{\leq K}(X, Y)=\min(\TD(X, Y), K)$ does not hold for any internal node $v$, and shifts $s \in [-K \dd K]$ under the old definition (\ref{eq:cap-old}).

Using our new definition, the following lemma relates capped tree distance to the original tree distance.

\lemcaptd
\begin{proof}
Consider any $v$, and $s \in [-K \dd K]$, and consider the subtree $T_v$ rooted at $v \in T$. Let $s_{s,w}^{\leq K}$ denote the optimal shift chosen at node $w \in T_v$ according to definition (\ref{eq:capped-td}) while computing $\TD_{v,s}^{\leq K}(X,Y)$. Clearly, these are valid shifts to choose while computing $\TD_{v,s}(X,Y)$ according to definition (\ref{eq:tree-distance}). Hence, $\TD_{v,s}^{\leq K}(X,Y) \geq \min(\TD_{v,s}(X,Y), K-|s|)$. 

Now, we show the other direction, that is $\TD_{v,s}^{\leq K}(X,Y) \leq \min(\TD_{v,s}(X,Y), K-|s|)$. We may assume $\TD_{v,s}(X,Y) <  K-|s|$ as otherwise the statement is clear. Let $s_{s,w}$ denote the optimal shift chosen for any node $w \in T_v$ according to definition (\ref{eq:tree-distance}) for computing $\TD_{v,s}(X,Y)$.
Consider the path from $v$ to $w$ in $T$, and let $x_i$ denote the $i$th vertex in that path. Therefore, $x_0=v$, and $x_l=w$ where $l$ is the length of the path. Then using (\ref{eq:tree-distance}) 
\begin{align*}
|s_{s,w}-s| \leq \sum_{i=1}^{l} |s_{s,x_{i}}-s_{s,x_{i-1}}| \leq K-|s|
\end{align*}
which implies $|s_{s,w}| -|s|\leq |s_{s,w}-s| \leq K-|s|$, or $s_{s,w} \in [-K \dd K]$. Therefore, the shifts $s_{s,w}$ for all nodes $w \in T_v$ are among the valid shifts to choose in the computation of $\TD_{v,s}^{\leq K}(X,Y)$. Hence, we have $\TD_{v,s}^{\leq K}(X,Y) \leq \TD_{v,s}(X,Y)$.
\end{proof}

\paragraph{Accuracy Parameters: Sparsification \& Pruning}
We introduce two approximation parameters $\alpha_v$ and $\beta_v$ for a node $v$, i.e., the output at node $v$ is allowed to be $(\alpha_v, \beta_v)$-approximate. The \emph{multiplicative} accuracy $\alpha_v$ determines the multiplicative approximation ratio at $v$. The \emph{additive} accuracy $\beta_v$ determines the number of characters allowed to be read at $v$ which is $(2\beta_v)^{-1}|I_v|$. 

\noindent{\it Sparsification.}
For every node $v$, we restrict the set of possible shifts $s$ at $v$ with respect to the additive accuracy $\beta_v$. In particular we define a set of \emph{allowed} shifts $S_{v} \sub [-K \dd K]$ with $\Pr[s\in S_v]=\min(1,48\beta_v^{-1}(b_v+1)\ln n)$ independently across all $s\in [-K\dd K]$. Then instead of going over all possible shifts in $[-K \dd K]$, we only try for the set of allowed shifts $S_v$. We show with high probability, this only incurs another additive $\beta_v/2$-approximation at each node. Combined with approximation due to the number of characters read at each node, this provides a total additive accuracy $\beta_v$ at node $v$. 

\noindent{\it Pruning.}
We also introduce the following pruning rule: \emph{If~$|I_v| \leq \beta_v$, then return $\Delta_{v, s}^{\le K}(X,Y) = 0$ for all $s \in S_v$.} This is correct, since he $K$-capped tree distance $\TD^{\le K}_{v,s}(X,Y)$ is at most $|I_v|$. Hence, $0$ is an $(\alpha_v, \beta_v)$ approximation of $\TD^{\le K}_{v,s}(X,Y)$.

We now specify the computational task for each node:
\begin{definition}[Tree Distance Problem] \label{def:problem}
Given a node $v$ of the precision sampling tree $T$, let $S_{v} \subseteq [-K\dd K]$ be defined as above. 
The goal is to compute, for every $s\in S_v$, a value $\Delta_{v, s}^{\le K}(X,Y)$ that is an $(\alpha_v,\beta_v)$-approximation of $\TD^{\leq K}_{v, s}(X, Y)$, that is,
\begin{equation} \label{eq:delta-approx}
    \tfrac1{\alpha_v} \TD_{v,s}^{\leq K}(X,Y) - \beta_v \leq \Delta^{\le K}_{v, s}(X,Y) \leq \alpha_v \TD^{\leq K}_{v, s}(X, Y) + \beta_v.
\end{equation}
\end{definition}

In fact, one can compute these values with the stronger guarantee that $\Delta_{v,s}^{\le K}(X,Y)$ is an $(\alpha_v, \beta_v)$-approximation of $\TD_{v, s}^{\leq K}(X, Y)$ w.h.p. (see \cref{lem:delta}). 

\paragraph{Precision Sampling.}
In the next step, we assign the appropriate parameter $\beta_v$ to each node in the precision sampling tree. The challenge now is to assign the parameters to $\beta_{v_h}$ the $b_v$ children of $v$ in such a way as to obtain a good approximation of tree distance at $v$ after combining the results from the children. Naive solutions would require $\sum_{h=0}^{b_v}\beta_{v_h}\le \beta_v$ since additive errors cumulate. The Precision Sampling Lemma~(\cref{cor:precision-sampling}) allows much larger errors $\beta_{v_h}$ with $\Exp[\beta_v / \beta_{v_h}]=\Oh(\polylog(n))$.

Finally, we need an efficient algorithm to combine the computed tree distances. Specifically, the following subproblem can be solved efficiently:

\begin{lemma}[{\cite[Lemma 17]{BCFN22b}}]\label{lem:range-minimum}
There is an $\Oh(|S|+|S'|)$-time algorithm for the following problem: Given integers $A_{s'}$ for $s' \in S'$ and a set $S$,
compute the following values for each $s\in S$:
\begin{equation*}
    B_s = \min_{s' \in S'} \left(A_{s'} + 2 \cdot |s - s'|\right).
\end{equation*}
\end{lemma}

Improving on the approximate solution in \cite{AKO10}, Bringmann, Cassis, Fischer, and Nakos \cite{BCFN22a, BCFN22b} give an exact algorithm to solve this. 

We now present the AKO algorithm (with our modifications) in full in \cref{alg:ako} and then proceed to bound the error probability and run time of this algorithm. We start with fixing some parameters first.

\begin{algorithm}[t]
\caption{The Andoni-Krauthgamer-Onak Algorithm} \label{alg:ako}
\begin{algorithmic}[1]
\medskip
\Statex \textbf{Input:} 
Strings $X,Y$, a node $v$ in tree $T$ with allowed shifts $S_{v}$ and additive accuracy $\beta_v$
\Statex \textbf{Output:} The values $\Delta_{v,s}^{\le K}(X,Y)$ for shifts $s \in S_{v}$
\medskip
\Statex \hrule

\Procedure{Combine}{$v,S,D$}
\For{$h\in [0\dd b_v)$}
    \State Compute $\widetilde A_{h, s} = \min_{(h',s',\Delta) \in D\;:\;h'=h} \left(\Delta + 2 \cdot |s - s'|\right)$ 
    for all  $s\in S$\label{alg:ako:line:range-minimum} \Comment{\cref{lem:range-minimum}}
\EndFor
\For{$s \in S$} \label{alg:ako:line:iter-output}
    \State $\delta_{v, s} = \Call{Recover}{\widetilde A_{0, s}, \dots, \widetilde A_{b_v-1, s},u_{v_0}, \dots, u_{v_{b_v-1}}}$ \Comment{\cref{cor:precision-sampling}} \label{alg:ako:line:recovery}
\EndFor
\State\Return $\Delta_{v,s}^{\le K}(X,Y)=\min(\delta_{v, s}, |I_v|,K-|s|)$ for all $s \in S$ \label{alg:ako:line:return}
\EndProcedure
\Statex \hrule
\Procedure{AKO}{$v,\beta_v$}
\If{$|I_v| \le \beta_v $} \label{alg:ako:line:test-2-condition}
    \State\Return $\Delta_{v,s}^{\le K}(X,Y)= 0$ for all $s \in S_{v}$ \label{alg:ako:line:test-2}
\EndIf
\If{$v$ is a leaf} \label{alg:ako:line:test-1-condition}
    \State\Return $\Delta_{v,s}^{\le K}(X,Y)= \TD^{\le K}_{v,s}(X,Y)$ for all $s \in S_{v}$ \label{alg:ako:line:test-1}
\EndIf
\For{$h \in [0\dd b_v)$} \label{alg:ako:line:iter-children}
    \State Let $v_h$ be the $h$-th child of $v$
    \State Sample $u_{v_h} \sim \hDst((4\log n)^{-1}, 0.01\cdot K^{-1} \cdot n^{-4})$ \label{alg:ako:line:precision} \Comment{\cref{cor:precision-sampling}}
    \State Call $\Call{AKO}{v_h,\frac12 \beta_v u_{v_h}}$ to compute $\Delta_{v_h, s_h}^{\le K}(X,Y)$ for all $s_h \in S_{v_h}$
    \label{alg:ako:line:recursion}
\EndFor
\State \Return {\sc Combine}($v,S_v,\{(h,s_h,\Delta_{v_h,s_h}^{\le K}(X,Y)): h\in [0\dd b_v), s_h\in S_{v_h}\}$)
\EndProcedure
\end{algorithmic}
\end{algorithm}

\paragraph{Fixing the Parameters.}
\label{par:parameter-sko-static}
We assume that $T$ is an $\Oh(b)$-ary tree with $n$ leaves and depth $d\le 2\log_b n$. We also fix the approximation parameters $\alpha_v$ and $\beta_v$ for each node $v$.
For the root, we set $\beta_v$ to be a constant fraction of $K$, say $K/1000$. If $w$ is a child of $v$ then sample~$u_w \sim \hDst((4 \log n)^{-1}, 0.01 \cdot K^{-1} \cdot n^{-4})$ independently and set $\beta_w = \frac12\beta_v u_w$.
For any node $v$, set $\alpha_v = 2 \cdot (1 - (4 \log n)^{-1})^{d_v}$ where $d_v$ is the depth of $v$. Note that $1\le \alpha_v \le 2$ since $d_v \leq 2\log_b n \le 2\log n$.

\subsection{Correctness}\label{subsec:correctness}
We show that for each node $v$, \cref{alg:ako} solves the tree distance problem at $v$ w.h.p. We proceed inductively over the depth of the computation tree assuming that the algorithm correctly solves all recursive calls to the children of $v$. 

The \emph{inactive} nodes with $|I_v|\le \beta_v$ are handled in \crefrange{alg:ako:line:test-2-condition}{alg:ako:line:test-2}.
This is correct because $0$ is an additive $\beta_v$-approximation of~$\TD^{\leq K}_{v, s}(X, Y)$, which is at most $|I_v|$.
The active leaf nodes are handled in \crefrange{alg:ako:line:test-1-condition}{alg:ako:line:test-1}. 
In that case, we compute $\TD^{\leq K}_{v, s}(X, Y)$ exactly. 

This serves as the base case. In the rest of the algorithm, we use the approximations recursively for $v$'s children to solve the tree distance problem at $v$. For the inductive step, assume that the values $\Delta_{v_h, s}^{\le K}(X,Y)
$ for $s \in S_v$ recursively computed in \cref{alg:ako:line:recursion} satisfy \cref{def:problem}. The algorithm then evaluates the following expression approximately (see \cref{def:tree-distance}):
\begin{equation*}
    \TD^{\leq K}_{v, s}(X, Y) = \min\Bigg(\sum_{h=0}^{b_v-1} A_{h, s},\,K-|s|\Bigg), \quad A_{h, s} = \min_{s_h \in [-K \dd K]} \left(\TD^{\leq K}_{v_h, s_h}(X, Y) + 2 \cdot |s - s_h|\right).
\end{equation*}

\newcommand{\shopt}{s_h^*}
\newcommand{\tshopt}{\tilde{s}_h^*}

In particular, for every shift $s\in S_v$, the algorithm computes
\begin{equation*}
    \widetilde A_{h, s} = \min_{s' \in S_{v_h}} \left( \Delta_{v_h, s_h} + 2 \cdot |s - s_h|\right),
\end{equation*}
By induction hypothesis, for every $s\in S_{v_h}$, the value $\Delta_{v_h,s_h}^{\le K}(X,Y)$ is an $(\alpha_{v_h}, \beta_{v_h})$-approximation of $\TD^{\le K}_{v_h,s_h}(X,Y)$. Suppose $\shopt$ is the shift that minimizes $A_{h,s}$. Let $\tshopt\in S_{v_h}$ be closest to $\shopt$.

Note that from \cref{def:tree-distance}, for any two shifts $s,s' \in [-K\dd K]$
\begin{equation}
    \Bigg|\TD_{v_h,s}^{\leq K}(X,Y)-\TD_{v_h,s'}^{\leq K}(X,Y)\Bigg|\leq 2(b_{v_h}+1)|s-s'|
\end{equation}
Now, by definition of $S_{v_h}$, for any shift $s \in [-K \dd K]$, the closest shift 
$\tilde{s} \in S_{v_h}$ satisfies 
\[\Pr\left[|s-\widetilde s| > \frac{\beta_{v_h}}{8(b_{v_h}+1)}\right] \le \left(1-\frac{48(b_{v_h}+1)\ln n}{\beta_{v_h}}\right)^{\frac{\beta_{v_h}}{8(b_{v_h}+1)}} \le \exp\left(-\frac{48(b_{v_h}+1)\ln n \cdot \beta_{v_h}}{\beta_{v_h}\cdot 8(b_{v_h}+1)}\right)\le   \frac{1}{n^6}.\]
Taking a union bound over all nodes $v_h\in T$ and shifts $s\in [-K\dd K]$, we conclude that, with probability at least $1-\frac{1}{n^4}$,
for any shift $s \in [-K\dd K]$ choosing $\tilde{s} \in S_{v_h}$ closest to $s$ gives
\begin{equation}
\label{eq:shiftsparse}
    \Bigg|\TD_{v_h,s}^{\leq K}(X,Y)-\TD_{v_h,\tilde{s}}^{\leq K}(X,Y)\Bigg|\leq \frac{\beta_{v_h}}{4}
\end{equation}
Assuming the above event, due to $1\le \alpha_{v_h}\le 2$, we have
\begin{align*}
    \tilde{A}_{h,s} & \leq \Delta_{v_h,\tshopt}^{\le K}(X,Y)+2|s-\tshopt| & \text{ since } \tshopt \in S_{v_h} \\  
    & \leq \alpha_{v_h} \TD_{v_h,\tshopt}(X,Y)+\beta_{v_h}+2|s-\shopt|+2|\shopt-\tshopt| & \text{ by induction hypothesis}\\
    & \leq \alpha_{v_h} \left(\TD_{v_h,\shopt}(X,Y)+\tfrac14\beta_{v_h}\right)+2|s-\shopt|+ \tfrac{3}{2}\beta_{v_h} & \text{from \cref{eq:shiftsparse} since $|\shopt-\tshopt| \leq \tfrac{\beta_{v_h}}{4(b_{v_h}+1)}$}\\
    & \le \alpha_{v_h} \TD_{v_h,\shopt}(X,Y) + \tfrac12 \beta_{v_h}+ 2|s-\shopt|+\tfrac{3}{2}\beta_{v_h} & \text{ since $\alpha_{v_h} \le 2$}\\
    & \le \alpha_{v_h}(\TD_{v_h,\shopt}(X,Y)+ 2|s-\shopt|)+2\beta_{v_h} & \text{ since $\alpha_{v_h} \ge 1$}\\
    &=\alpha_{v_h}\cdot A_{h,s}+2\beta_{v_h}   & \text{ from the choice of $\shopt$} 
\end{align*}
To show the other direction let $s'_h$ be the optimum shift for computing $\tilde{A}_{h,s}$.
\begin{align*}
A_{h,s} &\leq \TD_{v_h,s'_h}^{\leq K}(X,Y)+2|s-s'_h|\\
& \leq \alpha_{v_h} \Delta_{v_h,s'_h}^{\le K}(X,Y)+\beta_{v_h}+2|s-s'_h| &\text{ by induction hypothesis}\\
&\leq \alpha_{v_h}(\Delta_{v_h,s'_h}^{\le K}(X,Y)+2|s-s'_h|)+\beta_{v_h} & \text{ since $\alpha_{v_h} \ge 1$}\\
&=\alpha_{v_h}\tilde{A}_{h,s}+\beta_{v_h} & \text{ by the choice of $s'_h$}
\end{align*}
Therefore, overall we have $\tilde{A}_{h,s}$ is an $(\alpha_{v_h},2\beta_{v_h})$ approximation of $A_{h,s}$.
Recall that the additive approximation parameters $\beta_{v_h}$ are set to $\beta_{v_h} = \frac12\beta_v u_{v_h}$ for all children. 
The precision parameters $\{u_{v_h}\}$ are chosen as independent samples from $\hDst(\epsilon = (4\log n)^{-1}, \delta = 0.01 \cdot K^{-1} \cdot n^{-4})$.
We can also verify that $(1 + \epsilon) \cdot \alpha_{v_h} = 2 \cdot (1 + (4 \log n)^{-1}) \cdot (1 - (4 \log n)^{-1})^{d_{v_h}} \leq \alpha_v$.  Therefore in \cref{alg:ako:line:recovery} by applying the Precision Sampling Lemma (\cref{cor:precision-sampling}) we get $\Delta_{v, s}(X,Y)$ is an  $(\alpha_v, \beta_v)$ approximation of $\TD_{v,s}^{\leq K}(X,Y)$. 


Next, we bound the error probability of the algorithm. Except for the choice of $S_v$ (which we have already analyzed),
the only source of randomness in the algorithm is from precision sampling and the recovery using \cref{cor:precision-sampling}. The algorithm at $v$ succeeds if the recovery algorithm in \cref{alg:ako:line:recovery} succeeds for values~\makebox{$s \in S_v$}. Taking a union bound over these at most $2K+1\le 3K$ many error events, each happening at most with probability $n\delta = 0.01 \cdot K^{-1} \cdot n^{-3}$, we can bound the total error probability by $0.03/n^3$. Taking into account the error probability coming from the shifts $S_{v_h}$, the total error probability at any node $v$ does not exceed $1/n^3$.
Hence, we have the following

\begin{lemma}\label{lem:delta}
    Let $v$ be an internal node of the tree $T$. Assume that, for every child $v_h$ of $v$ and shift $s_h\in S_{v_h}$, the value $\Delta_{v_h,s_h}^{\le K}(X,Y)$ is an $(\alpha_{v_h},\beta_{v_h})$-approximation of $\TD^{\le K}_{v_h,s_h}(X,Y)$.
    With probability at least $1-\frac{1}{n^3}$, every value $\Delta_{v,s}^{\le K}(X,Y)$ returned by \Call{Combine}{$v,S_v,\{(h,s_h,\Delta_{v_h,s_h}^{\le K}(X,Y))\}$} is an $(\alpha_v,\beta_v)$-approximation of $\TD^{\le K}_{v,s}(X,Y)$.
\end{lemma}


\begin{corollary}\label{cor:ako-correctness}
With probability at least $1-\frac{1}{n^2}$,
the value $\Delta_{root,0}^{\le K}(X,Y)$ at the root $(\alpha_{root},\beta_{root})$-approximates $\TD^{\le K}(X,Y)$.
\end{corollary}

\subsection{Running Time.} We now analyze the run time of \cref{alg:ako}. First, we start by analyzing the additive accuracy $\beta_v$ at any node $v$.

\begin{lemma}\label{lem:exp-precision}
Let $T$ be a $\Oh(b)$-ary precision sampling tree with additive approximation $\beta_v$ at a node $v$. Then,
\[\Exp[\beta_v^{-1}] = \tfrac1K (\log n)^{\Oh(d_v)} = \tfrac1K \cdot (\log n)^{\Oh(\log_b \log n)}.\]
\end{lemma}

\begin{proof}

 Applying \cref{cor:precision-sampling} we obtain that, for any sample $u_v$ from $\hDst(\epsilon=(4\log n)^{-1}, \delta=0.01\cdot K^{-1} \cdot n^{-4})$,
\begin{equation*}
    \Exp[u_v^{-1}] \leq \widetilde{O}(\epsilon^{-2} \log(\delta^{-1}) \log n) = \polylog(n).
\end{equation*}
Since $\beta_v = (K \cdot u_{v_1} \dots u_{v_{d_v}})/(1000\cdot 2^{d_v})$ where $v_0, v_1, \dots, v_{d_v} = v$ is the path from root to node $v$, using the facts that $u_v$'s are sampled independently and $d_v \leq d$ we get:
\begin{equation}
    \Exp[\beta_v^{-1}] = \frac{1000}K \cdot \prod_{i=1}^{d_{v}} 2\cdot \Exp[u_{v_i}^{-1}] \leq \frac1K \cdot (\log n)^{\Oh(d_v)}  \leq \frac1K \cdot (\log n)^{\Oh(\log_b n)}.\qedhere
\end{equation}
\end{proof}

Note that the recursive calls do not necessarily reach every node in the tree. Some nodes are trivially solved by pruning; thus, their children are never explored. We say a node $v$ is \emph{active} all the ancestors $w$ of $v$ (including $v$) satisfy $|I_w|>\beta_w$. The following lemma bounds the expected number of active nodes in the computation tree by~$n / K \cdot (\log n)^{\Oh(\log_b n)}$. 

\begin{lemma}
\label{lem:node}
    The expected number of active nodes  in \cref{alg:ako} is $\frac{n}{K}(\log{n})^{\Oh(\log_b{n})}$.
\end{lemma}

\begin{proof}
We use the fact that a node $v$ is active only if $|I_v| > \beta_v$.
   Using Markov's inequality we obtain that,
\begin{equation}
\label{eq:active}
    \Pr[\text{$v$ is active}] = \Pr[\beta_v < |I_v|] = \Pr\left[\beta_v^{-1} > \frac{1}{|I_v|}\right] \leq |I_v| \cdot \Exp[\beta_v^{-1} ] \leq \frac{|I_v|}{K} \cdot (\log n)^{\Oh(d_v)}.
\end{equation}
So the expected number of active nodes at depth $d'$ is $\sum_v |I_v| / K \cdot (\log n)^{\Oh(d')} = n / K \cdot (\log n)^{\Oh(d')}$. Here the sum is taken over all nodes $v$ at depth $d'$; hence, we have $\sum_v |I_v| = |X| = n$. Therefore, summing over all depths and noting that the deepest level is $\Oh(\log_{b}{n})$, we get the lemma.
\end{proof}

We are now ready to show the running time bound.

\begin{lemma}\label{lem:ako-time}
The expected running time of \cref{alg:ako} is $nb^2 \cdot (\log n)^{\Oh(\log_b n)}$.
\end{lemma}

\begin{proof}
Consider the execution of \cref{alg:ako} at node $v$. If $v$ is a leaf, the total time taken in \crefrange{alg:ako:line:test-1-condition} {alg:ako:line:test-1} is $\Ot(|S_v|)$. If $v$ is an internal node,
we bound the running time of a single execution of \cref{alg:ako} (ignoring the recursive calls) by $\Ot(\sum_{h=0}^{b_v-1}(|S_v|+|S_{v_h}|))$.

\crefrange{alg:ako:line:test-2-condition}{alg:ako:line:test-2} and \cref{alg:ako:line:return} take $\Oh(|S_v|)$ time to produce the output.
In \crefrange{alg:ako:line:iter-children}{alg:ako:line:recursion} the for-loop runs for $b_v$ iterations. In each iteration, we sample a precision in \cref{alg:ako:line:precision}, and perform some recursive computation that we ignore for now. Therefore, ignoring the recursive computation, \crefrange{alg:ako:line:iter-children}{alg:ako:line:recursion} takes time 
$\Oh(b_v)$.

In \cref{alg:ako:line:range-minimum} we apply \cref{lem:range-minimum} 
running in time $\Oh(|S_v| + |S_{v_h}|)$ by taking $S=S_v $ and $S' = S_{v_h}$ for each $v_h$. The total time spent is thus 
$\Ot(\sum_{h=0}^{b_v-1}(|S_v|+|S_{v_h}|))$.


In \crefrange{alg:ako:line:iter-output}{alg:ako:line:recovery} the for-loop runs for $\Oh(|S_v|)$ iterations and, in each iteration, we apply the algorithm from \cref{cor:precision-sampling} with the parameters $\epsilon = \Omega(\log n)$ and $\delta \geq n^{-\Oh(1)}$. Each execution of the loop runs in time $\Oh(b_v \epsilon^{-2} \log(\delta^{-1})) = b_v \cdot (\log n)^{\Oh(1)}$. Therefore, \crefrange{alg:ako:line:iter-output}{alg:ako:line:recovery} take time 
$\Ot(b_v |S_v|)$.

Therefore, the total time taken by \cref{alg:ako} for execution at node $v$ is 
$\Ot(\sum_{h=0}^{b_v-1}(|S_v|+|S_{v_h}|))$. Since $\Exp[|S_w|]\le 48K (b_w+1) \ln n\cdot \Exp[\beta_w^{-1}]=b \cdot (\log n)^{\Oh(\log_b n)}$ holds by \cref{lem:exp-precision} for every node $w$, this is $b^2\cdot (\log n)^{\Oh(\log_b n)}$ per node. Across all nodes, the expected running time of \cref{alg:ako} is $nb^2 \cdot (\log n)^{\Oh(\log_b n)}$.
\end{proof}

\begin{remark}
\label{remark:staticvsdyn}
We could have also bounded the expected running time of \cref{alg:ako} by $\Ot(K \cdot b \cdot \sum_{v \text{is active}} \Exp[\beta_v^{-1} \mid \beta_v < |I_v|])$, which is $nb \cdot (\log n)^{\Oh(\log_b n)}$ using Lemma~\ref{lem:exp-precision} and \cref{eq:active}.
This gives a slightly better bound which is valid even if we had not sparsified the shifts. 
Moreover, the bound in \cref{lem:ako-time} remains valid even if we do not prune the inactive nodes. 
This also illustrates for the static algorithm, it suffices to either prune inactive nodes or sparsify the shifts. Only in the case of dynamic algorithms, we need to utilize both as will be evident later.
\end{remark}




\section{Dynamic Edit Distance with Substitutions}\label{sec:subXsubY}
We start with a simple dynamic scenario where 
in every time step a character in either the string $X$ or $Y$ is substituted by another character from the alphabet $\Sigma$. As part of preprocessing, we are given the precision sampling tree $\TD^{\leq K}(X,Y)$ along with $X$, $Y$, $|X|=|Y|=n$, the set of allowed shifts $S_v$ for all $v \in T$, $\tilde{A}_{v,s}$ for all $v \in T$ and $s \in S_v$, and the generated random variables $\{\beta_v, u_v\}$ for all $v \in T$ according to the distributions in Section~\ref{sec:ako-static} (See Fixing the parameters in Section~\ref{sec:ako-static}). Our goal is to maintain $\TD^{\leq K}(X,Y)$ with every update so that we can distinguish $\ED(X,Y) < k$ and $\ED(X,Y) \geq  K$ for $K=(\log{n})^{\Oh(\frac{1}{\epsilon})}k$.

\paragraph*{Algorithm.}
The main idea of our algorithm is simple. Upon a substitution at $X[i]$ or $Y[i]$, we use parts of \cref{alg:ako} to recompute all the (potentially) affected values $\Delta_{v,s}^{\le K}(X,Y)$ across active nodes $v\in T$ and shifts $s\in S_v$.

The following notion of \emph{relevance} characterizes the shifts for which $\Delta_{v,s}^{\le K}(X,Y)$ will be recomputed: 
We will show that a substitution can affect $\TD^{\le K}_{v,s}(X,Y)$ only if $s$ is relevant.
If $v$ is not a leaf, then, in order to recompute $\Delta_{v,s}^{\le K}(X,Y)$ for relevant shifts, we will need to access some values $\Delta_{v_h,s_h}^{\le K}(X,Y)$
at the children $v_h$ of $v$. To capture the scope of the necessary shifts $s_h$, we introduce the notion of \emph{quasi-relevance}.
\defrelevance
Every active node maintains the values $\Delta_{v,s}^{\le K}(X,Y)$ (which $(\alpha_v,\beta_v)$-approximate $\TD_{v,s}^{\le K}(X,Y)$) for all shifts $s \in S_v$ and can return any of them in $\Oh(1)$ time upon request. Any non-active node with an active parent can also return $\Delta_{v,s}^{\le K}(X,Y)=0$ in $\Oh(1)$ time upon request. 

Upon substitution at $X[i]$ for any $i \in [1 \dd |X|]$, the algorithm executes a {\sc Substitution-X} procedure, initially called at the root, that, at given node $v$, first recurses on all active children of $v$ that are in scope of the substitution at $X[x]$, and then updates $\Delta_{v,s}^{\le K}(X,Y)$ for all $s\in \rsX{i}{v}\cap S_v$. If $v$ is an internal nodes, this requires using the {\sc Combine} procedure whose 
outputs are the values $\Delta_{v,s}^{\le K}(X,Y)$ with $s\in \rsX{i}{v}\cap S_v$ and whose inputs are the values $\Delta_{v_h,s_h}^{\le K}(X,Y)$ for $s_h \in \qsX{i}{v}\cap S_{v_h}$ across all the children $v_h$ of $v$.

Upon substitution at $Y[i]$ for any $i \in [1 \dd |Y|]$, the algorithm executes a {\sc Substitution-Y} procedure implemented analogously.

\begin{algorithm}[t]
\caption{Dynamic-Substitution-Only} \label{alg:dyn-sub}
\begin{algorithmic}[1]
\medskip
\Statex \textbf{Input:}
The location $i$ of the substitution, a node $v\in T$.
\Statex \textbf{Output:} The updated values $\Delta_{v,s}^{\le K}(X,Y)$.
\medskip
\Statex \hrule
\Procedure{Substitution-$X$}{$i, v$}\label{alg:dyn-sub:proc-subx}
\If{$|I_v| > \beta_v$ \textbf{and} $\rsX{i}{v}\ne \emptyset$} \label{alg:dyn-sub:terminate-x}
\If{$v$ is a leaf}
    \State\Return $\Delta_{v,s}^{\le K}(X,Y) = \TD_{v,s}^{\le K}(X,Y)$ for all $s \in S_{v}\cap \rsX{i}{v}$ \label{alg:dyn-sub:leaf-x}
\EndIf
\State Recursively run \Call{Substitution-$X$}{$i,v_h$} for $h\in [0\dd b_v)$
\State \Return \Call{Combine}{$v,\rsX{i}{v}\cap S_v,\{(h,s_h,\Delta_{v_h,s_h}^{\le K}(X,Y)): h\in [0\dd b_v), s_h\in \qsX{i}{v}\cap S_{v_h}\}$}\label{alg:dyn-sub:combine-x}
\EndIf
\EndProcedure
\Statex \hrule
\Procedure{Substitution-$Y$}{$i, v$}
\label{alg:dyn-sub:proc-suby}
\If{$|I_v| > \beta_v$ \textbf{and} $\rsY{i}{v}\ne \emptyset$} \label{alg:dyn-sub:terminate-y}
\If{$v$ is a leaf}
    \State\Return $\Delta_{v,s}^{\le K}(X,Y) = \TD_{v,s}^{\le K}(X,Y)$ for all $s \in S_{v}\cap \rsY{i}{v}$  \label{alg:dyn-sub:leaf-y}
\EndIf
\State Recursively run \Call{Substitution-$Y$}{$i,v_h$} for $h\in [0\dd b_v)$
\State \Return \Call{Combine}{$v,\rsY{i}{v}\cap S_v,\{(h,s_h,\Delta_{v_h,s_h}^{\le K}(X,Y)): h\in [0\dd b_v), s_h\in \qsY{i}{v}\cap S_{v_h}\}$}\label{alg:dyn-sub:combine-y}
\EndIf
\EndProcedure
\end{algorithmic}
\end{algorithm}

\paragraph*{Analysis.}

Notice that a substitution either in $X$ or $Y$ does not alter the structure of the precision sampling tree; we just update the values $\Delta_{v,s}^{\le K}(X,Y)$ for active nodes $v\in T$ and relevant shifts $s\in S_v$. As for correctness, our invariant is that (with high probability), each value $\Delta_{v,s}^{\le K}(X,Y)$ 
is an $(\alpha_v,\beta_v)$-approximation of $\TD^{\le K}_{v,s}(X,Y)$. Thus, due to \cref{lem:delta}, it suffices to prove that $\TD^{\le K}_{v,s}(X,Y)$ may change only if the shift $s$ is relevant. As for the update time, the main effort is to bound the expected number of active nodes in scope of a given substitution and the number of quasi-relevant shifts at each such node.

As a warm-up, we consider a simple case of substitutions in $X$.
\paragraph{Substitution in $X$.}
\begin{lemma}\label{lem:subX}
Upon any substitution in $X$, we can update the values $\Delta_{v,s}^{\le K}(X,Y)$
across all active nodes $v\in T$ and shifts $s \in S_v$ in expected time $b^2(\log{n})^{\Oh(\log_b{n})}$.
\end{lemma}
\begin{proof}
Suppose the $i$-th character in $X$ is substituted. An inspection of~\eqref{eq:capped-td} reveals that $\TD^{\le K}_{v,s}(X,Y)$ does not depend on any character of $X$ outside $X_v$. This is consistent with \cref{def:relevant}, which reduces to $\qsX{i}{v}=\rsX{i}{v}=[-K\dd K]$ if $i\in I_v$ and $\qsX{i}{v}=\rsX{i}{v}=\emptyset$ otherwise. 
In other words, it is indeed true that updating $\Delta_{v,s}^{\le K}(X,Y)$ for $s\in \rsX{i}{v}\cap S_v$ is sufficient to maintain our invariants.
Thus, the {\sc Substitute-X} procedure is correct.

As for the running time analysis, observe that all nodes $v$ satisfying $i\in I_v$ lie on a single path $P$ connecting the root node $v_{root}$ to the $i$-th leaf node (from the left) in $T$.
Consider a single execution of procedure {\sc Substitution-$X$}
\cref{alg:dyn-sub:proc-subx} at node $v \in P$ (at other nodes, the procedure fails the test in \cref{alg:dyn-sub:terminate-x}).
If $v$ is a leaf, then the cost of executing \cref{alg:dyn-sub:leaf-x} is $\Ot(|S_v|)$.
Otherwise, due to $\qsX{i}{v}=\rsX{i}{v}=[-K\dd K]$, the cost of executing \cref{alg:dyn-sub:combine-x} is 
$\Ot(\sum_{h=0}^{b_v-1}(|S_v|+|S_{v_h}|))$.
Since $\Exp[|S_w|]\le 48K (b_w+1) \ln n\cdot \Exp[\beta_w^{-1}]=b \cdot (\log n)^{\Oh(\log_b n)}$ holds by \cref{lem:exp-precision} for every node $w$, this is $b^2\cdot (\log n)^{\Oh(\log_b n)}$ expected per node.
Across all $\Oh(\log n)$ nodes on the path $P$, the total expected time is still $b^2(\log{n})^{\Oh(\log_b{n})}$.
\end{proof}

\paragraph{Substitution in $Y$.}
We now consider the case when the $i$-th character in $Y$ is substituted. Here we observe that the set of nodes that need to be updated does not necessarily lie on a single path in $T$ as in the previous case. This is because of the asymmetric nature of the computation on the precision sampling tree w.r.t the strings $X$ and $Y$. However, we show that the number of quasi-relevant shifts is $\Oh(K\log_b n)$ (across all the nodes) 
and that the set of active nodes within the scope of a substitution of $Y[i]$ forms a connected subtree of $T$ of expected size $n^{\Oh(\epsilon)}$.

\begin{lemma}\label{lem:locality}
    Consider a node $v$ with $I_v = [i_v\dd j_v]$ of a precision sampling tree $T$ with $n$ leaves.
    Moreover, consider strings $X,X'\in \Sigma^n$ and $Y,Y'\in \Sigma^*$,
    as well as shifts $s,s'\in \mathbb{Z}$ and a threshold $t\in \mathbb{Z}_{\ge 0}$.
    If $X[i_v\dd j_v]=X'[i_v\dd j_v]$ and $Y[i_v+s-t\dd j_v+s+t]=Y'[i_v+s'-t\dd j_v+s'+t]$,
    then $\TD^{\le K}_{v,s'}(X,Y')\le t$ if and only if $\TD^{\le K}_{v,s}(X,Y)\le t$.
  \end{lemma}
    \begin{proof}
  Let us start with a proof that $\TD_{v,s}(X',Y')\le t$ if $\TD_{v,s}(X,Y)\le t$.
  We proceed by induction on the height of $v$.
  First, suppose that $v$ is a leaf, that is, $I_v = [i_v\dd i_v]$.
  Due to $t \ge 0$, we have $X[i_v]=X'[i_v]$ and $Y[i_v+s] = Y'[i_v+s]$.
  Consequently, $\TD_{v,s}(X',Y') = \ED(X'[i_v],Y'[{i_v+s}]) = \ED(X[i_v], \allowbreak Y[{i_v+s}]) = \TD_{v,s}(X,Y)$,
  Next, suppose that $v$ is an internal node with $b_v$ children $v_0,\ldots,v_{b_v-1}$.
  Recall that $\TD_{v,s}(X,Y)=\sum_{h=0}^{b_v-1} A_{h,s}$, where $A_{h,s}=\min_{s_h\in [-K\dd K]} \left(\TD_{v_h,s_h}(X,Y)+2|s-s_h|\right)$. The value $\TD_{v,s}(X,Y')$ can be expressed analogously using $A'_{h,s}$.
  Since $\TD_{v,s}(X,Y)\le t$, we must have $A_{h,s}\le t$ for every $h\in [0\dd b_v)$.
  Let us fix a shift $s_h$ such that $A_{h,s}=\TD_{v_h,s_h}(X,Y)+2|s-s_h|$ and denote $t_h = \TD_{v_h,s_h}(X,Y)$.
  Observe that $t_h \le t - 2|s-s_h| \le t - |s-s_h|$.
  Consequently, $[s_h-t_h\dd s_h+t_h]\sub [s-t\dd s+t]$ and, in particular, $Y[i_v+s_h-t_h\dd j_v+s_h+t_h]=Y'[i_v+s_h-t_h\dd j_v+s_h+t_h]$.
  Since $I_{v_h}\sub I_v$, we can use the inductive hypothesis, $\TD_{v_h,s_h}(X',Y')\le t_h = \TD_{v_h,s_h}(X,Y)$, which implies $A'_{h,s}\le A_{h,s}$. This inequality is valid for all $h\in[0\dd b_v)$, and therefore $\TD_{v,s}(X',Y')\le \sum_{h=0}^{b_v-1} A'_{h,s}\le  \sum_{h=0}^{b_v-1} A_{h,s} = \TD_{v,s}(X,Y)\le t$ holds as claimed.
  
  The proof of the converse implication is symmetric (swap the roles of $(Y',s')$ and $(Y,s)$).
  \end{proof}

\lemlocality
  \begin{proof}
  Let us start with a proof that $\TD^{\le K}_{v,s}(X',Y')\le \TD^{\le K}_{v,s}(X,Y)$.
  If $t = K-|s|$, then $\TD^{\le K}_{v,s}(X',Y') \le K-|s|=t=\TD^{\le K}_{v,s}(X,Y)$ holds by \cref{def:captd}.
  Otherwise, by \cref{lem:captd}, $t=\TD^{\le K}_{v,s}(X,Y)$, and \cref{lem:locality} implies $\TD^{\le K}_{v,s}(X',Y')\le t = \TD^{\le K}_{v,s}(X,Y)$.
  
  Thus, we have $t':=\TD^{\le K}_{v,s}(X',Y') \le \TD^{\le K}_{v,s}(X,Y)$.
  Consequently, $Y'[i_v+s-t'\dd j_v+s+t']=Y'[i_v+s-t'\dd j_v+s+t']$, so the claim above, with the roles of $(X,Y,t)$ and $(X',Y',t')$ swapped,
  implies $\TD^{\le K}_{v,s}(X,Y)\le \TD^{\le K}_{v,s}(X',Y')$, which completes the proof of the lemma.
  \end{proof}

\corlocality
\begin{proof}
    Denote $I_v = [i_v\dd j_v]$. Since $s\notin \rsY{i}{v}$ is not relevant, we have $i\notin [i_v+s+K_v\dd j_v+s+K_v]$.
    Due to $t:=\TD^{\le K}_{v,s}(X,Y)\le K_v$, this means $Y[i_v+s-t\dd j_v+s+t]$ has not been affected by the substitution, and thus, by \cref{lem:locality2},  $\TD^{\le K}_{v,s}(X,Y)$ is not affected either.
\end{proof}

We now use \cref{def:relevant} to characterize shifts quasi-relevant to a substitution at $Y[i]$.

\begin{lemma}\label{lem:subtree}
 The nodes within the scope of a substitution at $Y[i]$ form a connected subtree of $T$. 
 Moreover, the total number of quasi-relevant shifts across all these nodes is $\Oh(K\cdot d)$, where $d$ is the height of $T$.
\end{lemma}
\begin{proof}
Recall that $s\in \rsY{i}{v}$ if and only if $[i-s-K_v\dd i-s+K_v]\cap I_v \ne \emptyset$.
It is easy to observe that this property is monotone with respect to $I_v$. Consequently, $s\in \rsY{i}{v}$ implies that $s\in \rsY{i}{w}$
holds for every ancestor $w$ of $v$. This concludes the proof that the nodes within the scope of a substitution at $Y[i]$ form a connected subtree of $T$.

Next, observe that $s\in \rsY{i}{v}$ only if $[i-s-|I_v| \dd i-s+|I_v|]\cap I_v\ne \emptyset$,
and $s\in \qsY{i}{v}$ only if $[i-s-2|I_v| \dd i-s+2|I_v|]\cap I_v\ne \emptyset$.
Consequently, each node in scope contributes at most $5|I_v|$ quasi-relevant shifts. 
At the same time, each such node satisfies $[i-2K\dd i+2K]\cap I_v \ne \emptyset$ (because $|s|+K_v \le 2K$).
Let us fix a single level of the tree $T$. At this level, there are at most two nodes $v$ for which $I_v$ intersects $[i-2K\dd i+2K]$
yet is not contained in $[i-2K\dd i+2K]$. These two nodes contribute at most $4K+2$ quasi-relevant shifts.
As for the remaining nodes, the total size of their intervals $|I_v|$ does not exceed $4K+1$,
so their total contribution to the number of quasi-relevant shifts does not exceed $20K+5$.
Overall, we conclude that each level of $T$ contributes at most $24K+6$ quasi-relevant shifts.
Across all the $d$ levels, this gives a total of $\Oh(K\cdot d)$.
\end{proof}

We are now almost ready to bound the update time for substitution in $Y$. We need one additional lemma that bounds the number of active nodes in the scope of a given substitution.

\begin{lemma}\label{lem:active}
The expected number of active nodes in scope of a substitution at $Y[i]$ is $(\log{n})^{\Oh(\log_b{n})}$.
\end{lemma}
\begin{proof}
Recall from~\eqref{eq:active} (Section~\ref{sec:ako-static}) that the probability of $v$ being active does exceed $\frac{|I_v|}{K}\cdot (\log{n})^{\Oh(\log_b{n})}$,
which we can also bound as $\frac{K_v}{K}\cdot (\log{n})^{\Oh(\log_b{n})}$.
Moreover, note that each node $v$ in scope of a substitution at $Y[i]$ contributes at least $|\qsY{i}{v}|\ge K_v$ quasi-relevant shifts.
Consequently, for such node, the probability of being active does not exceed $\frac{|\qsY{i}{v}|}{K}\cdot (\log{n})^{\Oh(\log_b{n})}$.
By \cref{lem:subtree}, the total size $|\qsY{i}{v}|$ across all nodes is $\Oh(K\log_b n)$, so the expected number of active nodes in scope 
does not exceed $\Oh(\log_b n)\cdot (\log{n})^{\Oh(\log_b{n})} = (\log{n})^{\Oh(\log_b{n})}$.
\end{proof}

Finally, we analyze the update time for every substitution in $Y$. 

\begin{lemma}\label{lem:subY}
Upon any substitution in $Y$, we can update the values $\Delta_{v,s}^{\le K}(X,Y)$
across all active nodes $v\in T$ and shifts $s \in S_v$ in expected time $b^2(\log{n})^{\Oh(\log_b{n})}$.
\end{lemma}

\begin{proof}
Consider a substitution at $Y[i]$. By \cref{lem:subtree}, the procedure of {\sc Substitute-$Y$} visits all active nodes in scope of the substitution.
By \cref{cor:locality}, we value $\Delta_{v,s}^{\le K}(X,Y)$ needs to be updated only if $s\in \rsY{i}{v}$.
Moreover, since $\Delta_{v_h,s_h}^{\le K}(X,Y) \le K_v$, when updating the value $\Delta_{v,s}^{\le K}(X,Y)$,
we only need to inspect the values $\Delta_{v_h,s_h}^{\le K}(X,Y)$ with $|s-s_h|\le K_v$ (which implies $s_h \in \qsY{i}{v}$)
Thus, \cref{alg:dyn-sub:combine-y} correctly restricts the output of {\sc Combine} to $S_v\cap \rsY{i}{v}$,
and the inputs from the child $v_h$ to $S_{v_h}\cap \qsY{i}{v}$. Consequently, the recursive {\sc Substitute-$Y$} procedure correctly maintains 
the values $\Delta_{v,s}^{\le K}(X,Y)$.

As for the running time, observe that the {\sc Substitute-$Y$} procedure is called only for active nodes in scope of the substitution 
and for their children (which fail the text in \cref{alg:dyn-sub:terminate-y}). The total cost of the control instructions is proportional to $b\cdot (\log{n})^{\Oh(\log_b{n})}$ by \cref{lem:active}.
The cost of \cref{alg:dyn-sub:leaf-y} is $\Ot(|\rsY{i}{v}\cap S_v|)=\Ot(|\qsY{i}{v}\cap S_v|)$.
By \cref{lem:exp-precision}, we have $\Exp[|\qsY{i}{v}\cap S_v|] \le 48|\qsY{i}{v}|\cdot ({b_v+1})\ln n \cdot \Exp[\beta_v^{-1}]
= \frac{|\qsY{i}{v}|}{K}\cdot b (\log{n})^{\Oh(\log_b{n})}$. 
The cost of \cref{alg:dyn-sub:combine-y} is $\Ot(\sum_{h=0}^{b_v-1}(|\rsY{i}{v}\cap S_v|+|\qsY{i}{v}\cap S_{v_h}|))$,
which is $b^2\cdot \frac{|\qsY{i}{v}|}{K}\cdot (\log{n})^{\Oh(\log_b{n})}$ by the same reasoning.
The expected time spent at a fixed node $v$ is therefore 
$b^2\cdot \frac{|\qsY{i}{v}|}{K}\cdot (\log{n})^{\Oh(\log_b{n})}$. 
By \cref{lem:subtree}, this sums up to $b^2(\log{n})^{\Oh(\log_b{n})}$.
\end{proof}

We can summarize the main result of this section below.
\begin{theorem}\label{thm:sub}
For every $b\in [2\dd n]$, the dynamic $(k, \Theta(k  b \log_b n))$-gap edit distance problem can be maintained under substitutions in $b^2(\log{n})^{\Oh(\log_b{n})}$  expected amortized update time with probability at least $1-\frac{1}{n}$.
For $b =(\log n)^{\Theta(1/\epsilon)}$, the gap and the time become $(k, k\cdot (\log n)^{\Theta(1/\epsilon)})$ and $n^{\Oh(\epsilon)}$, respectively.
\end{theorem} 

\begin{proof}
Correctness of the approximation guarantee follows from the correctness analysis of Section~\ref{sec:ako-static} and from \cref{lem:subX,lem:subY}.
All the updates are successful with probability $1-\frac{1}{n^2}$  (Lemma~\ref{lem:delta}) on each substitution, and hence all updates over $n$ substitutions are successful with probability $1-\frac{1}{n}$. The expected update time per substitution follows from \cref{lem:subX,lem:subY}. Upon $n$ updates, the precision sampling tree can be rebuilt by paying $n b^2(\log{n})^{\Oh(\log_b{n})}$ time which amortized over $n$ updates is $b^2(\log{n})^{\Oh(\log_b{n})}$.
\end{proof}
\section{Dynamic Edits with Substitutions in X and Edits in Y}\label{sec:subXedY}
In this section, we consider a more complex dynamic scenario than the previous section where we allow all edits (insertions, deletions, and substitutions) in $Y$, but only allow substitutions in $X$. 

We now face a major bottleneck. Recall from Section~\ref{sec:subXsubY}, upon a substitution on $X$ or $Y$, we need to update $\Delta_{v,s}^{\le K}(X,Y)$ for all relevant shifts $s\in S_v$. In \cref{lem:subtree}, we show that number of such relevant shifts is small.
However, upon an insertion or deletion of $Y[y]$, all nodes may contain relevant shifts. In particular, if $v$ is a leaf with $I_v=\{i\}$,
when $\TD_{v,s}^{\le K}(X,Y)=\min(\ED(X[i],Y[{i+s}]))$ changes as long as $y < i+s$.  Therefore, running the algorithm from the previous section would require rebuilding the entire tree from scratch.

\paragraph*{Key Insight.} Our main idea here lies in the following observation. Given two strings $X$ and $Y$ with $\ED(X,Y) \leq k$, consider an optimal alignment of $Y$ with $X$. Removing at most $k$ characters corresponding to insertions and deletions and substituting according to the optimal alignment, we can decompose $X$ and $Y$ into at most $k$ disjoint parts such that each corresponding part is a perfect match. If we can find this decomposition under dynamic updates, then we can also maintain the edit distance. While we do not know how to find such a decomposition dynamically, we show it is possible to maintain a weaker decomposition as follows in $\tilde{O}(k)$ query time.
The string $Y$ can be decomposed into $2k$ disjoint parts, $Y=\bigodot_{i=0}^{2k} Y[y_i\dd y_{i+1})$ such that each even phrase $Y[y_i\dd y_{i+1})$  ($i$ is even) is matched with some substring $X[y_i+s_i\dd y_{i+1}+s_i)$ exactly with $s_i \in [-2K\dd2K]$ (it is possible $Y[y_i\dd y{i+1})$ is empty) and $|Y[y_i\dd y_{i+1})|=1$ for $i$ being odd (this captures the edits). Note that it is \emph{not} necessary that the substrings $X[y_i+s_i\dd y_{i+1}+s_i)$ are disjoint.

Now, consider the value $\TD^{\le K}_{v,s}(X,Y)$ for some node $v$ and shift $s$ such that \[I_v\sub (y_i-s+K_v\dd y_{i+1}-s-K_v)\] holds for some $i\in [0\dd 2k]$.
Then, if we denote $I_v = [i_v\dd j_v]$, we must have $[i_v+s-K_v\dd j_v+s+K_v)\sub (y_i\dd y_{i+1})$ and thus 
$Y[i_v+s-t\dd j_v+s+t) = X[i_v+s-t+s_i\dd j_v+s-t+s_i)$. Now, \cref{lem:locality2}
yields $\min(\TD_{v,s}(X,Y),K)=\min(\TD_{v,s+s_i}(X,X),K)$. 
Since the capping in $\TD^{\le K}_{v,s}(X,Y)$ depends on $s$ and since $|s_i|\le 2K$,
we can derive $\TD^{\le K}_{v,s}(X,Y) = \min(K-|s|, \TD_{v,s+s_i}^{\le 3K}(X,X))$; see \cref{lem:reuse}.

Therefore, if we can maintain $\Delta^{\le 3K}_{v,s+s_i}(X,X)$ efficiently (which provides an $(\alpha_v,\beta_v)$-approximation of $\TD_{v,s+s_i}^{\le 3K}(X,X)$), we can use that to quickly compute $\Delta^{\le K}_{v,s}(X,Y)$ just by copying the value and capping it by $K-|s|$.
Excitingly, we can use the results of \cref{sec:subXsubY} in a black-box fashion to maintain $\Delta^{\le 3K}_{v,s}(X,X)$ (since we only need to support substitutions in $X$).

On the other hand, if $I_v\sub [y_i-s+K_v\dd y_{i+1}-s+K_v)$ does not hold for any $i\in [0\dd 2k]$,
then we must have $I_v \cap [y_i-s-K_v\dd y_{i}-s+K_v]$ for some $i\in [0\dd 2k]$, which is equivalent to $s\in \rsY{y_i}{v}$ (see \cref{def:relevant}).
Due to \cref{lem:subtree}, the number of shifts $s$ satisfying this condition (for fixed $i$) is $\Oh(K\cdot d)$ and, if we focus on the sampled ones ($s\in S_v$) for which we actually need to compute $\Delta^{\le K}_{v,s}(X,Y)$, the number drops to $b^2 \cdot (\log n)^{\Oh(\log_b n)}$ in expectation.
Across all $2k+2$ values $i\in [0\dd 2k]$, this is $k\cdot b^2 \cdot (\log n)^{\Oh(\log_b n)}$ in expectation.
We can compute these values every $k'=\Theta(k)$ updates since a single update changes the edit distance by at most 1.
This way, the amortized expected update time will be $b^2 \cdot (\log n)^{\Oh(\log_b n)}$.

\subsection{Algorithm}

As before $\Delta^{\le K}_{v,s}(X,Y)$ represents an $(\alpha_v, \beta_v)$-approximation of $\TD_{v, s}^{\leq K}(X, Y)$. We use $\Delta^{\le 3K}_{v,s}(X,X)$ to represent $(\alpha_v, \beta_v)$-approximation of $\TD_{v, s}^{\leq 3K}(X, X)$. We maintain a precision sampling tree $T$ on $X$ along with the values $\Delta_{v,s}^{\le 3K}(X,X)$ for all $s\in S_v$ (which is now sampled from $[-3K\dd 3K]$).

Since the precision sampling tree is now built only on $X$, even though we allow insertions and deletions in $Y$ (along with substitutions in both $X$ and $Y$) the structure of the tree remains the same after an update. As argued above, the algorithm of \cref{sec:subXsubY} is used in a black-box fashion to maintain $\Delta_{v,s}^{\le 3K}(X,X)$. 
The challenge is to simulate compute the necessary values $\Delta_{v,s}^{\le K}(X,Y)$. 

We compute them from scratch at the beginning of every \emph{epoch} consisting of $k'=\Theta(k)$ updates (where $k'$ is small enough that we can accommodate additional additive error of $k'$). 
Even though our approximation of $\ED(X,Y)$ is maintained in a lazy fashion, we maintain up-to-date copies of $X$ and $Y$ in an instance of the following {\sc Dynamic Decomposition} algorithm.

\prpipm

Upon the beginning of a new epoch, we query the {\sc Dynamic Pattern Matching}.
If we return NO, then we are guaranteed that $\ED(X,Y)>k$, so we can forward the negative answer.
Otherwise, we obtain a partition $Y=\bigodot_{i=0}^{2k} Y[y_i\dd y_{i+1})$ and a sequence of shifts $(s_0,s_2,\ldots,s_{2k})\in [-k\dd k]$.
Let $P=\{y_0,\ldots,y_{2k}\}$ denote the set of phrase boundaries.

For each node $v$, we denote $\rsY{P}{v}=\bigcup_{y\in P} \rsY{y}{v}$ and $\qsY{P}{v}=\bigcup_{y\in P} \qsY{y}{v}$.
Moreover, we say that a node $v$ is relevant with respect to $P$ if $\rsY{P}{v}\ne \emptyset$.

Then, we recursively traverse the precision sampling tree, simulating the static algorithm.
The main difference is that we only aim at computing $\Delta_{v,s}^{\le K}(X,Y)$ for a subset of allowed shifts.
At the root, this subset contains only $\{0\}$. Otherwise, it is defined as $\qsY{P}{w}\cap S_v$, where $w$ is the parent of~$v$. 

For all shifts $s\in S\setminus \rsY{P}{v}$, we exploit the aforementioned observation that $\TD^{\le K}_{v,s}(X,Y)=\min(K-|s|,\TD^{\le 3K}_{v,s+s_i}(X,X))$ holds for an appropriate $i\in [0\dd 2k]$ (\crefrange{alg:dyn-edy:copy0}{alg:dyn-edy:copy}). We binary search for the appropriate $i$ and for a shift $\tilde{s}\in S_v$ that is closest to $s+s_i$. Then, we set $\Delta_{v,s}^{\le K}(X,Y)=\min(K-|s|,\Delta_{v,\tilde{s}}^{\le 3K}(X,X))$.

For the relevant shifts $s\in S\cap \rsY{P}{v}$, we proceed as in the static algorithm: If the node $v$ is not active, we set the values $\Delta_{v,s}^{\le K}(X,Y)$ to $0$. If it is a leaf, we compute $\Delta_{v,s}^{\le K}(X,Y)$.
Otherwise, we recurse, asking each child to compute the values $\Delta_{v_h,s'}^{\le K}(X,Y)$ for all quasi-relevant shifts $s'\in \qsY{P}{v}\cap S_{v_h}$.
Based on these answers, the {\sc Combine} procedure computes the sought values  $\Delta_{v,s}^{\le K}(X,Y)$ for $s\in S\cap \rsY{P}{v}$.

\begin{algorithm}[H]
\caption{Dyn-Edit-Y}\label{alg:dyn-edY}
\begin{algorithmic}[1]
\medskip
\Statex \textbf{Input:} 
Strings $X,Y$, a precision tree with value $\Delta_{v,s}^{\le 3K}(X,X)$ for $s\in S_v$. 
\Statex \textbf{Output:} $\Delta_{root, 0}(X,Y)$
\medskip
\Statex \hrule
\Statex \hrule
\State Compute $Y=\bigodot_{i=0}^{2k} Y[y_i\dd y_{i+1})$ and $S=\{s_0,s_2,\dd, s_{2k}\}$ \Comment{\cref{prp:ipm}}
\label{alg:dyn-edY:line:partition}
\If{Received NO instead} \Return ``$\ED[X,Y] > k$''
\EndIf
\State Let $P = \{y_0,\ldots,y_{2k}\}$
\label{alg:dyn-edY:line:edits}
\State \Call{SubXEditY}{$root,\{0\}, y_0,\ldots,y_{2k}, s_0,s_2,\ldots,s_{2k}$}
\Statex \hrule
\Procedure{SubXEditY}{$v, S,y_0,\ldots,y_{2k}, s_0,s_2,\ldots,s_{2k}$}
\label{alg:dyn-edY:proc-suby}
\For{$s \in S \setminus \rsY{P}{v}$}\label{alg:dyn-edy:copy0}
    \State Find $i\in [0\dd 2k]$ such that $I_v\sub (y_i-s+K_v\dd y_{i+1}-s-K_v)$
    \State Identify a shift $\tilde{s}\in S_v$ that is closest to $s+s_i$
    \State Set $\Delta^{\le K}_{v,s}(X,Y)=\min(K-|s|, \Delta^{\le 3K}_{v,\tilde{s}}(X,X))$\label{alg:dyn-edy:copy}
\EndFor
\If{$\rsY{i}{v}\ne \emptyset$}
\If{$|I_v| \ge \beta_v$}
    \State Set $\Delta^{\le K}_{v,s}(X,Y)=0$ for all $s\in S\cap \rsY{P}{v}$
\ElsIf{$v$ is a leaf}
    \State $\Delta_{v,s}^{\le K}(X,Y) = \TD_{v,s}^{\le K}(X,Y)$ for all $s \in S\cap \rsY{P}{v}$  \label{alg:dyn-edy:leaf}
\Else
\State Recursively run \Call{SubXEditY}{$v_h,\qsY{P}{v}\cap S_{v_h},y_0,\ldots,y_{2k}, s_0,s_2,\ldots,s_{2k}$} for $h\in [0\dd b_v)$
\State \Call{Combine}{$v,\rsY{P}{v}\cap S,\{(h,s_h,\Delta_{v_h,s_h}^{\le K}(X,Y)): h\in [0\dd b_v), s_h\in \qsY{P}{v}\cap S_{v_h}\}$}\label{alg:dyn-edy:combine}
\EndIf
\EndIf
\EndProcedure
\end{algorithmic}
\end{algorithm}

\subsubsection{Analysis}
We first establish the correctness of the procedure {\sc Dyn-Edit-Y} followed by the analysis of the update time. 
We need the following crucial lemma which basically states that copying from $\tilde{T}(X,X)$ in \cref{alg:dyn-edy:copy} of \cref{alg:dyn-edY} is correct. 

\begin{lemma}\label{lem:reuse}
Consider a node $v$ and a shift $s\in [-K\dd K]$ such that $I_v\sub (y_i-s+K_v\dd y_{i+1}-s-K_v)$
and $Y[y_i\dd y_{i+1})=X[y_i+s_i\dd y_{i+1}+s_i)$.
Then, \[\TD_{v,s}^{\le K}(X,Y) = \min(\TD_{v,s+s_i}^{\le 3K}(X,X),K-|s|).\] 
\end{lemma}
\begin{proof}
Denote $t = \TD_{v,s}^{\le K}(X,Y)$  and note that $t \le K_v$.
Thus, we also have $I_v \sub (y_i-s+t\dd y_{i+1}-s-t)$ or, equivalently, denoting $I_v= [i_v\dd j_v]$,
$(y_i\dd y_{i+1})\sub [i_v+s-t\dd j_v+s+t]$. 
Thus, we have $Y[i_v+s-t\dd j_v+s+t]=X[i_v+s+s_i-t\dd j_v+s+s_i+t]$.
By \cref{lem:locality}, this implies $\min(\TD_{v,s}(X,Y), K-|s|)=\min(\TD_{v,s+s_i}(X,X), K-|s|)$.
Since $\TD_{v,s}^{\leq K}(X,Y)=\min(\TD_{v,s}(X,Y), K-|s|)$ from \cref{lem:captd}, we have
\begin{align*}
    \TD_{v,s}^{\le K}(X,Y) 
    &= \min(\TD_{v,s}(X,Y), K-|s|) & \text{From \cref{lem:captd}}\\
    &= \min(\TD_{v,s+s_i}(X,X), K-|s|) & \text{From \cref{lem:locality}}\\
\end{align*}

Since $|s_i| \le 2K$, we have $3K - |s+s_i| \ge K-|s|$. Therefore, we have
\begin{align*}
     \min(\TD_{v,s+s_i}(X,X), K-|s|) 
    &= \min(\TD_{v,s+s_i}(X,X), K-|s|, 3K - |s+s_i|) \\
    &= \min(\TD_{v,s+s_i}^{\le 3K}(X,X), K-|s|) 
\end{align*}
The last equality again follows from \cref{lem:captd} as $|s+s_i| \leq 3K$ and
\[\TD_{v,s+s_i}^{\le 3K}(X,X) =\min(\TD_{v,s+s_i}(X,X), 3K - |s+s_i|).\qedhere\]
\end{proof} 

\begin{corollary}\label{cor:reuse}
    The value $\Delta_{v,s}^{\le K}(X,Y)$ set in \cref{alg:dyn-edy:copy} is, with high probability, an $(\alpha_v,2\beta_v)$-appro\-ximation of $\TD_{v,s}^{\le K}(X,Y)$.
\end{corollary}
\begin{proof}
    We repeat the reasoning within the proof of \cref{lem:delta}.
    Note that $|\TD_{v,s+s_i}^{\le 3K}(X,X)-\TD_{v,\tilde{s}}^{\le 3K}(X,X)|\le 2(b_v+1)|s+s_i-\tilde{s}|$
    and, with high probability, $|s+s_i-\tilde{s}|\le  \frac{\beta_{v}}{8(b_{v_h}+1)}$.
    Consequently, $|\TD_{v,s+s_i}^{\le 3K}(X,X)-\TD_{v,\tilde{s}}^{\le 3K}(X,X)|\le \frac14 \beta_v$.
    Since $\Delta_{v,\tilde{s}}^{\le 3K}(X,X)$ is an $(\alpha_v,\beta_v)$ approximation of $\TD_{v,\tilde{s}}^{\le 3K}(X,X)$,
    we conclude that it is also an $(\alpha_v,\beta_v +\alpha_v \cdot\frac14\beta_v)$ approximation of  $\TD_{v,s+s_i}^{\le 3K}(X,X)$.
    Due to $\alpha_v \le 2$, this means that  $\Delta_{v,\tilde{s}}^{\le 3K}(X,X)$ is an $(\alpha_v,2\beta_v)$-approximation of   $\TD_{v,s+s_i}^{\le 3K}(X,X)$, and thus $\Delta_{v,s}^{\le K}(X,Y)=\min(K-|s|,\Delta_{v,\tilde{s}}^{\le 3K}(X,X))$ is an $(\alpha_v,2\beta_v)$-approximation of 
    of $\TD_{v,s}^{\le K}(X,Y) = \min(\TD_{v,s+s_i}^{\le 3K}(X,X),K-|s|)$, where this equality follows from \cref{lem:reuse}.
\end{proof}

\begin{lemma}
    \label{lem:correctness}
   {\sc Dyn-Edit-Y} correctly solves the $(k, \Theta(k  b \log_b n))$-gap edit problem with probability $1-\frac{1}{n}$.
\end{lemma}
\begin{proof}
    Note that if $\ED(X,Y) \leq k$, then \cref{prp:ipm} will always find a decomposition. Hence, if it fails to do so, the algorithm correctly returns NO. 
    Using \cref{cor:locality} {\sc Dyn-Edit-Y} (\cref{alg:dyn-edY}) computes all $\Delta^{\le K}_{v,s}(X,Y)$ such that $s\in \qsY{P}{w}\cap S_v$,
    where $w$ is the parent of $v$; this follows from the fact that the recursive calls of {\sc Dyn-Edit-Y} set $S=\qsY{P}{v}\cap S_{v_h}$.
    If $s\notin \rsY{P}{v}$, the value $\Delta^{\le K}_{v,s}(X,Y)$ is computed in \cref{alg:dyn-edy:copy} and, by \cref{cor:reuse} is a $(\alpha_v,2\beta_v)$-approximation of $\TD^{\le K}_{v,s}(X,Y)$. In the remaining case of $s\in \rsY{P}{v}$, we simulate the static algorithm
    restricting the set of shifts at the children $v_h$ of $v$ to those in $\qsY{P}{v}\cap S_v$ (which is valid because any relevant shift is within a distance $K_v\le \TD^{\le K}_{v,s}(X,Y)$ from a quasi-relevant one).
    Consequently, as in the proof of \cref{lem:delta}, we can inductively argue that each value $\Delta^{\le K}_{v,s}(X,Y)$ computed by our algorithm is 
    an $(\alpha_v,2\beta_v)$-approximation of $\TD^{\le K}_{v,s}(X,Y)$.
    In particular, this is true for the value $\Delta^{\le K}_{root,0}(X,Y)$ computed at the topmost recursive call of {\sc Dyn-Edit-Y}.
    By \cref{lem:equiv-ed-td,lem:captd}, we conclude that {\sc Dyn-Edit-Y} correctly solves the $(k, \Theta(k  b \log_b n))$-gap edit distance problem with probability at least $1-\frac{1}{n}$.
\end{proof}

\paragraph{Running Time Analysis}
Maintaining the values $\Delta_{v,s}^{\le 3K}(X,X)$ and the data structure for {\sc Dynamic Pattern Matching} takes $b^2 (\log n)^{\Oh(\log_b n)}$-time 
due to \cref{lem:subX,lem:subY,prp:ipm}, respectively. 
It remains to bound the cost of {\sc Dyn-Edit-Y}, which is run every $k'=\Theta(k)$ updates.

\begin{lemma}\label{lem:runedY}
Every execution of the {\sc Dyn-Edit-Y} procedure takes $k b^2 (\log n)^{\Oh(\log_b n)}$ expected time.
\end{lemma}
\begin{proof}
The application of \cref{prp:ipm} to compute the decomposition $P$ costs $\Ot(k)$ time.
By \cref{lem:subtree}, the recursive procedure {\sc SubXEditY} visits all active nodes in scope of the $P$
as well as their children. For each such node, the set $S$ is set to $\qsY{P}{w}\cap S_v$, where $w$ is the parent of $v$;
the only exception is the root, when $S=\{0\}$.

The total cost of the control instructions is, in expectation, $k\cdot b\cdot (\log{n})^{\Oh(\log_b{n})}$;
this is every node within the scope of $P$ is, by definition, within the scope of a substitution at $Y[y_i]$ for some $i\in [0\dd 2k]$,
and the number of active nodes within the scope of a single substitution is $(\log{n})^{\Oh(\log_b{n})}$ by \cref{lem:active}.

The cost of handling each shift $s\notin \rsY{P}{v}$ is logarithmic since we can find the appropriate phrase $i\in [0\dd 2k]$ and shift $\tilde{s}$ by binary search. By \cref{lem:exp-precision}, we have $\Exp[|\qsY{P}{w}\cap S_v|] \le 48|\qsY{P}{w}|\cdot ({b_v+1})\ln n \cdot \Exp[\beta_v^{-1}]
= \frac{|\qsY{P}{w}|}{K}\cdot b (\log{n})^{\Oh(\log_b{n})}$. 
The cost of \cref{alg:dyn-edy:leaf} is $\Ot(|\rsY{P}{v}\cap S_v|)=\Ot(|\qsY{P}{v}\cap S_v|)$,
which is $\frac{|\qsY{P}{v}|}{K}\cdot b (\log{n})^{\Oh(\log_b{n})}$ by the same reasoning.
The cost of \cref{alg:dyn-edy:combine} is $\Ot(\sum_{h=0}^{b_v-1}(|\rsY{P}{v}\cap S_v|+|\qsY{P}{v}\cap S_{v_h}|))$,
which is $b^2\cdot \frac{|\qsY{P}{v}|}{K}\cdot (\log{n})^{\Oh(\log_b{n})}$, again due to \cref{lem:exp-precision}.

The expected time spent at a fixed node $v$ with parent $w$ is therefore 
$b\cdot \left( \frac{|\qsY{P}{w}|}{K} + b\cdot \frac{|\qsY{P}{v}|}{K}\right)\cdot (\log{n})^{\Oh(\log_b{n})}$. 
If we charge the first term to the parent $w$ (which can get charged by all of its $\Oh(b)$ children), the cost per node 
becomes $b^2 \cdot \frac{|\qsY{P}{v}|}{K}\cdot (\log{n})^{\Oh(\log_b{n})}$. 
Since, for each $v$, we have $\qsY{P}{v}=\bigcup_{i=0}^{2k}\qsY{y_i}{v}$, then the total number of quasi-relevant shifts (across all nodes)
is $\Oh(k\cdot K \cdot \log_b n)$. Consequently, the total expected time is  $kb^2(\log{n})^{\Oh(\log_b{n})}$.
\end{proof}

The following theorem summarizes the result of this section.
\begin{theorem}
For every $b\in [2\dd n]$, the dynamic $(k/2, \Theta(k  b \log_b n))$-gap edit distance problem can be maintained under substitutions in $X$ and edits in $Y$ in $b^2(\log{n})^{\Oh(\log_b{n})}$  expected amortized update time with probability at least $1-\frac{1}{n}$.
For $b =(\log n)^{\Theta(1/\epsilon)}$, the gap and the time become $(k, k\cdot (\log n)^{\Theta(1/\epsilon)})$ and $n^{\Oh(\epsilon)}$, respectively.
\end{theorem} 
\begin{proof}
Since we run {\sc Dyn-Edit-Y} in a lazy fashion after every $k/2$ updates, amortized over these $k/2$ updates the expected running time is $b^2(\log{n})^{\Oh(\log_b{n})}$ from \cref{lem:runedY}. The approximation bound follows from the correctness analysis using \cref{lem:correctness} and by noting that within each epoch, the true edit distance can only change by $k/2$.
\end{proof}

\subsection{Proof of \cref{prp:ipm}}\label{subsec:ipm}

The workhorse behind \cref{prp:ipm} is the following lemma:
\lemipm
\begin{proof}
Our algorithm maintains four components borrowed from~\cite{GKKLS18,CKW20,KK22}:
\begin{itemize}
  \item A data structure for Longest Common Extension queries in $S:=X\cdot Y$. These queries return the length of the longest common prefix of two fragments 
  $S[i\dd j)$ and $S[i'\dd j')$, identified by their endpoints. This component supports updates and queries in $\Ot(1)$ time~\cite{GKKLS18,CKW20}.
  \item A data structure for Internal Pattern Matching queries in $S:=X\cdot Y$. These queries, given two fragments $P=S[i\dd j)$ and $T=S[i'\dd j')$ of $S$ (identified by their endpoints) such that $|P| \le |T| < 2|P|$, return an arithmetic progression representing the exact occurrences of $P$ in $T$,
  that is, positions $t\in [0\dd |T|-|P|]$ with $P = T[t\dd t+|P|)$. This component supports updates and queries in $\Ot(1)$ time~\cite{CKW20}.
  \item For every integer \emph{level} $\tau\in [0\dd \lfloor{\log |X|}\rfloor]$, we maintain a partition of $X$ into blocks of length between $2^{\tau}$ and $2^{\tau+2}$.
  Every edit in $X$ affects exactly one block at a given level~$\tau$. Once the edit is performed, the block might become too long (and we split it into two halves) or too short (and then we merge it with its neighbor, possibly splitting the resulting block if is too long). Every such rebalancing operation can be charged to $\Omega(2^\tau)$ individual edits.
  \item  The dynamic suffix array of the entire string $X$ and of the concatenation of any up to 10 consecutive blocks at any level $\tau$. The dynamic suffix array of $X$, given a rank $r\in [1\dd |X|]$ identifies the $r$th lexicographically smallest suffix of $r$. The implementation of~\cite{KK22} supports updates and queries in $\Ot(1)$ time and, by the discussion above regarding block updates, the amortized update time among all dynamic suffix array instances is still $\Ot(1)$.
\end{itemize}

We proceed with the description of the query algorithm. If $i-k > |X|$, 
then $i-s > |X|$ holds for all $s\in [-k\dd k]$. Consequently, there is not shift $s\in [-k\dd k]$ such that $Y[i\dd j)\ne X[i+s\dd j+s)$ (because the latter fragment is always ill-defined), and thus the algorithm may return NO.

If $|Y[i\dd j)|>2k$ and $i-k \le |X|$, then we simply use an Internal Pattern Matching query with $P=Y[i\dd j)$ and $T=X[i-k\dd j+k)$ (trimmed appropriately if $i-k < 0$ or $j+k > |X|$). This is valid because $|T|\le |P|+2k < 2|P|$. Moreover, positions $t\in [|T|-|P|]$ with $P=T[t\dd t+|P|)$ correspond to shift $s\in [-k\dd k]$ such that $Y[i\dd j)=X[i-s\dd i+s)$. In this case, the query time is $\Ot(1)$, dominated by the IPM query.

The same strategy can be applied if $2k\ge |Y[i\dd j)|>\frac14k$ and $i-k \le |X|$. In that case, however, instead of a single IPM query, we perform up to $8$ queries with the pattern $P=Y[i\dd j)$ and the texts $T$ of length $2|P|$ covering $X[i-k\dd j+k)$ with overlaps of length at least $|P|-1$.

If $i-k \le |X|\le k$, then $i-2k \le 0 \le |X|\le j+2k$. Hence, it suffices to check if $Y[i\dd j)$ has any occurrence in $X$.
For this, we perform a binary search on top of the dynamic suffix array of $X$ to obtain the lexicographically smallest suffix $X[x\dd |X|)$ with $X[x\dd |X|)\succeq Y[i\dd j)$. If $Y[i\dd j)$ has any occurrence in $X$, then one such occurrence must start at position $x$.
At each step, we retrieve the $r$th lexicographically smallest suffix $X[x\dd |X|)$ of $X$ and the length $\ell$ of the longest common prefix of $X[x\dd |X|)$ and $Y[i\dd j)$. The lexicographic order of $X[x\dd |X|)$ and $Y[i\dd j)$ is determined by the lexicographic order between $X[x+\ell]$ and $Y[i+\ell]$
(with the corner cases of $x+\ell=|X|$ and $i+\ell=j$ handled appropriately).

It remains to consider the case when $|Y[i\dd j)| \le \frac14k$,  $|X| > k$, and $i-k \le |X|$. 
We pick a level $\tau = \lfloor{\log(\frac14k)}\rfloor$ and consider the level-$\tau$ blocks covering the fragment $X[i-k\dd j+k)$ (again, trimmed appropriately if $i-k < 0$ or $j+k>|X|$). Since the fragment is of length at most $\frac54k \le 5\cdot 2^{\tau+1}$, it is covered by at most 10 level-$\tau$ blocks.
Moreover, since the level-$\tau$ blocks are of length at most $2^{\tau+2} \le k$, the concatenation of these blocks is contained within $X[i-2k\dd j+2k)$.
Thus, it suffices to check if $Y[i\dd j)$ occurs at least once in the concatenation of these blocks. For these, we use the dynamic suffix array of this concatenation just like we used the dynamic suffix array of the whole string $X$ in the previous case.
\end{proof}

We are now ready to prove \cref{prp:ipm}, whose statement is repeated next for convenience.

\prpipm*
\begin{proof}
  We maintain $X$ and $Y$ using the dynamic algorithm of \cref{lem:ipm}. Below we describe an $\Ot(k)$-time query algorithm.
  If $|Y| > |X|-k$, we report that $\ED(X,Y) > k$. Otherwise, we process $Y$ from left to right in at most $2k+1$ rounds.
  Suppose we have already constructed a decomposition $Y=\bigodot_{i=0}^{j-1} Y[y_i\dd y_{i+1})$
  and a sequence of shifts $s_0,\ldots, s_{2\lceil{j/2}\rceil-2}\in [-2k\dd 2k]$ satisfying the following invariant:
  \begin{itemize}
    \item If $i$ is odd, then $|Y[y_i\dd y_{i+1})|=1$.
    \item If $i$ is even, then $Y[y_i\dd y_{i+1})=X[y_i+s_i\dd y_{i+1}+s_i)$ and $Y[y_i\dd y_{i+1}]\ne X[y_i+s\dd y_{i+1}+s]$ for any $s\in [-k\dd k]$.
  \end{itemize}
  If $j$ is odd, we simply set $y_{j+1}=y_j+1$ so that $|Y[y_j\dd y_{j+1})|=1$.
  If $j$ is even, we perform binary search on top of the query algorithm of \cref{lem:ipm}. Throughout the binary search, we maintain positions $y_{j+1}\le y'_{j+1}$ and a shift $s_j$ such that $Y[y_j\dd y_{j+1})=X[y_j+s_i\dd y_{j+1}+s_j)$ and $Y[y_j\dd y'_{j+1}]\ne X[y_j+s\dd y_{j+1}+s]$ for any $s\in [-k\dd k]$. Initially, we set $y_{j+1}=y_j$, $y'_{j+1}=|Y|$, and $s_j = \min(0, |X|-y_j)$.
  At each step of the binary search, we pick $\bar{y}_{j+1} = \lceil\frac12(y_{j+1}+y'_{j+1})\rceil$ and make a \cref{lem:ipm} query for $Y[y_j \dd \bar{y}_j)$. In case of a no answer, we update $y'_{j+1}:=\bar{y}_{j+1}-1$.
  Otherwise, we update $y_{j+1}:=\bar{y}_{j+1}$ and underlying witness shift $s$.
  In $\Oh(\log |Y|)$ iterations, this process converges to $y_{j+1}=y_{j+1}$, which satisfies our invariant.

  There are two ways that this process may terminate. If we have constructed $2k+1$ phrases without reaching the end of $Y$,
  we report that $\ED(X,Y) > k$. This is valid because any alignment of cost $k$ yields a decomposition of $Y$ into at most $k+1$ fragments matched exactly 
  and at most $k$ individual characters that have been inserted or substituted. Moreover, if a fragment $Y[y\dd y')$ is matched exactly, then $Y[y\dd y')=X[y+s\dd y'+s)$ for some $s\in [-k\dd k]$.
  The other possibility is that we reach the end of $Y$ having before we have generated $2k+1$ phrases. Since even-numbered phrases can be empty, we may assume that this happens after we have generated an even-number phrase (and cannot generate the next odd-numbered phrase of length 1).
  In that case, we increment the number of phrases in the decomposition by repeatedly splitting a non-empty even-number phrase into 
  an empty even-numbered numbered phrase, a length-1 odd-numbered phrase, and a shorter even-numbered phrase (sharing the original shift).
  Since $k\le |Y|$, we can continue this process until the number of phrases reaches $2k+1$.

  Our query algorithm takes $\Oh(k\log |Y|)$ time on top of making $\Oh(k\log |Y|)$ queries of \cref{lem:ipm}.
  Hence, the total query time is $\Ot(1)$.
\end{proof}

\section{Edits in both X and Y}\label{sec:edXedY}

The algorithm we develop here is similar to the one in \cref{sec:subXedY}.
The difference is that we now maintain approximately the values $\Delta^{\le 3K}_{v,s}(X,X)$ 
subject to insertions and deletions rather than substitutions.
\subsection{Maintaining $T(X,X)$ subject to insertions and deletions}

To ease notation, we consider the problem of maintaining the values  $\Delta^{\le K}_{v,s}(X,X)$.
Our final algorithm, described in \cref{sec:edXedY:alg}, will use this procedure with parameter $3K$ instead of $K$.
Moreover, in this section we distinguish $|X|$ from $n$, which will be a static upper bound on $|X|$;
occasional rebuilding of the entire data structure lets us assume that $n=\Theta(|X|)$.

Recall that, for each active node $v$ of the precision tree $T$, we maintain a set of allowed shifts $S_v\sub [-K\dd K]$
whose elements are sampled independently with rate $\Theta(\beta_v^{-1}\cdot b_v \log n)$.
For each shift $s\in S_v$, we need to maintain the value $\Delta^{\le K}_{v,s}(X,X)$,
which should be an $(\alpha_v,\beta_v)$-approximation of $\TD^{\le K}_{v,s}(X,X)$.

\paragraph{Balancing the precision sampling tree}
The first difficulty is that the tree $T$ cannot be static anymore: insertions and deletions require adding and removing leaves.
Multiple insertions can violate our upper bound $\Oh(b)$ on the node degrees, 
and thus occasional rebalancing is needed to simultaneously bound the degrees and the tree depth $d \le 2\log_b n$.
What makes rebalancing challenging in our setting is that the data we maintain has bi-directional dependencies.
On the one hand, the values $\TD^{\le K}_{v,s}(X,X)$ we approximate depend on the shape of the subtree rooted at $v$.
On the other hand, our additive approximation rates $\beta_v$ depend on the values $u_w$ sampled at the ancestors of $w$.
Thus, we cannot use rotations or other local rules.
Instead, we use weight-balanced B-trees~\cite{BDF05,AV03}, which support $\Oh(\log_b n)$-amortized-time updates and satisfy several useful properties:
\begin{itemize}
  \item All leaves of $T$ have the same depth $d=\log_b |T| + \Theta(1)$.
  \item The root has between $2$ and $4b$ children, all internal nodes have between $\frac14b$ and $4b$ children.
  \item The number of leaves in the subtree of a non-root node $v$ at depth $d_v$ is between $\frac12b^{d-d_v}$ and $2b^{d-d_v}$;
  in other words, $\frac12 b^{d-d_v}\le |I_v|\le 2b^{d-d_v}$.
  \item If a non-root node $v$ at depth $d_v$ is rebalanced, then $\Omega(b^{d-d_v})$ updates in its subtree are necessary before it is rebalanced again.
\end{itemize}
The last property is crucial for us: whenever a node $v$ is rebalanced, we will recompute all the data associated to all its descendants $w$ (including the approximation rates $\beta_w$, the precisions $u_w$, the sets $S_w$, and the values $\Delta^{\le K}_{w,s}(X,X)$). Moreover,
we will recompute the values $\Delta^{\le K}_{w,s}(X,X)$ for all ancestors $w$ of $v$.

We separate the description of the rebalancing subroutine from the description of the actual update algorithm. 
Technically, the weight-balanced B-tree $T_B$ might contain one \emph{invisible} leaf that is not present in our precision sampling tree $T$.
Upon insertion of a new character, we first insert an invisible leaf to $T_B$ (which may cause some rebalancing),
and only then insert this leaf in $T$. Similarly, upon deletion of a character, we first make the underlying leaf invisible, and only then remove it from $T_B$. Thus, to maintain the precision sampling tree $T$, we only need to implement rebalancing (which happens when the invisible leaf is added to or removed from $T_B$), as well as leaf insertion and deletion (which changes the string $X$ but preserves the shape of $T$, modulo the leaf in question).

As alluded earlier, rebalancing is quite straightforward.
\lemrebalance
\begin{proof}
 If $v$ is not active, there is nothing to do because rebalancing changes neither $|I_v|$ nor $\beta_v$.
 Otherwise, we call the \Call{AKO}{$v,\beta_v$} procedure of \cref{alg:ako} that recomputes all the necessary data for the descendants of $v$.
 As analyzed in the proof of \cref{lem:ako-time}, the cost of this call is $b^2 \cdot (\log n)^{\Oh(\log_b n)}$ per node,
 which is $|I_v|\cdot b^2 \cdot (\log n)^{\Oh(\log_b n)}$.
 Finally, we traverse the path from $v$ to the root and, at each node $w$ with children $w_0,\ldots,w_{b_{w}-1}$,
 run the \Call{Combine}{$w,S_w,\{(h,s_h,\Delta_{w_h,s_h}^{\le K}(X,X)):h\in [0\dd b_w), s_h\in S_{w_h}\}$} procedure to recompute the values $\Delta_{w,s}^{\le K}(X,X)$ for all $s\in S_w$. This call is a part of the regular \Call{AKO}{$w,\beta_w$} procedure, so its running time is also $b^2 \cdot (\log n)^{\Oh(\log_b n)}$ per node and $\log_b n \cdot b^2 \cdot (\log n)^{\Oh(\log_b n)}= (\log n)^{\Oh(\log_b n)}$ in total.

 The correctness follows from the fact that $\TD^{\le K}_{w,s}(X,X)$ does not depend on the shape of the precision sampling tree outside the subtree of $w$.
\end{proof}

\paragraph{Leaf insertions and deletions}

From now on, we assume that an insertion or a deletion of a single character $X[x]$ does not alter the precision sampling tree except that the underlying leaf is added or removed.
Our strategy is essentially the same as in \cref{sec:subXsubY}: for each node $v$, identify the set of \emph{relevant} shifts $\ri{x}{v}$,
for each shift $s\in S_v\cap \ri{x}{v}$, update the underlying values $\Delta^{\le K}_{v,s}(X,X)$ based on the values $\Delta^{\le K}_{v_h,s_h}(X,X)$
obtained from the children $v_h$, with the shift $s_h$ restricted to $s_h \in S_{v_h}\cap \qi{x}{v}$, where $\qi{x}{v}$ consists of quasi-relevant shifts that are close to relevant ones.

Unfortunately, this approach would yield $\Omega(K^2)$ relevant shifts: for every leaf $v$ with $I_v=\{i\}$, the value $\TD^{\le K}_{v,s}(X,X)$
changes whenever $x\in (i_v \dd i_v+s)$. 
This is simply because if $X$ has been obtained by inserting $X[x]$ to  $X'=X[0\dd x)\cdot X(x\dd |X|)$,
then the character $X[i_v+s]$ is derived from $X'[i_v+s-1]$ rather than $X'[i_v+s]$. In other words, we then have  $\TD^{\le K}_{v,s}(X,X)=\TD^{\le K}_{v,s-1}(X',X')$. Thus, instead of recomputing $\Delta^{\le K}_{v,s}(X,X)$ we should rather observe that we can copy $\Delta^{\le K}_{v,s-1}(X',X')$ instead.
Consequently, our notion of relevant shifts does not cover the cases when we are guaranteed that $\TD^{\le K}_{v,s}(X,X)=\TD^{\le K}_{v,s\pm 1}(X',X')$.

\newrelevance

Our strategy of handling the shifts $s$ such that $\TD^{\le K}_{v,s}(X,X)=\TD^{\le K}_{v,s\pm 1}(X',X')$
is to associate labels with characters of $X$ that stay intact even while insertions and deletions change the character's indices. 
In particular, the shift $s$ at node $v$ should be identified with the label of $X[i_v+s]$. 
Since $i_v+s$ might be an out-of-bound index, we actually assign labels to characters of a longer string $\Xd = \$^K\cdot X \cdot \$^K$. 
The labels can be stored in a balanced binary search tree built on top of $\Xd$ so that the conversion between labels and (current) indices takes logarithmic time. This strategy has already been used in~\cite{KK22}, and thus we repurpose their formalism.

For an alphabet $\Sigma$, we say that a \emph{labelled string} over $\Sigma$ is a string over $\Sigma \times \Zz$.
For $c:=(a,\ell)\in \Sigma\times \Zz$, we say that $\mathrm{val}(c):=a$ is the \emph{value} of $c$ and $\mathrm{label}(c):=\ell$ is the \emph{label}
of $c$. For a labelled string $S\in (\Sigma\times\Zz)^*$, we define the set $\Lbl(S)=\{\mathrm{label}(S[i]): i\in [0\dd |S|)\}$.
We say that $S$ is uniquely labelled if $|\Lbl(S)|=\|S\|$; equivalently, for each label $\ell\in \Lbl(S)$,
there exists a unique position $i\in [0\dd |S|)$ such that $\ell=\mathrm{label}(S[i])$.
\begin{lemma}[{Labelled String Maintenance; \cite[ Lemma 10.1]{KK22}}]\label{cor:labelled}
  There is a data structure maintaining a uniquely labelled string $S\in (\Sigma\times \Zz)^*$ using the following interface:
  \begin{description}
    \item[$\ins(S,i,a,\ell)$:] Given $i\in [0\dd |S|]$, $a\in \Sigma$, and $\ell\in \Zz \setminus \Lbl(S)$, set $S:= S[0\dd i)\cdot (a,\ell)\cdot S[i\dd |S|)$.
    \item[$\del(S,i)$:] Given $i\in [0\dd |S|)$, set $S:= S[0\dd i)\cdot S(i\dd |S|)$.
    \item[$\lbl(S,i)$:] Given $i\in [0\dd |S|)$, return $\mathrm{label}(S[i])$.
    \item[$\val(S,i)$:] Given $i\in [0\dd |S|)$, return $\mathrm{val}(S[i])$.
    \item[$\unlbl(S,\ell)$:] Given $\ell\in \Lbl(S)$, return $i\in [0\dd |S|)$ such that $\ell=\mathrm{label}(S[i])$.
  \end{description}
  Assuming that the labels belong to $[0\dd L)$, each operation can be implemented in $\Oh(\log L)$ time.
\end{lemma}

Our data structure uses \cref{cor:labelled} to maintain $\Xd = \$^{K}\cdot X \cdot \$^{K}$ as a uniquely labelled string (upon insertions, we assign a fresh label to the inserted character). This string extends $X$ with sentinel characters so that we can assume that the label of $X[i]$ is defined as long as $i\in [-K\dd |X|+K)$.
For a node $v\in T$ with $I_v = [i_v\dd j_v]$, we shall store $\lft_v = \lbl(\Xd,i_v+K)$ and $\rgt_v = \lbl(\Xd,j_v+K)$.%
\footnote{Observe that $X=\Xd[K\dd K+|X|)$, so $\Xd[i_v+K\dd j_v+K]$ corresponds to $X_v = X[i_v\dd j_v]$.}
Observe that, unlike indices $i_v,j_v$, which are shifted by edits at position $x<i_v$, the labels $\lft_v$ and $\rgt_v$ generally stay intact (unless the leaves in the subtree of $v$ are inserted or deleted). For a similar reason, we use labels to represent shifts. Specifically, for a node $v$ and label $\ell\in \Lbl(\Xd)$, we define $\shf_v(\Xd,\ell) = \unlbl(\Xd,\ell)-\unlbl(\Xd,\lft_v)$. 
Furthermore, if $s:= \shf_v(\Xd,\ell)\in [-K\dd K]$, we denote $\TD_{v,\ell}^{\le K}(X,X):=\TD_{v,s}^{\le K}(X,X)$.
We also say that $(v,\ell)$ is a \emph{relevant} for an insertion at $X[x]$ whenever $(v,\shf_v(\Xd,\ell))$ is relevant for that insertion.

The remaining challenge is that if an insertion at $X[x]$ transforms $X'$ to $X$ and $\shf_v(\Xd,\ell)\ne \shf_v(\Xd',\ell)$, then we should also reflect the change in the contents of the set $S_v$ of allowed shifts. In particular, our algorithm now maintains a set $L_v = \{\ell\in \Lbl(\Xd) : \shf_v(\Xd,\ell)\in S_v\}$, which is more `stable' than $S_v$ itself. The only issue is that, unlike $[-K\dd K]$, which is static, the set $\{\ell\in \Lbl(\Xd) : \shf_v(\Xd,\ell)\in [-K\dd K]\}$ may change. Fortunately, subject to an insertion at $X[x]$, it may change only for nodes in the scope of the insertion,
with at most one element entering this set (the label of the newly inserted character) and one leaving it (which used to correspond to $\pm K$ but now corresponds to $\pm (K+1)$). Thus, our algorithm can keep track of these changes, and, whenever a new label is inserted, we can insert it to $L_v$ with probability $\max(1, 48\beta_v^{-1}(b_v+1)\ln n)$ so that the set $S_v$ is still distributed as in \cite{AKO10}.

We are now ready to formally describe our insertion algorithm.
The algorithm for deletions is symmetric; in particular, at each node $v$, the set of \emph{irrelevant labels} is the same as for insertions.

\begin{algorithm}[t]
\caption{Dynamic-Insertion-X} \label{alg:dyn-insert}
\begin{algorithmic}[1]
\medskip
\Statex \textbf{Input:}
The location $x$ of the insertion, a node $v\in T$.
\Statex \textbf{Output:} The updated sets $L_v$ and values $\Delta_{v,\ell}^{\le K}(X,X)$ for $\ell\in L_v$.
\medskip
\Statex \hrule
\Procedure{Insert-$X$}{$x, v$}\label{alg:dyn-ins}
\If{$|I_v| > \beta_v$ \textbf{and} $\ri{x}{v}\ne \emptyset$}\label{alg:dyn-insert:terminate-y}
\State $i_v := \unlbl(\Xd, \lft_v)-K$
\If{$x\in [i_v-K\dd i_v+K]$}
  \State With probability $\max(1, 48\beta_v^{-1}(b+1)\ln n)$, insert $\lbl(\Xd,x+K)$ to $L_v$
  \State Remove $\lbl(\Xd,i_v-1)$ and $\lbl(\Xd,i_v+2K+1)$ from $L_v$
\EndIf
\If{$v$ is a leaf}
    \State\Return $\Delta_{v,s}^{\le K}(X,X) = \ED(X_v, X_{v, s})$ for all $s \in S_{v}\cap \ri{x}{v}$ \label{alg:dyn-insert:leaf-x}
\EndIf
\State Recursively run \Call{Insert-$X$}{$x,v_h$} for $h\in [0\dd b_v)$
\State \Return \Call{Combine}{$v,\ri{x}{v}\cap S_v,\{(h,s_h,\Delta_{v_h,s_h}^{\le K}(X,X)): h\in [0\dd b_v), s_h\in \qi{x}{v}\cap S_{v_h}\}$}\label{alg:dyn-insert:combine-x}
\EndIf
\EndProcedure
\end{algorithmic}
\end{algorithm}

\paragraph*{Analysis}

\begin{lemma}\label{lem:locality_edit}
  Consider a string $X\in \Sigma^+$, a position $x\in [0\dd |X|)$, and a string $X'=X[0\dd x)\cdot X(x\dd |X|)$.
  Consider a node $v\in T$ with $x\notin I_v = [i_v\dd j_v]$ as well as a shift $s\in [-K\dd K]$ such that $\TD^{\le K}_{v,s}(X,X) = t$.
  \begin{enumerate}
    \item If $x > \max(j_v,j_v+s+t)$, then $\TD^{\le K}_{v,s}(X,X)=\TD^{\le K}_{v,s}(X',X')$.
    \item If $i_v+s-t > x > j_v$ and $|s|+t<K$, then $\TD^{\le K}_{v,s}(X,X)=\TD^{\le K}_{v,s-1}(X',X')$.
    \item If $j_v+s+t < x < i_v$ and $|s|+t<K$, then $\TD^{\le K}_{v,s}(X,X)=\TD^{\le K}_{v,s+1}(X',X')$.
    \item If $x < \min(i_v,i_v+s-t)$, then $\TD^{\le K}_{v,s}(X,X)=\TD^{\le K}_{v,s}(X',X')$.
  \end{enumerate}
\end{lemma}
\begin{proof}
By symmetry, it suffices to consider the two cases involving $x>j_v$.

\paragraph*{$x > \max(j_v,j_v+s+t)$.}
We proceed by induction on the height of $v$.
If $v$ is a leaf, then $\TD^{\le K}_{v,s}(X,X)=\min(\ED(X[i_v],X[i_v+s]),K-|s|)=\min(\ED(X'[i_v],X'[i_v+s]),K-|s|)=\TD^{\le K}_{v,s}(X',X')$.
Next, suppose that $v$ is an internal node with $b_v$ children $v_0,\ldots,v_{b_v-1}$.
Recall that $\TD^{\le K}_{v,s}(X,X)=\min\left(K-|s|,\sum_{h=0}^{b_v-1} A_{h,s}\right)$, where $A_{h,s}=\min_{s_h\in [-K\dd K]} \left(\TD^{\le K}_{v_h,s_h}(X,X)+2|s-s_h|\right)$. The value $\TD^{\le K}_{v,s}(X',X')$ can be expressed analogously using $A'_{h,s}$.

Let us start by proving $\TD^{\le K}_{v,s}(X',X')\le \TD^{\le K}_{v,s}(X,X)$.
If $\TD^{\le K}_{v,s}(X,X)=K-|s|$, then $\TD^{\le K}_{v,s}(X',X')\le K-|s|=\TD^{\le K}_{v,s}(X,X)$ holds trivially.
Otherwise, we must have $A_{h,s}\le t$ for every $h\in [0\dd b_v)$.
Let us fix a shift $s_h$ such that $A_{h,s}=\TD^{\le K}_{v_h,s_h}(X,X)+2|s-s_h|$ and denote $t_h = \TD^{\le K}_{v_h,s_h}(X,X)$.
Observe that $t_h \le t - 2|s-s_h| \le t - |s-s_h|$.
Consequently, $x > j_v+s_h+t_h$.
Since $I_{v_h}\sub I_v$, we can use the inductive hypothesis, $\TD^{\le K}_{v_h,s_h}(X',X')=\TD^{\le K}_{v_h,s_h}(X,X)$, which implies $A'_{h,s}\le A_{h,s}$. This inequality is valid for all $h\in[0\dd b_v)$, and therefore $\TD^{\le K}_{v,s}(X',X')\le \sum_{h=0}^{b_v-1} A'_{h,s}\le  \sum_{h=0}^{b_v-1} A_{h,s} = \TD^{\le K}_{v,s}(X,X)$ holds as claimed.

It remains to show $\TD^{\le K}_{v,s}(X,X)\le \TD^{\le K}_{v,s}(X',X')$; note that we already have $\TD^{\le K}_{v,s}(X',X')\le t$.
If $\TD^{\le K}_{v,s}(X',X')=K-|s|$, then $\TD^{\le K}_{v,s}(X,X)\le K-|s|=\TD^{\le K}_{v,s}(X',X')$ holds trivially.
Otherwise, we must have $A'_{h,s}\le t$ for every $h\in [0\dd b_v)$.
Let us fix a shift $s'_h$ such that $A'_{h,s}=\TD^{\le K}_{v_h,s'_h}(X',X')+2|s-s'_h|$ and denote $t'_h = \TD^{\le K}_{v_h,s'_h}(X',X')$.
Observe that $t'_h \le t - 2|s-s'_h| \le t - |s-s'_h|$.
Consequently, $x > j_v+s'_h+t'_h$.
Since $I_{v_h}\sub I_v$, we can use the inductive hypothesis, $\TD^{\le K}_{v_h,s'_h}(X,X)=\TD^{\le K}_{v_h,s'_h}(X',X')$, which implies $A_{h,s}\le A'_{h,s}$. This inequality is valid for all $h\in[0\dd b_v)$, and therefore $\TD^{\le K}_{v,s}(X,X)\le \sum_{h=0}^{b_v-1} A_{h,s}\le  \sum_{h=0}^{b_v-1} A'_{h,s} = \TD^{\le K}_{v,s}(X',X')$ holds as claimed.

\paragraph*{$i_v+s-t > x > j_v$ and $|s|+t<K$.}
We proceed by induction on the height of $v$.
If $v$ is a leaf, then $\TD^{\le K}_{v,s}(X,X)=\min(\ED(X[i_v],X[i_v+s]),K-s)=\ED(X[i_v],X[i_v+s])=\ED(X'[i_v],X'[i_v+s-1])=\min(\ED(X'[i_v],X'[i_v+s-1]),K-s+1) = \TD^{\le K}_{v,s-1}(X',X')$.
Next, suppose that $v$ is an internal node with $b_v$ children $v_0,\ldots,v_{b_v-1}$.
Recall that $\TD^{\le K}_{v,s}(X,X)=\min\left(K-|s|,\sum_{h=0}^{b_v-1} A_{h,s}\right)$, where $A_{h,s}=\min_{s_h\in [-K\dd K]} \left(\TD^{\le K}_{v_h,s_h}(X,X)+2|s-s_h|\right)$. The value $\TD^{\le K}_{v,s-1}(X',X')$ can be expressed analogously using $A'_{h,s-1}$.

Let us start by proving $\TD^{\le K}_{v,s-1}(X',X')\le \TD^{\le K}_{v,s}(X,X)$.
Since $t < |K|-s$, we must $A_{h,s}\le t$ for every $h\in [0\dd b_v)$.
Let us fix a shift $s_h$ such that $A_{h,s}=\TD^{\le K}_{v_h,s_h}(X,X)+2|s-s_h|$ and denote $t_h = \TD^{\le K}_{v_h,s_h}(X,X)$.
Observe that $t_h \le t - 2|s-s_h| \le t - |s-s_h|$.
Consequently, $i_v+s_h-t_h > x$ and $|s_h|+t_h < K$. 
Since $I_{v_h}\sub I_v$, we can use the inductive hypothesis, $\TD^{\le K}_{v_h,s_h-1}(X',X')=\TD^{\le K}_{v_h,s_h}(X,X)$, which implies $A'_{h,s-1}\le A_{h,s}$. This inequality is valid for all $h\in[0\dd b_v)$, and therefore $\TD^{\le K}_{v,s-1}(X',X')\le \sum_{h=0}^{b_v-1} A'_{h,s-1}\le  \sum_{h=0}^{b_v-1} A_{h,s} = \TD^{\le K}_{v,s}(X,X)$ holds as claimed.

It remains to show $\TD^{\le K}_{v,s}(X,X)\le \TD^{\le K}_{v,s-1}(X',X')$; note that we already have $\TD^{\le K}_{v,s-1}(X',X')\le t$.
If $\TD^{\le K}_{v,s-1}(X',X')=K-s+1$, then $\TD^{\le K}_{v,s}(X,X)\le K-s < K-s+1=\TD^{\le K}_{v,s-1}(X',X')$ holds trivially.
Otherwise, we must have $A'_{h,s-1}\le t$ for every $h\in [0\dd b_v)$.
Let us fix a shift $s'_h$ such that $A'_{h,s-1}=\TD^{\le K}_{v_h,s'_h-1}(X',X')+2|s-s'_h|$ and denote $t'_h = \TD^{\le K}_{v_h,s'_h-1}(X',X')$.
Observe that $t'_h \le t - 2|s-s'_h| \le t - |s-s'_h|$.
Consequently, $i_v+s'_h-t'_h > x$ and $|s'_h|+t'_h < K$.
Since $I_{v_h}\sub I_v$, we can use the inductive hypothesis, $\TD^{\le K}_{v_h,s'_h}(X,X)=\TD^{\le K}_{v_h,s'_h-1}(X',X')$, which implies $A_{h,s}\le A'_{h,s-1}$. This inequality is valid for all $h\in[0\dd b_v)$, and therefore $\TD^{\le K}_{v,s}(X,X)\le \sum_{h=0}^{b_v-1} A_{h,s}\le  \sum_{h=0}^{b_v-1} A'_{h,s-1} = \TD^{\le K}_{v,s-1}(X',X')$ holds as claimed.
\end{proof}

\corlocalityedit
\begin{proof}
  Let $I_v = [i_v\dd j_v]$ and note that $t := \TD_{v,\ell}^{\le K}(X,X) \le  K_v$.
  By symmetry, we assume without loss of generality that $x > j_v$.
  If $x > i_v + s$, then $\shf_v(\Xd',\ell)=\unlbl(\Xd',\ell)-\unlbl(\Xd',\lft_v)=(i_v+s)-i_v = s$.
  Due $x\notin [i_v+s-K_v\dd j_v+s+K_v]$, we further have 
  $x > j_v + s + t$, so \cref{lem:locality_edit} implies $\TD_{v,\ell}^{\le K}(X,X)=\TD_{v,s}^{\le K}(X,X)=\TD_{v,s}^{\le K}(X',X')=\TD_{v,\ell}^{\le K}(X',X')$.
  If $x < i_v + s$, on the other hand, then $\shf_v(\Xd',\ell)=\unlbl(\Xd',\ell)-\unlbl(\Xd',\lft_v)= (i_v+s-1)-i_v = s-1$.
  Since $x\notin [i_v+s-K_v\dd j_v+s+K_v]$, we further have 
  $x < i_v + s - t$. 
  In particular, $x\in [i_v \dd j_v+s]$, which additionally implies $|s|< K-|I_v| \le K-t$, and thus $|s|+t < K$.
  Consequently, \cref{lem:locality_edit} implies $\TD_{v,\ell}^{\le K}(X,X)=\TD_{v,s}^{\le K}(X,X)=\TD_{v,s-1}^{\le K}(X',X')=\TD_{v,\ell}^{\le K}(X',X')$.
\end{proof}

\begin{lemma}\label{lem:subtree-edit}
  The nodes within the scope of an insertion at $Y[i]$ form a connected subtree of $T$. 
  Moreover, the total number of quasi-relevant shifts across all these nodes is $\Oh(K\cdot d)$, where $d$ is the height of $T$.
 \end{lemma}
 \begin{proof}
  As for the connectivity, it suffices to observe that if $s\in \ri{x}{v}$, then $s\in \ri{x}{w}$ holds for every ancestor of $w$ (because the definition of $\ri{x}{v}$ is monotone with respect to $I_v$).

  Next, observe that $s\in \ri{x}{v}$ only if $[x-2K\dd x+2K]\cap I_v \ne \emptyset$ (because $|s|+K_v \le 2K$).
  Let us fix a single level of the tree $T$. At this level, there are at most two nodes $v$ for which $I_v$ intersects $[i-2K\dd i+2K]$
  yet is not contained in $[x-2K\dd x+2K]$. These two nodes contribute at most $4K+2$ quasi-relevant shifts.
  Moreover, there is at most one node for which $x\in I_v$, and it contributes at most $2K+1$ quasi-relevant shifts.
  As for the remaining nodes, observe that $s\in \ri{x}{v}$ only if $[i-s-|I_v| \dd i-s+|I_v|]\cap I_v\ne \emptyset$
  or $|s|>K-|I_v|$. Consequently, $s\in \qi{x}{v}$ only if $[i-s-2|I_v| \dd i-s+2|I_v|]\cap I_v\ne \emptyset$
  or $|s|>K-2|I_v|$ Thus, each node $v$ (such that $x\notin I_v$) contributes at most $9|I_v|$ quasi-relevant shifts.
  Across all nodes $v$ (at the considered level) with $I_v \sub [x-2K\dd x+2K]$, this yields $36K+9$ quasi-relevant shifts.
 Overall, we conclude that each level of $T$ contributes at most $42K+12$ quasi-relevant shifts.
 Across all the $d$ levels, this gives a total of $\Oh(K\cdot d)$.
 \end{proof}

 \begin{lemma}\label{lem:active-edit}
 The expected number of active nodes in the scope of an insertion at $X[x]$ is $(\log{n})^{\Oh(\log_b{n})}$.
 \end{lemma}
 \begin{proof}
 Recall from~\eqref{eq:active} that the probability of $v$ being active does not exceed $\frac{|I_v|}{K}\cdot (\log{n})^{\Oh(\log_b{n})}$,
 which we can also bound as $\frac{\min(|I_v|,K)}{K}\cdot (\log{n})^{\Oh(\log_b{n})}$.
 Moreover, note that each node $v$ in the scope of a insertion at $X[i]$ contributes at least $|\qi{x}{v}|\ge \min(|I_v|,K)$ quasi-relevant shifts.
 Consequently, for such a node, the probability of being active does not exceed $\frac{|\qi{x}{v}|}{K}\cdot (\log{n})^{\Oh(\log_b{n})}$.
 By \cref{lem:subtree-edit}, the total size $|\qsY{x}{v}|$ across all nodes is $\Oh(K\log_b n)$, so the expected number of active nodes in scope 
 does not exceed $\Oh(\log_b n)\cdot (\log{n})^{\Oh(\log_b{n})} = (\log{n})^{\Oh(\log_b{n})}$.
 \end{proof}

 \begin{lemma}\label{lem:insert}
 Upon any insertion in $X$, we can update the values $\Delta_{v,\ell}^{\le K}(X,X)$
 across all active nodes $v\in T$ and labels $\ell \in L_v$ in expected time $b^2(\log{n})^{\Oh(\log_b{n})}$.
 \end{lemma}
 \begin{proof}
 Consider an insertion at $X[i]$. By \cref{lem:subtree-edit}, the procedure {\sc Insert-$X$} visits all active nodes in scope of the substitution.
 By \cref{cor:locality-edit}, we value $\Delta_{v,\ell}^{\le K}(X,X)$ needs to be updated only if $s=\shf(\Xd,\ell)\in \ri{x}{v}$.
 Moreover, since $\Delta_{v_h,s_h}^{\le K}(X,X) \le K_v$, when updating the value $\Delta_{v,s}^{\le K}(X,X)$,
 we only need to inspect the values $\Delta_{v_h,s_h}^{\le K}(X,X)$ with $|s-s_h|\le K_v$ (which implies $s_h \in \qi{x}{v}$)
 Thus, \cref{alg:dyn-insert:combine-x} correctly restricts the output of {\sc Combine} to $S_v\cap \ri{x}{v}$,
 and the inputs from the child $v_h$ to $S_{v_h}\cap \qsY{i}{v}$. Consequently, the recursive {\sc Insert-$X$} procedure correctly maintains 
 the values $\Delta_{v,\ell}^{\le K}(X,X)$.
 
 As for the running time, observe that the {\sc Insert-$X$} procedure is called only for active nodes in the scope of the substitution 
 and for their children (which fail the text in \cref{alg:dyn-insert:terminate-y}). The total cost of the control instructions is proportional to $b\cdot (\log{n})^{\Oh(\log_b{n})}$ by \cref{lem:active-edit}.
 The cost of \cref{alg:dyn-insert:leaf-x} is $\Ot(|\rsY{i}{v}\cap S_v|)=\Ot(|\qsY{i}{v}\cap S_v|)$.
 By \cref{lem:exp-precision}, we have $\Exp[|\qsY{i}{v}\cap S_v|] \le 48|\qsY{i}{v}|\cdot ({b_v+1})\ln n \cdot \Exp[\beta_v^{-1}]
 = \frac{|\qsY{i}{v}|}{K}\cdot b (\log{n})^{\Oh(\log_b{n})}$. 
 The cost of \cref{alg:dyn-insert:combine-x} is $\Ot(\sum_{h=0}^{b_v-1}(|\rsY{i}{v}\cap S_v|+|\qsY{i}{v}\cap S_{v_h}|))$,
 which is $b^2\cdot \frac{|\qsY{i}{v}|}{K}\cdot (\log{n})^{\Oh(\log_b{n})}$ by the same reasoning.
 The expected time spent at a fixed node $v$ is therefore 
 $b^2\cdot \frac{|\qsY{i}{v}|}{K}\cdot (\log{n})^{\Oh(\log_b{n})}$. 
 By \cref{lem:subtree}, this sums up to $b^2(\log{n})^{\Oh(\log_b{n})}$.
 \end{proof}

 By a symmetric argument, we obtain the following result:
 \begin{lemma}\label{lem:delete}
  Upon any deletion in $X$, we can update the values $\Delta_{v,\ell}^{\le K}(X,X)$
  across all active nodes $v\in T$ and labels $\ell \in L_v$ in expected time $b^2(\log{n})^{\Oh(\log_b{n})}$.
\end{lemma}

\subsection{The full algorithm}\label{sec:edXedY:alg}

As mentioned earlier, we use \cref{lem:rebalance,lem:insert,lem:delete} to maintain a precision sampling tree $T$
and, for all active nodes $v\in T$, the values $\Delta_{v,s}^{\le 3K}(X,X)$ for every $s\in S_v$.
On top of that, we maintain the data structure of \cref{prp:ipm} so that, after every update,
we can use the procedure {\sc Dyn-Edit-Y} of \cref{sec:subXedY} that lets us compute $\Delta_{root,0}^{\le K}(X,Y)$
in $k \cdot b^2 \cdot (\log n)^{\Oh(\log_b n)}$ time. 
Altogether, this yields the following result:

\thmmain
\newpage

\bibliographystyle{alphaurl}
\bibliography{biblio}

\appendix
\section{Proof of \cref{cor:precision-sampling}}

\precisionsampling*
\begin{proof}
Consider a random variable $u\sim \Dst(\epsilon,\frac12\delta)$ following the distribution of \cref{lem:precision-sampling}
and define a value $\tau\in (0,1]$ such that $\Pr[u < \tau] \le \frac12\delta \le \Pr[u \le \tau]$.
We define the distribution $\hDst(\epsilon,\delta)$ so that $\max(u,\tau)\sim \hDst(\epsilon,\delta)$.

As for accuracy, we use the original recovery algorithm for $\Dst(\epsilon,\frac12\delta)$.
To analyze its success probability, consider $v_1,\ldots, v_n \sim \hDst(\epsilon,\delta)$ and let $u_1,\ldots,u_n\sim \Dst(\epsilon,\frac12\delta)$ be such that $v_i = \max(u_i,\tau)$. If, for each $i$, the value $\widetilde A_i$ is an $(\alpha, \beta \cdot u_i)$-approximation of $A_i$,
then the algorithm $((1+\epsilon) \cdot \alpha, \beta)$-approximates $\sum_i A_i$ with probability at least $1-\frac12\delta$.
For a given $i$, the assumption that the value $\widetilde A_i$ is an $(\alpha, \beta \cdot u_i)$-approximation of $A_i$
may fail only if $u_i \ne v_i$, i.e., $u_i < \tau$, which happens with probability at most $\frac12\delta$.
By the union bound, the overall failure probability is at most $(n+1)\cdot\frac12\delta \le n\delta$.

Next, consider $u \sim \Dst(\epsilon, \frac12\delta)$ and the event $E=E(u,4\delta^{-1})$ of \cref{lem:precision-sampling}.
Moreover, let $F = \{u \le \tau\}$. 
Note that $\Pr[F]+\Pr[E] \ge \frac12\delta+1-\frac14\delta \ge 1+\frac14\delta$, so $\Pr[F\cap E]\ge \frac14\delta$.
Thus, \[\tau^{-1} \le \Exp\left[u^{-1} | F\cap E\right] \le \frac{\Pr[E]}{\Pr[F\cap E]}\cdot \Exp\left[u^{-1} | E\right] \le \tfrac{4}{\delta}\cdot \Exp\left[u^{-1} | E\right].\]
At the same time, if $v = \max(u,\tau)\sim \hDst(\epsilon,\delta)$, then
\begin{multline*}\Exp\left[v^{-1}\right] \le \Pr[E]\cdot \Exp\left[v^{-1} | E\right] + (1-\Pr[E])\cdot \tau^{-1} \le \Exp\left[u^{-1} | E\right]+\tfrac{\delta}{4}\cdot \tfrac{4}{\delta}\cdot \Exp\left[u^{-1} | E\right]\\ 
= 2\cdot \Exp\left[u^{-1} | E\right] = \widetilde{O}(\epsilon^{-2} \log(2\delta^{-1}) \log(4\delta^{-1}))=\widetilde{O}(\epsilon^{-2} \log^2(\delta^{-1})).\qedhere
\end{multline*}
\end{proof}

\end{document}